\documentclass[11pt]{article} 
\usepackage[utf8]{inputenc}
\usepackage{geometry}  
\usepackage{fullpage}

\usepackage[english]{babel}
\usepackage{mathtools}

\usepackage{natbib}
\usepackage{comment}

\usepackage{subcaption}
\usepackage[utf8]{inputenc} 
\usepackage[T1]{fontenc}    
\usepackage{url}            
\usepackage{booktabs}       
\usepackage{amsfonts}       
\usepackage{amsthm}
\usepackage{nicefrac}       
\usepackage{microtype}      
\usepackage{xcolor}         
\usepackage{amsmath}
\usepackage{float}
\usepackage{enumitem}
\usepackage{amssymb}

\usepackage[linktocpage=true,
pagebackref=true,colorlinks,citecolor=blue,
bookmarks,bookmarksopen,bookmarksnumbered]
{hyperref}
\usepackage[noabbrev, nameinlink]{cleveref}
\crefname{property}{property}{Property}
\creflabelformat{property}{(#1)#2#3}
\crefname{equation}{eq}{Eq}
\creflabelformat{equation}{(#1)#2#3}
\crefname{algocf}{Algorithm}{Algorithm}
\creflabelformat{algocf}{#1#2#3}

\usepackage{multirow}
\usepackage{placeins}

\usepackage{amsthm}
\usepackage{amssymb,amsmath}
\usepackage{bbm}
\usepackage{xcolor}
\usepackage{microtype}
\usepackage{graphicx}
\usepackage{mdframed}
\usepackage{caption}
\usepackage{verbatim}

\usepackage{tcolorbox}
\tcbuselibrary{skins,breakable}
\tcbset{enhanced jigsaw}

\definecolor{ceruleanblue}{rgb}{0.16, 0.32, 0.75}
\definecolor{darkmidnightblue}{rgb}{0.0, 0.2, 0.4}
\definecolor{darkpastelgreen}{rgb}{0.01, 0.75, 0.24}
\definecolor{bleudefrance}{rgb}{0.19, 0.55, 0.91}
\definecolor{RURed}{rgb}{0.65, 0.2, 0.15}

\usepackage[ruled, vlined, linesnumbered]{algorithm2e}
\usepackage{algorithmicx}
\let\oldnl\nl
\newcommand{\nonl}{\renewcommand{\nl}{\let\nl\oldnl}} 
\newcommand{\algcomment}[1]{// \textcolor{gray}{#1}}

\usepackage{thm-restate}

\newtheorem{theorem}{Theorem}
\newtheorem{lemma}{Lemma}[section]
\newtheorem{proposition}[lemma]{Proposition}
\newtheorem{corollary}[theorem]{Corollary}
\newtheorem{claim}[lemma]{Claim}
\newtheorem{fact}[lemma]{Fact}

\newtheorem{observation}{Observation}

\newtheorem*{claim*}{Claim}
\newtheorem*{proposition*}{Proposition}
\newtheorem*{lemma*}{Lemma}
\newtheorem*{corollary*}{Corollary}
\newtheorem*{remark*}{Remark}

\theoremstyle{definition}
\newtheorem{definition}{Definition}

\newtheorem{problem}{Problem}
\newtheorem*{problem*}{Problem}
\newtheorem{remark}{Remark}

\usepackage{tocloft}

\newtheorem{mdresult}{Result}

\newtheorem{mdalg}{Algorithm}

\providecommand{\email}[1]{\href{mailto:#1}{\nolinkurl{#1}\xspace}}

\newcommand{\eps}{\ensuremath{\varepsilon}}

\newcommand{\bracket}[1]{\left[#1\right]}
\newcommand{\paren}[1]{\ensuremath{\left(#1\right)}\xspace}
\newcommand{\card}[1]{\left\vert{#1}\right\vert}

\newcommand{\cO}{\ensuremath{\mathcal{O}}\xspace}
\newcommand{\cT}{\ensuremath{\mathcal{T}}\xspace}
\newcommand{\cTstar}{\ensuremath{\mathcal{T}^{*}}\xspace}

\newcommand{\expect}[1]{\Exp\bracket{#1}}

\newcommand{\poly}{\mbox{\rm poly}}
\newcommand{\polylog}{\mbox{\rm  polylog}\,\xspace}

\newcommand{\OPT}{\ensuremath{\mbox{\sf OPT}}\xspace}

\DeclareMathOperator*{\Exp}{\ensuremath{{\mathbb{E}}}}
\DeclareMathOperator*{\Prob}{\ensuremath{\textnormal{Pr}}}
\renewcommand{\Pr}{\Prob}

\newenvironment{tbox}{\begin{tcolorbox}[
		enlarge top by=5pt,
		enlarge bottom by=5pt,
		 breakable,
		 boxsep=0pt,
                  left=4pt,
                  right=4pt,
                  top=10pt,
                  arc=0pt,
                  boxrule=1pt,toprule=1pt,
                  colback=white
                  ]
	}
{\end{tcolorbox}}

\newcommand{\distr}[2]{\ensuremath{\textnormal{dist}_{#1}(r,\, #2)}\xspace}

\newcommand{\Otilde}{\ensuremath{\tilde{O}}\xspace}

\SetKwComment{Comment}{// }{}

\newcommand{\cost}{\ensuremath{\textnormal{\textsf{cost}}}}
\newcommand{\rev}{\ensuremath{\textnormal{\textsf{rev}}}}

\newcommand{\textlca}{\textnormal{\texttt{LCA}}}

\newcommand{\lcatree}[3]{\ensuremath{\textlca_{#3}(#1, \,#2)}\xspace}

\newcommand{\lcatreeset}[2]{\ensuremath{\textlca_{#2}(#1)}\xspace}

\newcommand{\pTree}{\ensuremath{\mathcal{I}}\xspace}
\newcommand{\Vp}{\ensuremath{\widetilde{V}}\xspace}
\newcommand{\Vh}{\ensuremath{V}_{\text{H}}\xspace}
\newcommand{\Vinput}{\ensuremath{{V}_{\textnormal{input}}}\xspace}

\newcommand{\cTree}[1]{\ensuremath{\mathcal{T}_{#1}}\xspace}

\newcommand{\leaves}[2]{\ensuremath{\textnormal{\texttt{leaves}}_{#1}{[\,#2\,]}}\xspace}
\newcommand{\nonleaves}[2]{\ensuremath{\textnormal{\texttt{non-leaves}}_{#1}{[\,#2\,]}}\xspace}

\newcommand{\ustar}{\ensuremath{u^*}\xspace}

\newcommand{\ellstar}{\ensuremath{\ell^{*}}\xspace}

\newcommand{\myqed}[1]{\let\qed\relax \hspace*{\fill} #1 $\square$}

\newcommand{\levelr}[2]{\ensuremath{\texttt{level}_{#1}(#2)}\xspace}

\newcommand{\nh}{\ensuremath{n_{H}}\xspace}

\newcommand{\ptalg}[2]{\ensuremath{\texttt{weak-partial-tree}^{#2}(#1)}\xspace}
\newcommand{\testalg}[4]{\ensuremath{\texttt{counterpart-tester}^{(\cO, \Vh)}(#1,\, #2, \, #3, \, #4)}}
\newcommand{\splitalg}[2]{\ensuremath{\texttt{vertex-split}^{#2}(#1)}\xspace}
\newcommand{\testalgstrong}[5]{\ensuremath{\texttt{counterpart-tester-strong}^{(\cO, #5)}(#1,\, #2, \, #3, \, #4)}}
\newcommand{\treemerge}[3]{\ensuremath{\texttt{tree-merge}^{\cO}(#1, \,#2, \, #3)}\xspace}
\newcommand{\testalgprev}[3]{\ensuremath{\texttt{predecessor-tester}^{(\cO, \Vh)}(#1,\, #2, \, #3)}}

\newcommand{\testorphan}{\ensuremath{\texttt{test-orphan-predecessor}^{(\cO, \Vh, \Vp)}}}

\newcommand{\strongptalg}[2]{\ensuremath{\texttt{strong-partial-tree}^{#2}(#1)}\xspace}

\newcommand{\np}{\ensuremath{\tilde{n}}\xspace}
\newcommand{\Vsmall}{\ensuremath{V^{\textnormal{small}}}\xspace}
\newcommand{\Vorphan}{\ensuremath{V^{\textnormal{orphan}}}\xspace}
\newcommand{\Vbefore}[1]{\ensuremath{V^{\textnormal{small}}(\leq #1)}\xspace}
\newcommand{\Vafter}[1]{\ensuremath{V^{\textnormal{small}}(\geq -#1)}\xspace}

\newcommand{\OPTdas}{\ensuremath{\OPT^{\textnormal{Das}}}}
\newcommand{\OPTmw}{\ensuremath{\OPT^{\textnormal{MW}}}}

\newcommand{\cut}{\ensuremath{\textnormal{\textbf{cut}}}\xspace}

\def \calD    {\mdef{\mathcal{D}}}

\def \calF    {\mdef{\mathcal{F}}}
\def \calG    {\mdef{\mathcal{G}}}
\def \calH    {\mdef{\mathcal{H}}}

\def \calT    {\mdef{\mathcal{T}}}

\def \calX    {\mdef{\mathcal{X}}}

\newcommand{\EEx}[2]{\ensuremath{\underset{#1}{\mathbb{E}}\left[#2\right]}}

\newcommand{\mdef}[1]{{\ensuremath{#1}}\xspace}  

\DeclareMathOperator*{\argmin}{argmin}

\newcommand{\pa}[1]{\ensuremath{\textsf{pa}\left(#1\right)}\xspace}

\title{Learning-Augmented Hierarchical Clustering}
\usepackage{times}
\author{
Vladimir Braverman\thanks{Johns Hopkins University, Rice University, and Google Research. E-mail: \email{vova@cs.jhu.edu}}
\and
Jon C. Ergun\thanks{Carnegie Mellon University. E-mail: \email{ergunjon0@gmail.com}}
\and
Chen Wang\thanks{Rice University and Texas A\&M University. E-mail: \email{chen.wang.research@gmail.com}}
\and
Samson Zhou\thanks{Texas A\&M University. E-mail: \email{samsonzhou@gmail.com}}
}
\date{}

\begin{document}
\allowdisplaybreaks

\maketitle

\pagenumbering{roman}

\begin{abstract}
Hierarchical clustering (HC) is an important data analysis technique in which the goal is to recursively partition a dataset into a tree-like structure while grouping together similar data points at each level of granularity. Unfortunately, for many of the proposed HC objectives, there exist strong barriers to approximation algorithms with the hardness of approximation. Thus, we consider the problem of hierarchical clustering given auxiliary information from natural oracles. Specifically, we focus on a \emph{splitting oracle} which, when provided with a triplet of vertices $(u,v,w)$, answers (possibly erroneously) the pairs of vertices whose lowest common ancestor includes all three vertices in an optimal tree $\mathcal{T}^{*}$, i.e., identifying which vertex ``splits away'' from the others in $\mathcal{T}^{*}$. Using such an oracle, we obtain the following results:
\begin{itemize}[leftmargin=30pt, rightmargin=30pt]
\item A polynomial-time algorithm that outputs a hierarchical clustering tree with $O(1)$-approximation to the Dasgupta objective (Dasgupta [STOC'16]).
\item A near-linear time algorithm that outputs a hierarchical clustering tree with $(1-o(1))$-approximation to the Moseley-Wang objective (Moseley and Wang [NeurIPS'17]).
\end{itemize}
Under the plausible Small Set Expansion Hypothesis, no polynomial-time algorithm can achieve any constant approximation for Dasgupta's objective or $(1-C)$-approximation for the Moseley-Wang objective for some constant $C>0$. As such, our results demonstrate that the splitting oracle enables algorithms to outperform standard HC approaches and overcome hardness constraints. Furthermore, our approaches extend to sublinear settings, in which we show new streaming and PRAM algorithms for HC with improved guarantees.

At the heart of our techniques is to construct \emph{partial HC trees} with the splitting oracle. These are data structures that capture the approximate composition of the optimal HC trees and can be efficiently built regardless of the input graph and objectives. As such, we believe the partial HC trees can be applied to a broad range of HC problems and can be of independent interest.  
\end{abstract}
\thispagestyle{empty}

\newpage
\small
\setcounter{tocdepth}{2}
{\hypersetup{hidelinks}\tableofcontents}
\normalsize
\newpage

\pagenumbering{arabic}
\setcounter{page}{1}

\section{Introduction}
Hierarchical clustering (HC) is a popular data analysis technique that recursively partitions a dataset throughout a tree-like structure, so that similar data points are grouped together at different levels of granularity. 
Specifically, the input is a set of $n$ data points and a measure of similarity or dissimilarity between the points, which induces a weighted graph whose vertices represent the data points and whose edge weights represent the pairwise measure between the vertices. 
The output is a binary dendrogram, which is a rooted tree whose leaves represent the individual data points and whose internal nodes each represent a cluster of the data points in its subtree, thus providing a hierarchical representation of relationships within the dataset. 

Hierarchical clustering has several advantages over flat clusterings such as $k$-means or $k$-median, where the dataset is partitioned into a fixed number of clusters. 
For example, the ``correct'' number of clusters in flat clusterings is often a difficult question that is the focus of a sequence of works dating back to the 1950s~\citep{thorndike1953belongs}. 
In hierarchical clustering, there is no fixed number of clusters that needs to be determined in advance. 
Another advantage of hierarchical clustering is that the dendrogram simultaneously captures structure at all levels of granularity, whereas flat clustering does not identify further structure inside each of the clusters. 
Hence, hierarchical clustering arises in various applications where data exhibits hierarchical structure, such as biology and phylogenetics~\citep{sneath1973numerical,sotiriou2003breast}, image and text analysis~\citep{steinbach2000comparison}, and community detection~\citep{leskovec2020mining}. 

Despite a wealth of heuristics for both agglomerative bottom-up~\citep{ward1963hierarchical} and divisive top-down approaches~\citep{guenoche1991efficient}, formal mathematical understanding of hierarchical clustering often stagnated due to the absence of well-posed objectives until a relatively recent work by Dasgupta~\citep{Dasgupta16}. 
Subsequently, additional objectives~\citep{MoseleyW17,cohen2019hierarchical} were proposed to quantify the performance of dendrograms with $n$ leaves, so that high-revenue similarity trees and low-cost dissimilarity trees correspond to desirable hierarchical partitions of the dataset. 

For a number of these objectives, various algorithms have been proven to achieve specific approximation guarantees. 
For example, a divisive clustering algorithm based on the sparsest cut subroutine was shown to give an $O(\sqrt{\log n})$ approximation~\citep{charikar2017approximate,cohen2019hierarchical,Deng0U0Z25} for Dasgupta's objective~\citep{Dasgupta16}, while the long-used agglomerative heuristic was shown to give a $2$-approximation for the dissimilarity objective proposed by \citep{cohen2019hierarchical} and a $\frac{1}{3}$-approximation for the similarity objective proposed by \citep{MoseleyW17}, i.e., the Moseley-Wang (MW) objective. These objectives were further explored in the contexts of better approximation factors~\cite{ChatziafratisYL20,AlonAV20}, sublinear computation models~\cite{RajagopalanVVCP21,AssadiCLMW22,AgarwalKLP22}, and graphs with special properties~\cite{CharikarCNY19,ManghiucS21}.

Unfortunately, \citep{CharikarCN19} showed that average-linkage cannot do better than $\frac{3}{2}$-approximation for the former objective or better than $\frac{1}{3}$-approximation for the latter. 
More general hardness of approximation results were given, showing the impossibility of achieving roughly $1.003$-approximation~\citep{ChatziafratisGL20} for dissimilarity under the Unique Games Conjecture and the impossibility of achieving $(1-C)$-approximation for the Moseley-Wang objective~\citep{ChatziafratisYL20} under the Small Set Expansion hypothesis, for a fixed constant $C>0$. 
Moreover, for the Dasgupta objective, \citep{CharikarC17,RoyP17} showed that under the Small Set Expansion hypothesis, there is no constant approximation in polynomial time for \emph{any constant}.
Thus, we seek new practical approaches that enable better approximation guarantees without assumptions about the underlying dataset or weight function. 

\paragraph{Learning-augmented algorithms.}
We draw inspiration from the recent advances in the predictive capabilities of machine learning models. 
On one hand, datasets often have additional auxiliary information that can be used to improve algorithmic performance if accurate.  
For example, in many applications, the input dataset can retain insightful patterns exhibited by similar datasets generated from previous instances.  
On the other hand, machine learning models lack provable guarantees and can result in wildly inaccurate predictions when generalizing to unfamiliar inputs \citep{SzegedyZSBEGF13}. 
Nevertheless, \emph{learning-augmented algorithms}~\citep{MitzenmacherV20} that overcome worst-case computational limits have been designed for a number of applications, such as warm-starts for faster algorithms~\citep{DinitzILMV21,ChenSVZ22,DaviesMVW23}, data structures optimized for specific query distributions~\citep{KraskaBCDP18,Mitzenmacher18,LinLW22,FuSZ24}, online algorithms with some ``forecast'' of the future~\citep{PurohitSK18,GollapudiP19,LattanziLMV20,WangLW20,WeiZ20,BamasMS20,ImKQP21,LykourisV21,AamandCI22,Anand0KP22,AzarPT22,GrigorescuLSSZ22,KhodakBTV22,GPSSNeurips22,JiangLLTZ22,AntoniadisCEPS23,ShinLLA23,BenomarP23}, input-sensitive sketches for more space-efficient streaming algorithms~\citep{HsuIKV19,IndykVY19,JiangLLRW20,ChenIW22,ChenEILNRSWWZ22,LLLVW23}, and classical NP hard problems~\cite{BravermanDSW24,Cohen-Addadd0LP24,BravermanDJNWZZ25}.

For clustering problems, \citep{ErgunFSWZ22,NguyenCN23,CLRSZ24} introduced \emph{flat} clustering algorithms that use polynomial runtime and achieve approximation guarantees beyond NP hardness limits. 
Though their techniques are specific to $k$-means and $k$-median clustering, their work nevertheless serves as an important conceptual message that demonstrates machine learning oracles can be used to improve upon traditional techniques for cluster analysis. 
Furthermore, for graph-base problems, recent results by \cite{Cohen-Addadd0LP24,BravermanDSW24,Dong0V25} have shown that natural learning-augmented oracles could help overcome NP-hardness constraints as well.
We thus ask whether machine learning models can be used to provably improve (graph-based) \emph{hierarchical} clustering. 

\subsection{Our Contributions}
In this paper, we consider the problem of hierarchical clustering given a possibly erroneous oracle that uses auxiliary information, e.g., through clusterings of similar datasets, to provide local information about the relationship between queried data points. 
In particular, we consider a \emph{splitting oracle} that, on an input query of a triplet $(u,v,w)$ of vertices, outputs the vertex that is first separated away from the other two with respect to an optimal or near-optimal hierarchical clustering tree. 
In other words, if the oracle is consistent with some ground-truth tree, it will output the vertex that is \emph{not} in the same subtree as the other two vertices, under their least common ancestor. 
We remark that such oracle advice is natural due to the plethora of machine learning models that are trained on related instances of graphs, where the triplet relationships are already labeled. 

Using triplet split-away information is common in the literature, and there have been results explored in similar settings, e.g., tree reconstruction with \emph{accurate} triplet relationships given \cite{aho1981inferring}, triplets are given as constraints \cite{ChatziafratisNC18}, noisy triplet information with fresh randomness \cite{Emamjomeh-Zadeh18}, and algorithms with quartet information \cite{JiangKL00,SnirY11,AlonSY14}.
Furthermore, the reconstruction of phylogenetic CSPs, which provides an oracle with a similar form to ours, is an extensively studied line of work (see, e.g., \cite{ChatziafratisM23}). From the machine learning perspective, it is possible to learn such oracles in the PAC learning framework (see~\Cref{sec:oracle-and-learning-theory} for details).

In line with existing literature on learning-augmented algorithms, we investigate a stochastic and independently responding splitting oracle, where randomness is introduced only \emph{once}. 
Specifically, this implies that the oracle correctly responds with a probability of $p$ for some constant $p>1/2$, independently across vertex triplets. 
Additionally, repeated queries of the same triplet consistently yield the same (possibly erroneous) responses, which rules out basic boosting strategies such as repeatedly making the same query to the oracle. 
We remark this splitting oracle mirrors numerous machine learning models that are trainable with data yet exhibit inherent noise.

We now present our main results for learning-augmented hierarchical clustering.
We first show that for the Dasgupta objective, we can use such a splitting oracle to achieve a constant factor approximation in polynomial time. 

\begin{restatable}{theorem}{thmconstantpolytimehc}
\label{thm:constant-poly-time-HC}
There exists an algorithm that, given a weighted undirected graph $G=(V,E,w)$ and a splitting oracle $\cO$, with high probability, in polynomial time and $O(n^3)$ queries computes a hierarchical clustering tree $\cT$ such that $\cost_{G}(\cT)\leq O(1)\cdot \OPTdas(G)$, where $\OPTdas(G)$ is the cost of the optimal hierarchical clustering tree $\cTstar$, i.e., $\OPTdas(G)=\cost_{G}(\cTstar)$.
\end{restatable}

By comparison, \citep{CharikarC17,RoyP17} showed that there is no polynomial-time algorithm that could achieve any constant approximation to Dasgupta's objective, under the Small Set Expansion hypothesis. 
Hence, \Cref{thm:constant-poly-time-HC} illustrates that the power of a splitting oracle can be used to break complexity hardness limitations. 
On the other hand, we remark that the runtime of the algorithm of \Cref{thm:constant-poly-time-HC}, although polynomial, is perhaps embarrassingly large. 
We thus give an algorithm that uses the splitting oracle and $\tilde{O}(n^3)$ time to achieve approximation guarantees beyond the current state-of-the-art oblivious algorithms. 

\begin{restatable}{theorem}{thmrootloglogtimehc}
\label{thm:root-loglog-n4-time-HC}
There exists an algorithm that, given a weighted undirected graph $G=(V,E,w)$ and a splitting oracle $\cO$, with high probability, in $O(n^3 \log{n})$ time and $O(n^3)$ queries computes a hierarchical clustering tree $\cT$ such that $\cost_{G}(\cT)\leq O(\sqrt{\log\log{n}})\cdot \OPTdas(G)$, where $\OPTdas(G)$ is the cost of the optimal hierarchical clustering tree $\cTstar$, i.e., $\OPTdas(G)=\cost_{G}(\cTstar)$.
\end{restatable}
\noindent
By comparison, the best-known polynomial-time oblivious algorithm achieves $O(\sqrt{\log n})$-approximation \citep{charikar2017approximate,cohen2019hierarchical}. 
Thus our algorithm that achieves $O(\sqrt{\log\log n})$ approximation gives very competitive practical bounds -- the improvement of our algorithm is approximately $2.3$ times better for $n=10^{10}$. 

Turning our attention to the Moseley-Wang objective, we similarly show that a splitting oracle can also be used to achieve any constant factor approximation in polynomial time. 

\begin{restatable}{theorem}{thmhcalgmw}
\label{thm:hc-alg-MW}
There exists an algorithm that, given a weighted undirected graph $G=(V,E,w)$ and a splitting oracle $\cO$, with high probability, in $O(n^2 \cdot \polylog{n})$ time and $O(n^2)$ queries computes a hierarchical clustering tree $\cT$ such that $\rev_{G}(\cT)\geq (1-o(1))\cdot \OPTmw(G)$, where $\OPTmw(G)$ is the revenue of the optimal hierarchical clustering tree $\cTstar$, i.e., $\OPTmw(G)=\rev_{G}(\cTstar)$.
\end{restatable}

We note that \citep{ChatziafratisYL20} showed the APX-hardness of the $(1-C)$ approximation for Moseley-Wang objective under the Small Set Expansion hypothesis, for a fixed constant $C \in (0,1)$. As such, \Cref{thm:hc-alg-MW} again also shows the power of splitting oracles to overcome impossibility barriers. Since a $n$-vertex graph could have input size as large as $\Theta(n^2)$, the time complexity in \Cref{thm:hc-alg-MW} is near-linear in the worst case.

Finally, we observe that our algorithms possess favorable properties that are extremely amenable to sublinear algorithms. As such, we can obtain the following results in the streaming and parallel computation (PRAM) settings.
\begin{theorem}
In the single-pass graph streaming and the PRAM settings, there exists:
\begin{itemize}[leftmargin=15pt]
\item a single-pass (dynamic) streaming algorithm that, given a weighted undirected graph $G=(V,E,w)$ in a $\poly(n)$-length dynamic stream and an offline splitting oracle $\cO$, with high probability, uses $O(n\cdot \log^{3}{n})$ bits of space and polynomial time computes a hierarchical clustering tree $\cT$ such that $\cost_{G}(\cT)\leq O(1)\cdot \OPTdas(G)$ (\Cref{thm:hc-das-streaming}).
\item a PRAM algorithm that, given a weighted undirected graph $G=(V,E,w)$ and a splitting oracle $\cO$, with high probability, in $O(n^2 \cdot \polylog{n})$ work and $\log^3{n}$ depth computes a hierarchical clustering tree $\cT$ such that $\rev_{G}(\cT)\geq (1-o(1))\cdot \OPTmw(G)$ (\Cref{thm:hc-MW-PRAM}).
\end{itemize}
\end{theorem}

Our results in the sublinear settings similarly outperform the state-of-the-art in the HC algorithms without oracle advice. For instance, in the single-pass graph streaming setting, \citep{AssadiCLMW22,AgarwalKLP22} designed semi-streaming algorithms with $O(1)$ approximation but exponential time. 
By comparison, our algorithm only uses polynomial time, leveraging the advantage of the splitting oracle.

\section{Technical Overview}
\label{sec:tech-overview}
In this section, we give a high-level overview of our techniques. 
We also provide intuition on our algorithmic design choices, including a number of potential pitfalls, as well as a number of natural other approaches and why they do not work. 

\subsection{Why not simply follow the oracle (or other related strategies)?} 
At first glance, one might wonder whether the splitting oracle trivializes the problem. 
A natural question is whether it is possible to simply follow the oracle to recover the optimal tree $\cTstar$. 
Since the oracle only returns the relative information among a \emph{triplet} of vertices $(u,v,w)$, it is not immediately clear how to translate the answers from the oracle to a partition of vertices. 
After taking a closer look at the problem, we could observe issues with a handful of straightforward approaches. 

The first natural approach is to pretend the oracle is always correct and construct a tree from the ``splitting-away'' information between the triplets. 
Unfortunately, due to the error probability and adversarial answers, there may \emph{not} exist an underlying tree consistent with the answers to the queries. 
As such, it is unclear how the algorithm could produce a definitive answer. 

The second approach we could try is to frame the problem as a phylogenetic reconstruction problem, e.g., take all the ``splitting away'' for triplets as constraints, and try to construct an HC tree that satisfies as many constraints as possible. 
However, such an approach has two issues: $i).$ by a recent result of \cite{ChatziafratisM23}, the phylogenetic reconstruction problem is itself UG-hard; and $ii).$ the HC tree we constructed may prioritize a small number of \emph{wrong} answers from the oracle that happen to induce very large additive error. 

A more involved idea is to ``aggregate'' the oracle answers to construct the HC tree's partitions. To this end, an algorithm to determine the partition of a vertex $v$ is to fix $u$ in the smaller subtree and look into the number of vertices $t\in V$ that split away from $(u,v)$. More concretely, consider the split of the tree on the root $V \rightarrow (S_1, S_2)$, and suppose we know a vertex $u$ that is on the smaller subtree of the root partition (this is a big ``suppose'' as we will see later). Then for any vertex $v$ that is in the same subtree of $u$, we can get many vertices $t\in V$ from $\cO$ with the answer ``$t$ split away from $(u,v)$''. On the other hand, for a vertex $v$ that is in the opposite subtree of $u$, only a few vertices $t$ from $\cO$ would answer ``$t\in V$ split away from $(u,v)$''. The gap is large enough to apply concentration inequalities and find a separation between the cases. As such, we can recursively apply the above procedure and produce the optimal tree $\cTstar$.

Unfortunately, the above idea only works for the idealized case where we indeed know a vertex $u$ from the smaller part of the root partition. For the general case, the algorithm requires a surprising amount of new ideas and technical work. 
In particular, note that the aforementioned algorithm faces two major challenges: $(1).$ as the partition goes deep down the HC tree, the sizes of the subtrees become too small for high-probability guarantees; and $(2).$ it is not clear how to find a ``good'' vertex $u$ that induces the root cut. 
To elaborate on challenge $(1)$, note that when the subtree induces $o(\log{n})$ leaves, it is generally not possible for us to guarantee correctness for the subsequent partitions. 
As such, we must handle some form of ``ambiguity'' when dealing with subtrees induced on vertex sets with smaller sizes. 
Our approach to this challenge is to forgo the guarantees inside each leaf with $o(\log{n})$ vertices and work with the respective objective functions to show that the additive error is tolerable. 
In particular, we use the notion of \emph{partial} hierarchical clustering trees that approximately capture the structure of the optimal HC tree $\cTstar$ until the size of the induced vertex set becomes too small. 
In particular, we require specific structural properties of the costs of HC trees under Dasgupta's and the Moseley-Wang objectives. 
We provide more details about partial HC trees and how to use them to overcome challenge $(1)$ in \Cref{subsec:partial-tree-and-HC}.

Challenge $(2)$ is even trickier and requires more care. 
Observe that in the example of root cut $V \rightarrow (S_1, S_2)$, if $u$ is in the \emph{bigger} side of the partition, the argument may \emph{not} work. 
However, since we only have access to the \emph{triplet split} information, retrieving whether a vertex $u$ is on the smaller side of a particular tree split seems to be too much to ask. 
In particular, consider an example that the optimal tree $\cTstar$ first makes two splits of small subtrees of size $n^{0.99}$, as illustrated in \Cref{fig:hard-example-split}. 
Here, if we use $u_2$ as the ``baseline'' vertex to perform the split, we can still get a valid partition. 
However, the structure of the obtained tree is very different from $\cTstar$, and the additive error could be huge. 
Furthermore, since the actual tree $\cTstar$ is hidden from us, it is not immediately clear how could we distinguish a partition obtained by using $u_1$ vs. $u_2$. 
The problem becomes even more intriguing when we want to obtain near-linear time efficiency. 
We will discuss the intuition and techniques to handle challenge $(2)$ in \Cref{subsec:partial-tree-construction}.

\begin{figure}[!htb]
	\centering
	\includegraphics[scale=0.55]{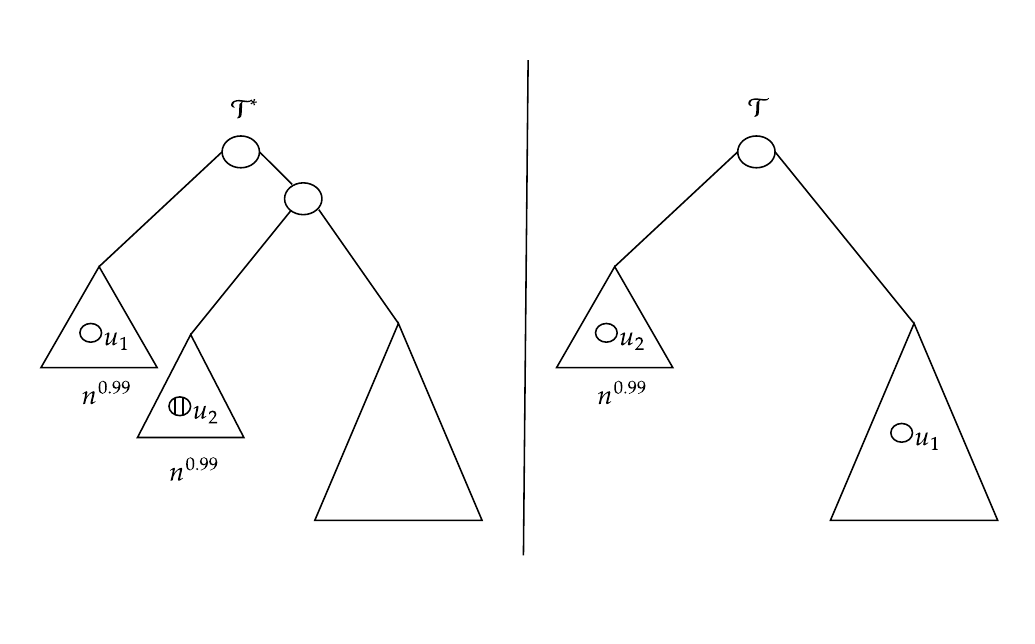}
	\caption{\label{fig:hard-example-split} An illustration of the hard example that the straightforward majority voting does not work. Left: the optimal HC tree $\cTstar$; Right: the outcome of the tree if we use $u_2$ as the baseline to perform the partition we discussed. }
\end{figure}

\subsection{Partial hierarchical clustering trees and HC}
\label{subsec:partial-tree-and-HC}
We now delve into more details about the definition of partial trees and how we use them to obtain low additive errors in Dasgupta's and the Moseley-Wang objectives.

\paragraph{The definition of partial hierarchical clustering trees.}
As discussed, a central element of our techniques is the notion of partial hierarchical clustering (HC) trees. 
The generic definition of the partial HC tree is similar to the normal HC tree, with the root representing the entire vertex set and the internal nodes representing the subsets of vertices. 
However, on the \emph{leaf} level, we allow the leaves of partial HC trees to contain \emph{multiple vertices}, up to $O(\log{n})$ many vertices, and contract them into a single leaf, which we call ``super-vertices''. 
The partial HC tree is then allowed to be oblivious of the clustering \emph{inside} each leaf node.

A partial HC tree is only useful if it can somehow capture the optimal tree $\cTstar$.
To this end, we introduce the \emph{strongly consistent} and \emph{weakly consistent} partial HC trees. 
Roughly speaking, a partial HC tree $\pTree$ is said to be strongly consistent with the optimal tree $\cTstar$ if it follows \emph{every partition} of the optimal tree in a top-down manner until the induced size of vertices is of size less than $O(\log{n})$, in which case we simply collapse the leaf into a super-vertex. 
Note that in the strongly consistent partial HC trees, every super-vertex induces a maximal tree in $\cTstar$ -- here, a maximal tree means a tree where its induced vertices are exactly the leaves of their lowest common ancestor. 
By comparison, the \emph{weakly consistent} partial HC tree allows vertices that do not necessarily form maximal trees to collapse into a single super-vertex. 
For instance, if there are $O(\log{n})$ vertices in \emph{multiple} subtrees in $\cTstar$, and suppose the LCAs of these subtrees are close to the root and form a consecutive segment in $\cTstar$, the weakly consistent partial HC tree can still collapse \emph{all} of the vertices into a single super-vertex. 
We provide illustrations of the weakly and strongly consistent partial trees in \Cref{fig:strong-consistent} and \Cref{fig:weak-consistent}, and we defer their formal definition to \Cref{def:strongly-consistent-tree,def:weakly-consistent-tree} in \Cref{sec:partial-tree}. 

\begin{figure}[!htb]
	\centering
	\includegraphics[scale=0.5]{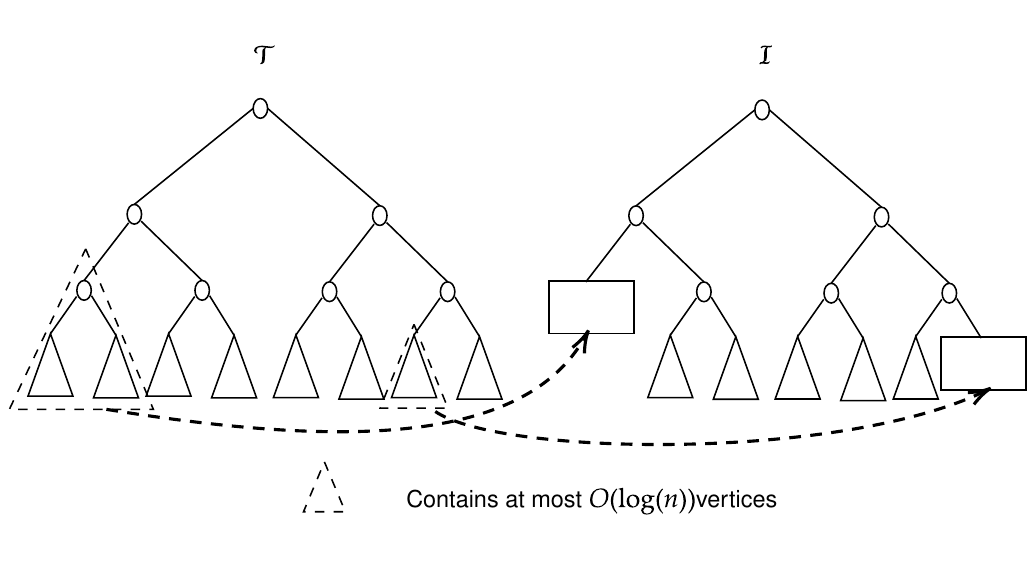}
	\caption{\label{fig:strong-consistent} An illustration of the strongly consistent partial HC trees as defined in \Cref{def:strongly-consistent-tree}. The boxes indicate super-vertices whose clustering is \emph{unknown} in the partial HC tree.}
\end{figure}

\begin{figure}[!htb]
	\centering
	\includegraphics[scale=0.5]{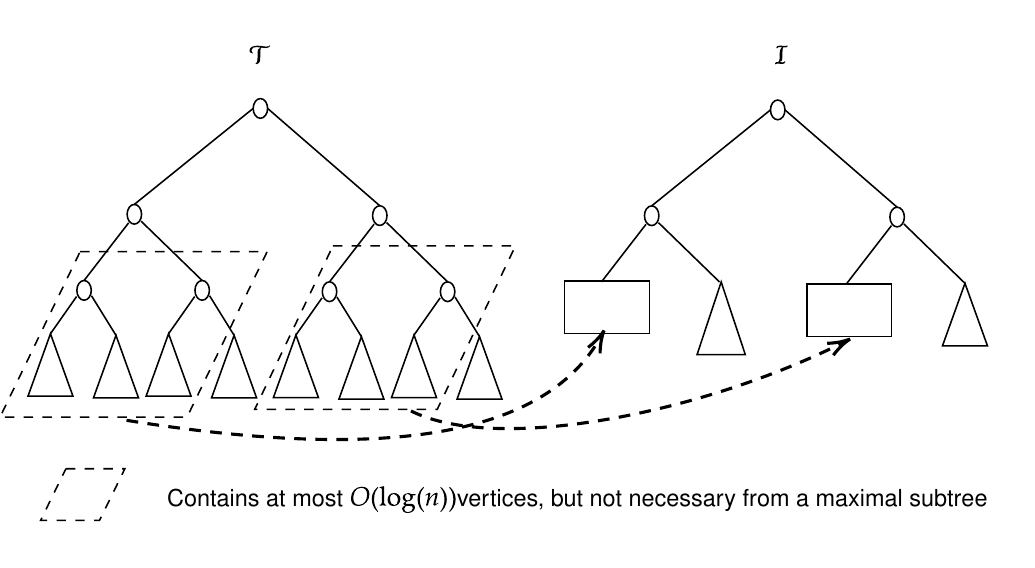}
	\caption{\label{fig:weak-consistent} An illustration of the weakly consistent partial HC trees as defined in \Cref{def:weakly-consistent-tree}. The boxes indicate super-vertices whose clustering is \emph{unknown} in the partial HC tree.}
\end{figure}

We will eventually show that, given the splitting oracle, we can efficiently construct both the strongly and the weakly consistent partial trees \emph{regardless} of the input graph and the objective functions. 
However, for now, we first discuss \emph{why} partial HC trees are useful for our HC objectives.

\paragraph{Using partial HC trees for hierarchical clustering.}
Intuitively, if we can obtain a partial HC tree that is consistent with the optimal tree $\cTstar$, we can build an actual HC tree by ``fixing'' the clustering of the super-vertices locally to obtain a good approximation. 
Indeed, we note that if a partial HC tree is \emph{strongly} consistent with $\cTstar$, we can straightforwardly obtain approximation algorithms for both Dasgupta's and the Moseley-Wang objectives. For Dasgupta's objective, we can run the optimal or approximate \emph{recursive sparsest cuts} for the subgraphs induced by the super-vertices. 
Note that since the subgraphs are only of size $O(\log{n})$, we can even afford to find the exact optimal recursive sparsest cuts in \emph{polynomial time}. 
The case for the Moseley-Wang objective is even easier: since the number of leaves outside each super-vertex is at least $n-O(\log{n})$, any arbitrary partition of the super-vertices can still give us a $(1-o(1))$ approximation.

Unfortunately, as we will see shortly, the strongly consistent partial tree can only be implemented in $\Otilde(n^3)$ time -- a tolerable yet far-from-optimal efficiency. 
As such, for the Moseley-Wang objective, we further investigate the algorithm that only uses the \emph{weakly consistent} partial HC trees, which we can build in near-linear time. 
In this case, for two vertices that are in the same super-vertex, we can replicate the argument for the strongly consistent partial HC trees again to get a $(1-o(1))$ approximation. 
However, challenges arise when the two vertices $(u,v)$ are in different super-vertices of the partial HC tree. 
Here, the number of induced leaves can differ by an $O(\log{n})$ additive factor, but the number of non-leaves induced by $(u,v)$ can be very small in $\cTstar$ (say, $o(\log{n})$). As such, the $O(\log{n})$ difference on the \emph{size} of non-leaves might lead to \emph{infinity multiplicative gap} in the revenue, which makes controlling the overall approximation factor hard. 
To tackle this issue, we prove some new structural results for HC trees under the Moseley-Wang objective: we show for all the edges $(u,v)$ that only induce a very small number of non-leaves \emph{and} are ``far away'' from each other in the optimal tree, the contribution of such edges to the optimal objective can only be an $o(1)$ fraction. 
As such, we can simply ignore the approximation guarantees on these edges and obtain an $(1-o(1))$ approximation of the Moseley-Wang objective.

\subsection{The construction of the partial HC trees}
\label{subsec:partial-tree-construction}

We now come back to the efficient construction of the partial HC trees that are strongly and weakly consistent with the optimal HC tree. 
For simplicity, we slightly abuse the notation to let $V$ always denote the vertex set in the high-level discussion, even if we are talking about a subset of vertices\footnote{In the formal analysis of \Cref{sec:strong-partial-tree,sec:weak-partial-tree}, we use $\Vp$ as the set of vertices of the current recursion level.}.

\subsubsection{Strongly consistent partial HC trees}
We first discuss the case for the strongly consistent HC tree, which only requires top-down splits.
As we have discussed before, for any fixed partition, since we do \emph{not} have any information about which side a vertex $u$ is on, it is generally very hard for us to know which obtained partition is actually consistent with the split in $\cT^*$.
To address this challenge, introduce the \emph{small-tree splitting order}, which, roughly speaking, is a thought process that recursively draws the smaller side of the subtree in $\cTstar$. 
For instance, in the tree $\cTstar$ prescribed in \Cref{fig:hard-example-split}, the first $n^{0.99}$ vertices form the first small tree $\Vsmall_{1}$, the second $n^{0.99}$ vertices form the second small tree $\Vsmall_{2}$, and so on.

We shall show that if $u$ is among the first few small trees in the small-tree splitting order, we can recover the set of vertices as the \emph{sibling} of the small tree\footnote{Since ``sibling'' is a generic word, we call this set ``counterpart'' in our formal description in \Cref{sec:strong-partial-tree,sec:weak-partial-tree} to avoid confusion.}. 
For example, if we select $u_2$ in \Cref{fig:hard-example-split}, we can recover the subtree on the right but not the subtree that contains $u_1$. 
Our strategy is as follows: for a fixed vertex $u$ and a vertex $v$ whose split is to be determined, in addition to testing how many vertices $t\in V$ such that $t$ splits away from $(u,v)$, we also test the number of vertices $t\in V$ such that \emph{$v$ splits away from $(u,t)$}. 
To see why this additional test helps, let us again look at the example in \Cref{fig:hard-example-split}. With the additional subroutine, if we use $u_2$ as the fixed vertex, the vertices in $\Vsmall_{1}$ will split away from many $(u_2,t)$ pairs. 
On the other hand, for every vertex $v$ on the sibling subtree of $\Vsmall_{2}$, it splits away from $(u, t)$ only if $t\in \Vsmall_{2}$, which creates a clear signal. 
By careful handling of cases, we could argue that the algorithm works for general cases as long as $u$ belongs to an ``early enough'' small tree.

The above strategy provides a new way to identify a ``good'' $u$: it suffices to only look at the \emph{size} of the set of vertices we recover. 
In particular, if $u$ is among the \emph{root} cut, it surely induces the largest size on the set of the recovered vertices. 
As such, a simple exhaustive search can find such a vertex $u$ and the corresponding set $T$. 
Since there are $n$ vertices to be tested, and each test requires $O(n^2)$ time, the total time for each partition is at most $O(n^3)$. 
We can then recursively run this procedure, which will lead to a partial tree that is \emph{strongly consistent} in $O(n^4)$ time.
Furthermore, using a simple sampling trick, we could reduce the time for each test to $O(n\log{n})$ time and queries, which brings the total number of time and queries to $\Otilde(n^3)$.
Furthermore, since we only need to maintain counters for each vertex, the entire algorithm can be implemented in $O(n \log{n})$ space.

\subsubsection{Weakly consistent partial HC trees}
The exhaustive search subroutine in the above idea inevitably leads to $\tilde{\Theta}(n^3)$ time and queries on the splitting oracle. 
This gives us a new, and perhaps more intriguing challenge: if we only want to get partial trees that are \emph{weakly consistent} with the optimal tree, can we improve the efficiency? 
Note that if we target a near-linear running time, we cannot always hope to get a $u$ from the smaller side of the \emph{root} partition. 
For a concrete example, let us look at the tree $\cTstar$ in \Cref{fig:hard-example-split} again. 
Here, before we ``hit'' a vertex in the first small tree of size $n^{0.99}$, we will \emph{not} be able to produce a root cut. 
However, by the size of the first small tree, we will need to test at least $n^{0.01}$ vertices if we sample vertices uniformly at random. 
The overhead can be further enlarged: suppose the root cut splits the vertices into $n-1$ vertices and a single vertex, and suppose this process continues for $n^{0.99}$ levels; then, it is entirely unclear how to avoid the $n^{0.99}$ overhead.

\paragraph{The vertical split idea.} The above hard instance inspires us to resort to ``vertical'' splits of the tree -- that is, instead of finding a vertex $u$ on the split of the root, we use a vertex $u$ that is ``sufficiently early'' in the small-tree split order. 
To elaborate, we can efficiently find a vertex $u$ that is among the \emph{union} of smaller subtrees that collectively induced at least ${n}/\polylog{n}$ leaves. 
In this way, we can still recover a \emph{maximal subtree} whose induced leaves $T$ is of size at least $(n-n/\polylog{n})$ -- a size reduction that is significant enough for the entire algorithm to converge in $\polylog{n}$ iterations. 
Finally, our algorithm will guarantee that $V$ is a composable set from a \emph{single} maximal tree, which implies that $T$ and $V\setminus T$ are composable sets, which allow us to recurse on both sides.

\paragraph{The use of ``horizon sets''.} 
There is yet another subtle issue in the above idea: we have to ensure both parts of the split always maintain the \emph{weak consistency} property. 
Specifically, it is crucial to maintain super-vertices with an out-degree of at most $2$, where each is linked to at most one parent node in $\cTstar$ and one sibling node in $\cTstar$. 
Within our vertical split concept, as $T$ invariably forms a single composable set, achieving this is straightforward. 
However, in the residual part of the split algorithm—here, $V \setminus T$—sustaining weak consistency becomes notably more complex. 
There can be two cases for such a guarantee to hold: either $a).$ $V\setminus T$ itself is a single composable set, which happens when $V$ is a split on the \emph{root} vertex, or $b).$ $V\setminus T$ has some ``orphaned'' vertices -- the sibling vertices of the subtree induced by $T$ in $V$ (see \Cref{def:orphan-vertex} for the formal definition).

Our approach to handle both the $a)$ and $b)$ cases is to use a semi-invariant \emph{horizon set} $\Vh \supseteq V$. 
The idea here is that instead of finding $T$ on $V$, we find it on $\Vh$, which is roughly defined as the set of vertices for us to find the partition $T$ on \emph{before} a split on the root node. 
In particular, suppose the set of orphaned vertices is of relatively small size in $V$. 
We can always find a vertex $u\in V$ such that $u$ is split \emph{earlier} than the orphaned vertex set in the small-tree split order of $\Vh$. 
Therefore, we can make sure that the set $T$ to be found in the new iteration will include the orphaned vertices in $V$. 
In this way, we always keep at most \emph{one} edge connecting to the sibling of the orphaned vertices in the \emph{current iteration}. 
An illustration of the role of the horizon set can be shown as \Cref{fig:horizon-sets}.

\begin{figure}[!htb]
	\centering
	\includegraphics[scale=0.75]{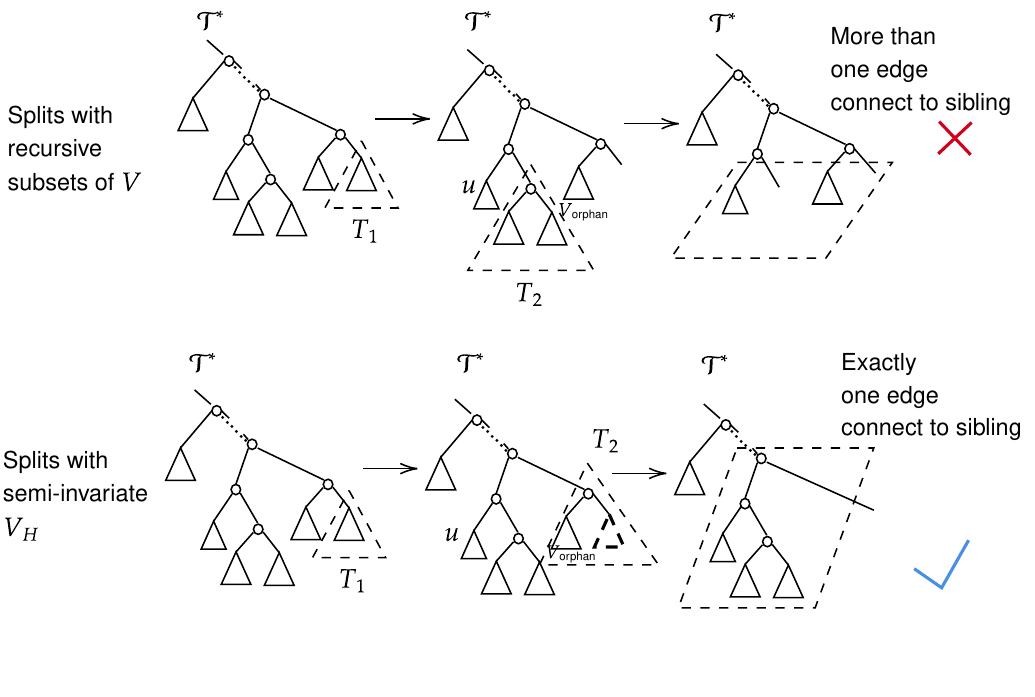}
	\caption{\label{fig:horizon-sets} An illustration of how horizon sets help maintain weak consistency. With the same $u$, the reason for the different outcomes is that on the top, the small-tree split order \emph{restricted on} $V$ has changed compared to the last iteration. The use of the horizon set helps keep the small-tree split order \emph{invariant} before a root cut happens.}
\end{figure}

\paragraph{The candidate vertex and root test.} 
The above analysis assumed the size of orphaned vertices is relatively small in $V$. However, what if the size of orphaned vertices becomes large in $V$? 
In this scenario, one of two sub-cases must happen: either we run into a root split, or we have a large set of vertices which is not on the root of $\Vh$, but occupies a large fraction among the remaining $V$. 
(Note that in case of a root split, the entire set of $V\setminus T$ is an orphaned set.) 
The challenge here is that we should update $\Vh$ in the former case while keep using the same $\Vh$ in the latter case, which requires us to distinguish the cases. 

To this end, we employ a new idea to select a ``candidate vertex'' that splits away from the orphaned set to address this challenge. 
Concretely, since the set of orphaned vertices is sufficiently large, when doing random sampling, we can get a vertex $u'$ from the orphaned set. 
Then, we use $\Vh$ to test whether there exist vertices that split away from $(u',t)$ for sufficiently many $t\in \Vh$. 
The idea here is that if $u'$ is in the orphaned set, and there still exists a vertex $u$ that splits earlier than $u'$ in (the small-tree split order of) $\Vh$, then $u$ should split away from many $(u',t)$ pairs. 
On the other hand, if $u'$ is on the smaller side of the root cut of $\Vh$, there is only a small number of $t$ that any $u\in V$ can split away from. 
As such, we can make progress by either identifying the ``right'' candidate vertex that splits earlier than $u'$, or by switching the horizon $\Vh$ and recurse on the root cut. 

\paragraph{Merging of two weakly consistent partial trees.} 
In the case of strongly consistent partial trees, the merging of subtrees is very straightforward: since the splits always follow the top-down order of internal nodes, we can easily merge the two subtrees with a common parent node. 
In the case of weakly consistent partial trees, the story is much more complicated. 
For the merge to be correct, we have to correctly identify the ``orphaned'' subtree in the previous recursion; however, since we only have access to the splitting oracle, and the actual optimal tree structure is hidden from us, it is not immediately clear how could we identify the ``correct'' internal node to merge the trees.

Fortunately, we could utilize the ``good vertex'' from the previous recursion to identify the ``orphaned'' set of vertices $\Vorphan$. 
In particular, let $u$ be the ``good vertex'' we used to split the tree; since $u$ also belongs to $\Vorphan$, for each vertex $v$, we can test how many times $v$ \emph{splits away} from $(u,t)$ for $t \in T$, where $T$ is the single maximal tree to be merged in the level of recursion. 
If $v \in \Vorphan$, such a vertex should not split away from $(u,t)$; otherwise, if $v \not\in \Vorphan$, $v$ should split away from $(u,t)$. 
Since the size of $T$ is large enough, we could identify $\Vorphan$ correctly with high probability, and perform the merge correctly. 
An illustration for this idea to identify $\Vorphan$ is in \Cref{fig:fast-merge}.

\begin{figure}[!htb]
	\centering
	\includegraphics[scale=0.8]{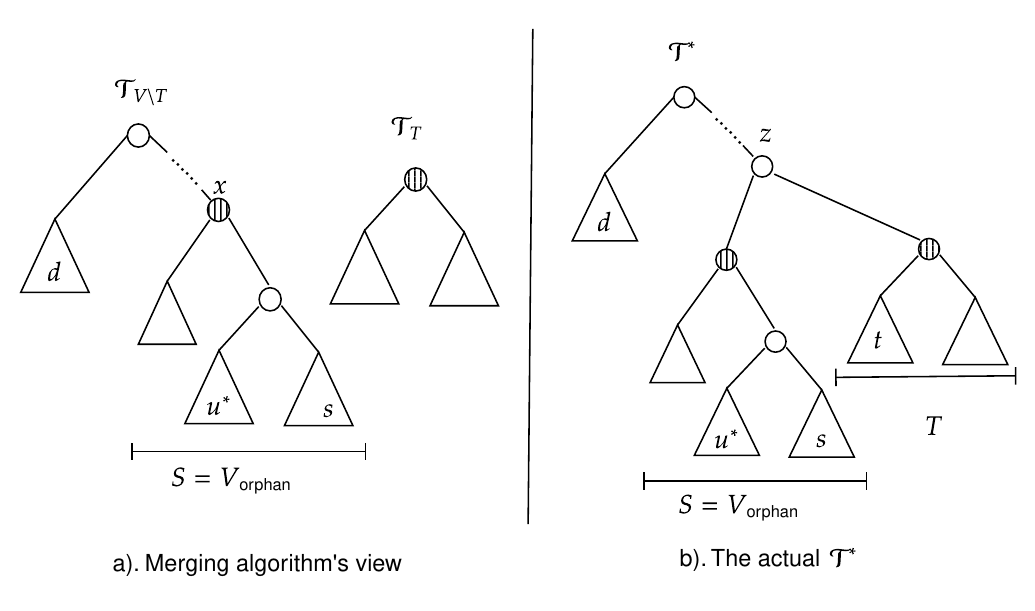}
	\caption{\label{fig:fast-merge} An illustration of the algorithm that merges $\cTstar(V\setminus T)$ and $\cTstar(T)$ The shaded internal node is the ``actual'' node to split $S$ and $T$ in $\cTstar$. If the oracle is correct, $s$ does \emph{not} split away from $(u^*, t)$, but $d$ \emph{does} split away from $(u^*, t)$. The size of $T$ is large enough to overcome the adversarial noise.}
\end{figure}

\paragraph{The complexity of the algorithm for weakly consistent trees.} Similar to the case for the strongly consistent trees, the subroutine that gives the sibling of a small tree for a fixed vertex $u$ takes $O(n^2)$ time. However, in the new algorithm, we only need to sample and test $O(\log{n})$ vertices $u$ for each iteration. Similarly, the root test and the tree merging subroutines both use $O(\log{n})$ vertices and $\tilde{O}(n^2)$ time. Furthermore, since we can roughly reduce the instance size by a $(1-1/\polylog{n})$ factor every iteration, the entire process converges in $\polylog{n}$ iterations. As such, we could argue that the total running time is $\tilde{O}(n^2)$, and the longest chain of dependent calls is $\polylog{n}$. This would imply a near-linear time offline algorithm and a parallel algorithm with near-linear work and poly-logarithmic depth.


\section{Preliminaries}
\label{sec:prelim}
We present the definition of the hierarchical clustering problem, our splitting oracle model, and the HC objectives in this section.

\subsection{The hierarchical clustering problem} 
We consider the hierarchical clustering problem with a splitting oracle $\cO$. The hierarchical clustering problem is defined as follows. We are given an $n$-vertex weighted undirected input graph $G=(V,E,w)$, and our goal is to produce a binary tree $\cT$ whose root node corresponds to the vertex set $V$ and the leaves represent the singleton vertices. The vertices set contained in the internal nodes form a Laminar set family: suppose node $x$ has children $(y,z)$, it represents a split $S_x \rightarrow (S_y, S_z)$, where $S_x = S_{y}\cup S_{z}$. In this work, we assume without loss of generality $n\geq 200 \log{n}$ -- the bound holds for any $n\geq 2500$, and if $n$ is a constant we can simply use a brute-force algorithm.

Hierarchical clustering trees only define a data structure, and there are many ways to construct ``valid'' HC trees. What eventually matters is to construct ``good'' HC trees -- a notion that does not have a universal way to define. Popular approaches include heuristics, which work well subjectively but lack formal guarantees, and objective functions, which provide rigorous frameworks to study the \emph{optimal} trees and the \emph{approximation algorithms}. In recent years, the latter approach has attracted considerable attention with popular objective functions by \citep{Dasgupta16,MoseleyW17,cohen2019hierarchical}. 

\paragraph{Notation.} For each internal node $x$ in $\cT$, we use $\leaves{\cT}{x}$ to denote the leaves in the induced subtree of $x$. Each internal node of an HC tree can be described by \emph{lowest common ancestor} (LCA) of vertices. For two vertices $(u,v)$ on the leaves of $\cT$, we use $\lcatree{u}{v}{\cT}$ to denote the node that is the lowest common ancestor of $u$ and $v$. We can further generalize this notion to a set of vertices, i.e., for a set $X\subseteq V$, the node $\lcatreeset{X}{\cT}$ refers to the lowest common ancestor of \emph{all} vertices in $X$. For a set $X$, we call the induced subtree $\cTree{X}$ a \emph{maximal subtree} of $\cT$ if $\leaves{\cT}{\lcatreeset{X}{\cTree{X}}}=X$, i.e., the lowest common ancestor of $X$ in $\cTree{X}$ induces all leaves of $X$ in $\cT$.

Let $r$ be the root of a hierarchical clustering tree $\cT$, and for any internal node $x$, we let $\distr{\cT}{x}$ be the number of edges on the shortest path between $r$ and $x$. 
We say node $x$ on level $\levelr{\cT}{x}$ is a \emph{higher level} node than node $y$ with level $\levelr{\cT}{y}$ in $\cT$ if $\distr{\cT}{x}>\distr{\cT}{y}$.
For two internal nodes $x$ and $y$ in $\cT$, we use $x=\pa{y}$ to denote the relationship of $x$ being the \emph{parent node} of $y$. Note that if $x=\pa{y}$, it automatically implies that $\levelr{\cT}{y} = \levelr{\cT}{x} + 1$.

\paragraph{The split-away vertex.} Note that in any hierarchical clustering tree, if we look at a triplet of vertices $(u,v,w)$, there must exist a vertex that \emph{split away} from the two others in the optimal tree $\cTstar$, i.e., two vertices with a LCA that is same as the LCA of \emph{all} three vertices in $\cTstar$. Formally, for a triplet of vertices $(u,v,w)$, we define ``$w$ splits away from $(u,v)$'' as follows.
\begin{definition}[Split-away vertex]
\label{def:split-away}
Let $G=(V,E,w)$ be a $n$-vertex graph, and let $\cTstar$ be the optimal HC tree for $G$. Given a triplet of vertices $(u,v,w)$, we say $w$ \emph{splits away from} $(u,v)$ (in $\cTstar$) if $\lcatree{w}{u}{\cTstar}$ (resp.  $\lcatree{w}{v}{\cTstar}$) is \emph{equal to} $\lcatreeset{\{u,v,w\}}{\cTstar}$.
\end{definition}

\subsection{The splitting oracle model}
We study the hierarchical clustering problem with a natural oracle advice model. In particular, we assume an oracle $\cO: V\times V\times V \rightarrow V$ that takes a triplet of vertices $(u,v,w)$, probabilistically correctly returns the vertex that ``split away'' from the other two vertices in the \emph{optimal tree}\footnote{We provide a discussion for splitting oracles with an \emph{approximately} optimal HC tree in \Cref{subsec:split-oracle-approx-tree}.}. The formal definition is given as follows.

\begin{definition}[The splitting oracle for hierarchical clustering]
\label{def:orl-model}
Let $G=(V,E)$ be a $n$-vertex graph, and let $\cTstar$ be the optimal hierarchical clustering tree of $G$. The oracle $\cO: V\times V\times V \rightarrow V$ is a function that upon being queried with a triplet of vertices $(u,v,w)$, responds as follows
\begin{itemize}
\item with probability $p$, the correct answer on which vertex splits away from the two others in $\cTstar$.
\item with probability $(1-p)$, an arbitrary (adversarial) answer on which vertex splits away from the two other vertices. 
\end{itemize}
The randomness is taken independently over all the queries and is \emph{fixed} across different queries on the same triplet. We assume each query to the oracle takes $O(1)$ time.
\end{definition}

Assuming the correct probability of an oracle is some constant $p>1/2$ is very common in the literature, especially for graph problem~\cite{BravermanDSW24,Cohen-Addadd0LP24,Dong0V25}.
For the convenience of presentation, we assume $p=\frac{9}{10}$ in this paper, and we provide a discussion about general success probabilities in \Cref{subsec:general-success-prob}. Observe that by the fixed randomness for each triplet, there are at most $\binom{n}{3}$ many answers that $\cO$ can have. This setting rules out trivial algorithms that simply get the correct by querying multiple times and boosting the success probability.

Several hierarchical clustering objectives are proved to be hard to approximate in polynomial time under very plausible complexity assumptions. As such, our goal is to explore whether we can obtain better approximation guarantees with the splitting oracle. 


\subsection{Objective functions for hierarchical clustering}
\label{sec:objective-functions}
We introduce the objective functions for hierarchical clustering we are going to discuss in this paper. These include the Dasgupta \emph{minimization} (cost) objective \citep{Dasgupta16} and Moseley-Wang \emph{maximization} (revenue) objective \citep{MoseleyW17}. We start with the minimization objective as prescribed by \citep{Dasgupta16}.

\begin{problem}[HC under Dasgupta's cost function]\label{prob:HC-das}
Given an $n$-vertex weighted graph $G=(V,E,w)$ with vertices corresponding to data points and edges measuring their similarity, create a rooted tree $\cT$ whose leaf nodes are $V$. The goal is to \emph{minimize} the cost of this tree $\cT$ defined as 
    \begin{align}
    \label{eq:hc-cost-das}
       \cost_G(\cT) := \sum_{e=(u,v) \in E} w(e) \cdot |\leaves{\cT}{\lcatree{u}{v}{\cT}}|,
    \end{align}
    where $|\leaves{\cT}{\lcatree{u}{v}{\cT}}|$ is the number of leaf-nodes in the sub-tree of $\cT$ rooted at the lowest common ancestor of $u$ and $v$. We use $\OPTdas(G)$ to denote the cost of an optimal HC tree under Dasgupta's cost for the graph $G$. 
\end{problem}

Roughly speaking, Dasgupta's objective accumulates the cost on an edge $(u,v)$ by the number of leaves \emph{inside} the subtree where $u$ and $v$ are first split. In contrast, the Moseley-Wang objective focuses on the dual of Dastupta's objective: it gathers the revenue on an edge $(u,v)$ by the number of leaves \emph{outside} the subtree where $u$ and $v$ are first split. Formally, the Moseley-Wang objective can be given as follows.

\begin{problem}[HC under Moseley-Wang revenue function]\label{prob:HC-MW}
Given an $n$-vertex weighted graph $G=(V,E,w)$ with vertices corresponding to data points and edges measuring their similarity, create a rooted tree $\cT$ whose leaf nodes are $V$. The goal is to \emph{maximize} the revenue of this tree $\cT$ defined as 
    \begin{align}
    \label{eq:hc-cost-MW}
    \begin{split}
       \rev_G(\cT) & := \sum_{e=(u,v) \in E} w(e) \cdot (n - |\leaves{\cT}{\lcatree{u}{v}{\cT}}|) \\
       & = \sum_{e=(u,v) \in E} w(e) \cdot |\nonleaves{\cT}{\lcatree{u}{v}{\cT}}|,
    \end{split}
    \end{align}
    where $|\leaves{\cT}{\lcatree{u}{v}{\cT}}|$ is the number of leaf-nodes in the sub-tree of $\cT$ rooted at the lowest common ancestor of $u$ and $v$, and $|\nonleaves{\cT}{\lcatree{u}{v}{\cT}}|$ is the number of nodes that are \emph{not} among $\leaves{\cT}{\lcatree{u}{v}{\cT}}$. We use $\OPTmw(G)$ to denote the revenue of an optimal HC tree under Dasgupta's cost for the graph $G$. 
\end{problem}

Observe that both objectives are \emph{composeable} w.r.t. edges, i.e., it is possible to divide the total objective to objectives induced by each (or each set of) edge(s). For any HC tree $\cT$ and any set of edge $E_1 \subseteq E$, we use $\rev_{G}(\cT, E_1)$ and $\cost_{G}(\cT, E_1)$ to denote the revenue and the cost induced by the edges in $E_1$.

By a straightforward calculation, one can show that for any HC tree $\cT$, there is $\rev_G(\cT) = \sum_{e=(u,v) \in E} w(e) \cdot n - \cost_G(\cT)$. Since $\sum_{e=(u,v) \in E} w(e) \cdot n$ is a deterministic function of the graph $G$ itself, the optimal HC tree $\cTstar$ under the two objectives are the same.  However, the two objectively admits vastly different \emph{approximation} algorithms. In particular, for the minimization objective, \citep{Dasgupta16} and the following work \citep{roy2016hierarchical,charikar2017approximate} showed that we can achieve an $O(\sqrt{\log{n}})$ approximation in polynomial time, and there is no $O(1)$ approximation in polynomial time assuming Small Set Expansion (SSE) hypothesis. On the other hand, for the revenue maximization objective, \citep{MoseleyW17} proved that the average-linkage heuristic can achieve a $1/3$ approximation in polynomial time. Therefore, we would naturally expect different results for hierarchical clustering with the splitting oracle with the two objectives.

\section{The Definitions and Results for Partial Hierarchical Clustering Trees}
\label{sec:partial-tree}
A technical backbone of our algorithms in this paper is the \emph{partial hierarchical clustering tree}. Roughly speaking, these structures replicate the organizational framework of the optimal HC tree, exhibiting only minor "ambiguity" within small subsets of vertices. In this section, we formally define the strong and weak partial trees and give efficient construction algorithms for them. We remark that our constructions of the partial HC trees are entirely based on the vertex set $V$ and the oracle $\cO$ of the input graph, irrespective of specific objective functions. This inherent independence renders our partial HC trees highly versatile and potentially of significant interest in their own right.

\subsection{Partial hierarchical clustering trees}
\label{subsec:partial-tree-definition}
We start by showing the definition of partial hierarchical clustering trees, which are very similar to the normal HC trees: the internal nodes represent subsets of vertices, and the leaves are individual vertices. However, in partial HC trees, we allow a collection of vertices that are not \emph{too large} to have \emph{unknown} local clustering, and we simply represent the whole set of vertices as a leaf node in the tree. The formal definition is as follows.

\begin{definition}[Partial hierarchical clustering trees]
\label{def:partial-tree}
A partial hierarchical clustering tree $\pTree$ is a binary tree such that
\begin{enumerate}[leftmargin=12pt]
\item The root represents the vertex set $V$.
\item For a node $x$ with children $(y,z)$, it represents a split $S_x \rightarrow (S_y, S_z)$, where $S_x = S_{y}\cup S_{z}$. 
\item The leaves of $\pTree$ corresponds to
\begin{itemize}
\item either a singleton vertex in $V$.
\item or a set of vertices $S\subseteq V$ such that $S\leq 50000 \log{n}$. In this case, we call the leave a \emph{super-vertex}.
\end{itemize}
\end{enumerate}
\end{definition}
Compared to the full hierarchical clustering tree, the partial HC tree allows the leaves to be `contracted' vertices with size at most $O(\log{n})$. We now define a partial tree that is \emph{strongly consistent} with a hierarchical clustering tree $\cT$.
\begin{definition}[Partial tree \emph{strongly} consistent with $\cT$]
\label{def:strongly-consistent-tree}
Let $\pTree$ be a partial hierarchical clustering tree and let $\cT$ be a (standard) hierarchical clustering tree. We say $\pTree$ is \emph{strongly consistent} with $\cT$ if 
\begin{enumerate}[leftmargin=12pt]
\item (\textit{Strong contraction property}) Each super-vertex induces a maximal subtree in $\cT$.
\item (\textit{Subtree preservation property}) For any pair of leaves $(x,y)$ in $\pTree$, let $X$ and $Y$ be the set of leaves corresponding to $x$ and $y$ in $\cT$ (recall that the leaves of $\pTree$ can be super-vertices). The subtree induced by $\lcatree{x}{y}{\pTree}$ contains the \emph{exactly} the same set of vertices as induced by $\lcatreeset{X\cup Y}{\cT}$.
\end{enumerate}
\end{definition}
In other words, a partial tree $\pTree$ is strongly consistent with $\cT$ if there exists a way to locally arrange tree structures for every super-vertex to \emph{exactly recover} $\cT$. An illustration of the strongly consistent partial tree (w.r.t. $\cT$) can be found in \Cref{fig:strong-consistent} in \Cref{subsec:partial-tree-and-HC}.

The strong partial tree is a very helpful data structure for HC. Nevertheless, finding such a strong partial tree could be challenging. As such, we also define partial trees that are \emph{weakly} consistent with the tree $\cT$ as follows.
\begin{definition}[Partial tree \emph{weakly} consistent with $\cT$]
\label{def:weakly-consistent-tree}
Let $\pTree$ be a partial hierarchical clustering tree and let $\cT$ be a (standard) hierarchical clustering tree. We say $\pTree$ is \emph{weakly consistent} with $\cT$ if
\begin{enumerate}[leftmargin=12pt]
\item (\textit{Weak contraction property}) Each super-vertex corresponds to a collection of maximal subtrees in $\cT$, i.e., $\cup_{i} V_{i}$ such that each $V_{i}$ satisfies
\[\leaves{\cT}{\lcatreeset{V_{i}}{\cT}} = V_{i}.\]
Furthermore, the collection $\cup_{i} V_{i}$ is with out-degree at most $2$ in $\cT$ such that
\begin{enumerate}
\item At most one edge is connected to a node that is the parent of the LCA of $\cup_{i} V_{i}$.
\item At most one edge is connected to a node that is a sibling of a maximal subtree of $\cT$ induced by (a subset of) $\cup_{i} V_{i}$.
\end{enumerate}
\item (\textit{Subtree preservation property}) For any pair of leaves $(x,y)$ in $\pTree$, let $X$ and $Y$ be the set of leaves corresponding to $x$ and $y$ in $\cT$ (recall that the leaves of $\pTree$ can be super-vertices). The subtree induced by $\lcatree{x}{y}{\pTree}$ contains the \emph{exactly} the same set of vertices as induced by $\lcatreeset{X\cup Y}{\cT}$.
\end{enumerate}
\end{definition}

The difference between the weak and strong consistency is that in weakly consistent partial trees, the ``contraction'' of vertices can happen in any consecutive region of the original tree $\cT$. An illustration of the weakly consistent partial tree (w.r.t. $\cT$) can be found in \Cref{fig:weak-consistent} in \Cref{subsec:partial-tree-and-HC}.

Our goal is to use oracle $\cO$, and vertex set $V$ to construct a partial HC tree $\pTree^{*}$ that is consistent with the optimal HC tree $\cT^*$.

\subsection{Main results of partial HC trees}
\label{subsec:partial-tree-results}
We now give our main results for the strong and weak partial trees, respectively. In particular, our results include
\begin{itemize}
\item An algorithm that, with high probability, constructs a partial tree strongly consistent with the optimal tree $\cTstar$ in $O(n^3\log{n})$ time and $O(n^3)$ queries to the splitting oracle. Furthermore, the algorithm uses only $\tilde{O}(n)$ space; and
\item An algorithm that, with high probability, constructs a partial tree weakly consistent with the optimal tree $\cTstar$ in $\tilde{O}(n^2)$ time and queries to the splitting oracle. Furthermore, the algorithm can be implemented in the PRAM model with $\tilde{O}(n^2)$ work and $\polylog{n}$ depth.
\end{itemize}

\paragraph{Result for strongly consistent partial HC trees}
Our main theorem to construct strongly consistent partial HC trees is as follows.
\begin{restatable}{theorem}{thmstrongpartialtree}
\label{thm:strong-partial-tree}
There exists an algorithm that given a splitting oracle $\cO$ of a weighted undirected graph $G=(V,E,w)$, with high probability, in $O(n^3\log{n})$ time and $O(n^3)$ queries computes a partial hierarchical clustering tree $\pTree$ that is strongly consistent with the optimal hierarchical clustering tree $\cTstar$. 
Furthermore, the algorithm has the following properties.
\begin{enumerate}[label=\roman*)., leftmargin=20pt]
\item The runtime of the algorithm is deterministic, and the high probability randomness is over the correctness guarantee.
\item The algorithm can be implemented in $O(n\log{n})$ space.
\end{enumerate}
\end{restatable}

While the algorithm outlined in \Cref{thm:strong-partial-tree} necessitates $O(n^3)$ queries to the oracle $\cO$, rendering it less efficient, the overall running time of $\Otilde(n^3)$ is tolerable, especially considering the hardness of HC.

\paragraph{Result for weakly consistent partial HC trees}
We now show our algorithmic results for the weakly consistent partial HC trees, which enjoy much better efficiency in both the running time and the number of oracle queries.

\begin{restatable}{theorem}{thmweakpartialtree}
\label{thm:weak-partial-tree}
There exists an algorithm that given a splitting oracle $\cO$ of a weighted undirected graph $G=(V,E,w)$, with high probability, in $O(n^2 \cdot \polylog{n})$ time and $O(n^2)$ queries computes a partial hierarchical clustering tree $\pTree$ that is weakly consistent with the optimal hierarchical clustering tree $\cTstar$.
\end{restatable}

We note that since the number of longest dependent calls for our weak partial tree is at most $O(\log^3{n})$, \Cref{thm:weak-partial-tree} implies a PRAM algorithm with $O(n^2\cdot \polylog{n})$ work and $O(\log^3{n})$ depth. The formal statement is as follows.
\begin{restatable}{corollary}{corweakpartialtree}
\label{cor:parallel-weak-partial-tree}
There exists a PRAM algorithm that given a graph $G=(V,E)$ and a splitting oracle $\cO$, with high probability, in $O(n^2 \cdot \polylog{n})$ work and $O(\log^{3}{n})$ depth computes a partial hierarchical clustering tree $\pTree$ that is weakly consistent with the optimal HC tree $\cTstar$.
\end{restatable}

We suspect that by modifying some subroutines of the algorithm in \Cref{thm:weak-partial-tree}, we could possibly bring the number of time and queries to $\Otilde(n)$.
The study of a \emph{sublinear-time} algorithm is an interesting open problem for future exploration.

In what follows, we first present the HC algorithms using the results for partial HC trees. We defer the detailed analysis of the partial trees to \Cref{sec:partial-tree-prelims,sec:strong-partial-tree,sec:weak-partial-tree}.

\section{Polynomial Time Algorithms for Dasgupta's Hierarchical Clustering Objective}
\label{sec:dasgupta-objective}
We introduce our polynomial time algorithms for Dasgupta's HC objective in this section. These results include an $O(1)$-approximation algorithm in polynomial time (albeit some large constant on the exponent) and an $O(\sqrt{\log\log{n}})$-approximation algorithm in $\tilde{O}(n^3)$ time. Our algorithms crucially rely on the strongly consistent partial tree in \Cref{thm:strong-partial-tree}.

\subsection{Existing Techniques for Dasgupta's Objective}
\label{subsec:known-hc-techniques}
To begin with, we discuss some known techniques for Dasgupta's minimization HC objective in this section. We will use these techniques in our hierarchical clustering algorithms for Dasgupta's objective.

\subsubsection*{Optimal Hierarchical Clustering Trees}
We first give an observation that characterizes the ``composability'' of HC costs with respect to the edges under Dasgupta's objective.

\begin{observation}[\citep{Dasgupta16}]
\label{obs:edge-partition}
	Let $G$ be any graph, and let $E_1$ and $E_2$ be
	two disjoint subsets of edges in $G$. For any HC tree $\cT$, let $\cost_{G}(\cT, E_1)$ and $\cost_{G}(\cT, E_2)$ be the HC costs induced by edges in $E_1$ and $E_2$, respectively. 
        Then, 
	\[
		\cost_{G}(\cT) = \cost_{G}(\cT, E_1) + \cost_{G}(\cT, E_2). 
	\]
\end{observation}

\Cref{obs:edge-partition} shows that to bound the total cost of the HC tree under Dasgupta's objective, it suffices to bound the edges split by the internal nodes.

\subsubsection*{Approximate HC trees with recursive balanced min-cuts and sparsest cuts}
Dasgupta's work proved that finding the optimal trees for the hierarchical clustering function is NP-hard \citep{Dasgupta16}. Consequently, significant attention has been directed towards developing \emph{approximation algorithms} for efficient hierarchical clustering. A well-known approach involves obtaining an approximation of the optimal hierarchical clustering by iteratively employing \emph{sparsest cuts} on the graph, e.g., \citep{Dasgupta16,Deng0U0Z25}.  
Formally, we can define the sparsest cuts and the HC trees created by recursively applying the sparsest cuts on the induced subgraphs as in \Cref{def:sparse-cut} and \Cref{def:sparse-cuts-procedure}.

\begin{definition}[\textbf{Sparsest Cuts}]\label{def:sparse-cut}
For any parameter $\beta$ such that $0 < \beta < 1$, we say that a cut $(A^* ,B^*)$ is a \emph{sparsest cut} if its \emph{sparsity (edge expansion)} is minimized, i.e.
\[ \frac{w(A^*,B^*)}{\min\{\card{A^*},\card{B^*}\}} \leq \frac{w(A,B)}{\min\{\card{A},\card{B}\}}\] 
for any cut $(A,B)$ of $G$.
\end{definition}

\begin{definition}[\textbf{Recursive Sparsest Cut Procedure}]\label{def:sparse-cuts-procedure}
We say an HC tree $\cT$ is obtained by the \emph{recursive sparsest procedure} on $G$ if for each non-leaf node $z$ of $\cT$, the $\cut(\cT[z])$ is obtained by a (possibly approximate) sparsest cut $(A,B)$ on the subgraph induced by $\cT[z]$. We call an HC tree obtained by recursively applying (approximate) sparsest cuts on induced subgraphs as a recursive sparsest cut HC tree.
\end{definition}

Previous work (see, e.g.~\citep{charikar2017approximate,AssadiCLMW22}) proved that if one applies the procedure in \Cref{def:sparse-cuts-procedure}, we can get an $O(1)$ approximation of the optimal HC tree. 
\begin{proposition}[\citep{charikar2017approximate,AssadiCLMW22}]\label{prop:const-recursive-HC}
	For any graph $G=(V,E,w)$, let $\cT_{sparse}$ be an HC tree obtained by the recursive sparsest cut procedure in \Cref{def:sparse-cuts-procedure}, there is 
    \[
    \cost_G(\cT_{sparse}) \leq O(1) \cdot \OPT(G).
    \]
\end{proposition}

Note that finding the exact sparsest cut for \Cref{prop:const-recursive-HC} is NP-hard. The first way to circumvent this issue is to use approximation algorithms, especially for the best-known $O(\sqrt{\log{n}})$ approximation for sparsest cut in polynomial time (\citep{AroraHK04}). The formal guarantee is as follows.
\begin{proposition}[\citep{AroraHK04}]
\label{prop:sparse-cut}
There exists a randomized algorithm that given a graph $G=(V,E,w)$, with high probability, in $\tilde{O}(\card{V}^2)$ time finds a partition $A \cup B=V$ such that
\[\frac{w(A,B)}{\min\{\card{A},\card{B}\}} \leq O(\sqrt{\log{\card{V}}})\cdot \frac{w(A^*,B^*)}{\min\{\card{A^*},\card{B^*}\}},\]
where $(A^*,B^*)$ is the sparsest cut of $G$. 
\end{proposition}
We cannot immediately massage \Cref{prop:sparse-cut} with \Cref{prop:const-recursive-HC} since \Cref{prop:const-recursive-HC} does \emph{not} state what will happen for \emph{approximate} sparsest cuts. Fortunately, by the results in \citep{charikar2017approximate,AssadiCLMW22}, we can indeed obtain an $O(\alpha)$-approximation algorithm for $\OPT(G)$ by recursively applying the $\alpha$-approximate sparsest cut.
\begin{proposition}[\citep{charikar2017approximate,AssadiCLMW22}]
\label{prop:beta-compute}
Let $\cT$ be a recursive sparsest cut tree obtained by recursively applying $\alpha$-approximation sparsest cuts on the induced subgraphs (as in \Cref{def:sparse-cuts-procedure}). Then, we have
\begin{align*}
\cost_G(\cT) \leq O(\alpha)\cdot \OPT(G).
\end{align*}
\end{proposition}

There is another way to deal with the NP-hardness issue of the sparsest cut. If we can reduce the \emph{input size}, we can possibly obtain the \emph{exact optimal} cuts on induced subgraphs of size $O(\log{n})$. This strategy allows us to leverage our strong partial tree whose `unknown' clustering is only restricted to the induced subgraphs with $O(\log{n})$ size.

\subsection{A Polynomial-time Algorithm for \texorpdfstring{$O(1)$}{O1}-approximation on Dasgupta's HC Objective}
\label{subsec:constant-poly-time-HC}
We now introduce our algorithm that finds an HC tree with $O(1)$-approximation to Dasgupta's objective in polynomial time. The formal statement is given as follows.

\thmconstantpolytimehc*

Our algorithm of \Cref{thm:constant-poly-time-HC} is simply as follows.

\begin{algorithm}
\caption{A polynomial-time algorithm for the Dasgupta's HC objective} \label{alg:HC-Das}
\KwIn{Input graph $G=(V,E,w)$; Splitting oracle $\cO$}
\KwOut{A hierarchical clustering tree $\cT$}
Run the strong partial tree approximation algorithm in \Cref{thm:strong-partial-tree} to obtain partial tree $\pTree$ \;
\For{each super-vertex in $\pTree$}{
On input vertex set $S$, \emph{exhaustively search} the sparsest cut $(A,B)$ on the induced subgraph $G[S]$\;
Partition the vertices as $S\rightarrow (A,B)$, and recurse on $G[A]$ and $G[B]$\; 
}
\end{algorithm}

We now prove the efficiency and the approximation guarantees for Dasgupta's objective. The following lemma provides the efficiency for \Cref{alg:HC-Das}.

\begin{lemma}
\label{lem:HC-das-poly-efficiency}
\Cref{alg:HC-Das} runs (deterministically) in $O(n^{50002})$ time and uses $O(n^3)$ queries.
\end{lemma}
\begin{proof}
By \Cref{thm:strong-partial-tree}, the first step of the algorithm that computes the strong partial tree takes $\tilde{O}(n^3)$ time and $O(n^3)$ queries. Note that we only take queries in this step.

For the second step, when the input size is $s$, an exhaustive search on the sparsest cut takes $O(2^s)$ time. As such, let $X$ be the set of induced vertices for a single super-vertex of $\pTree$, since we have $\card{X}\leq 50000\log{n}$, it only takes $O(n^{50000})$ time to find the sparsest cut. Similarly, we can show that in each recursive call, the runtime is at most $O(n^{50000})$. By \Cref{fct:num-splits-tree}, there are at most $O(\log{n})$ nodes in a binary tree with $O(\log{n})$ leaves. As such, the recursive sparsest cut for a single super-vertex in $\pTree$ takes $O(n^{50000}\cdot \log{n})$ time. There are at most $O(n)$ super-vertices in the tree; therefore, the total runtime of the second step takes $O(n^{50000}\cdot \log{n}\cdot n)=O(n^{50002})$ time.

Combining the efficiency of the two steps gives us the desired efficiency bound. 
\end{proof}

We now move to the approximation guarantee of the algorithm.
\begin{lemma}
\label{lem:HC-das-const-approx}
Conditioning on the high probability guarantees of \Cref{thm:strong-partial-tree}, \Cref{alg:HC-Das} outputs an HC tree $\cT$ that achieves $O(1)$-approximation to the Dasgupta's objective.
\end{lemma}
\begin{proof}
For any partial tree $\pTree$, we first partition the edges in to $E_{\text{cross}}$ and $E_{\text{same}}$ based on whether the edge $(u,v)\in E$ crosses different partial trees, i.e.,
\begin{enumerate}
\item $(u,v)\in E_{\text{cross}}$ iff $u\in X$ and $v\in Y$ for some super-vertices $X\neq Y$ in $\pTree$.
\item $(u,v)\in E_{\text{same}}$ iff $u,v \in X$ for some super-vertex $X$ in $\pTree$.
\end{enumerate}
Since $E=E_{\text{cross}} \cup E_{\text{same}}$, by using \Cref{obs:edge-partition}, we can show that $\OPTdas = \cost_{G}(\cTstar)= \cost_{G}(\cTstar, E_1)+\cost_{G}(\cTstar, E_2)$. We now analyze the costs w.r.t. to $E_1$ and $E_2$, respectively.
\begin{enumerate}
\item For $E_{\text{cross}}$, we argue that $\cost_{G}(\cT, E_{\text{cross}})=\cost_{G}(\cTstar, E_{\text{cross}})$. To see this, note that if $u$ and $v$ are of different super-vertices, by the definition of partial HC trees that are \emph{strongly consistent} with the optimal tree $\cTstar$, there is 
\[\leaves{\cT}{\lcatree{u}{v}{\cT}} = \leaves{\cTstar}{\lcatree{u}{v}{\cTstar}}.\]
As such, we have $\cost_{G}(\cT, E_{\text{cross}})=\cost_{G}(\cTstar, E_{\text{cross}})$ by the definition of the cost function.
\item For $E_{\text{same}}$, we argue that $\cost_{G}(\cT, E_{\text{same}})\leq O(1)\cdot \cost_{G}(\cTstar, E_{\text{same}})$. Formally, for each super-vertex $X$, we can use \Cref{prop:const-recursive-HC} on $G[X]$ to argue that the $\cost_{G}(\cT, E_{\text{same}}[X])\leq O(1)\cdot \cost_{G}(\cTstar, E_{\text{same}}[X])$, where $E_{\text{same}}[X]$ stands for the set of edges in $E_{\text{same}}$ with \emph{both} endpoints in $X$. Therefore, we can apply this calculation to every super-vertex to get the desired approximation factor.
\end{enumerate}
We now use \Cref{obs:edge-partition} again on $E_{\text{cross}}$ and $E_{\text{same}}$ to bound that
\begin{align*}
\cost_{G}(\cT)&=\cost_{G}(\cT, E_{\text{cross}}) + \cost_{G}(\cT, E_{\text{same}}) \\
&\leq \cost_{G}(\cTstar, E_{\text{cross}}) + O(1)\cdot \cost_{G}(\cTstar, E_{\text{same}})\\
&\leq O(1)\cdot \paren{\cost_{G}(\cTstar, E_{\text{cross}}) + \cost_{G}(\cTstar, E_{\text{same}})}\\
&= O(1)\cdot\cost_{G}(\cTstar)=O(1)\cdot \OPTdas,
\end{align*}
as desired.
\end{proof}

\paragraph{Finalizing the proof of \Cref{thm:constant-poly-time-HC}.} With \Cref{alg:HC-Das}, we can obtain the poly-time efficiency from \Cref{lem:HC-das-const-approx}. Furthermore, since the strong partial tree algorithm of \Cref{thm:strong-partial-tree} succeeds with high probability, the approximation guarantee of \Cref{lem:HC-das-const-approx} holds with high probability as well. This concludes the proof.

\subsection{An \texorpdfstring{$\tilde{O}(n^3)$}{nto3} Time Algorithm for \texorpdfstring{$O(\sqrt{\log\log{n}})$}{rootloglogn}-approximation on Dasgupta's HC Objective}
\label{subsec:loglogn-n4-time-HC}
One drawback of the algorithm we have in \Cref{subsec:constant-poly-time-HC} is that the efficiency is ``theoretical only'' -- after all, a runtime of $O(n^{50002})$ is nowhere near being practical. Observe that the subroutine that leads to the very large exponent is the exhaustive search of the \emph{optimal} sparsest cut. Therefore, we can hope to use some more efficient approximation for sparsest cuts while not sacrificing too much on the approximation guarantee. This intuition leads us to the following theorem.

\thmrootloglogtimehc*

Our algorithm for \Cref{thm:root-loglog-n4-time-HC} is almost identical to that of \Cref{thm:constant-poly-time-HC} -- the only difference is that we substitute the exact sparsest cuts with \emph{approximate} ones. The description of the algorithms is as follows.

\begin{algorithm}
\caption{An $\Otilde(n^3)$ time algorithm for the Dasgupta's HC objective}\label{alg:HC-Das-fast}
\KwIn{Input graph $G=(V,E,w)$; Splitting oracle $\cO$}
\KwOut{A hierarchical clustering tree $\cT$}
Run the strong partial tree approximation algorithm in \Cref{thm:strong-partial-tree} to obtain partial tree $\pTree$ \;
\For{each super-vertex in $\pTree$}{ 
On input vertex set $S$, find an $O(\sqrt{\log \card{S}})$-approximation of the sparsest cut $(A,B)$ on the induced subgraph $G[S]$ using the algorithm of \Cref{prop:sparse-cut}\; 
Partition the vertices as $S\rightarrow (A,B)$, and recurse on $G[A]$ and $G[B]$ \;
}
\end{algorithm}

We now proceed to prove the efficiency and the approximation guarantees for Dasgupta's objective with \Cref{alg:HC-Das-fast}. Again, we first characterize the efficiency of the algorithm.

\begin{lemma}
\label{lem:HC-das-n4-efficiency}
With high probability, \Cref{alg:HC-Das-fast} runs in $O(n^3\log{n})$ time and uses $O(n^3)$ queries.
\end{lemma}
\begin{proof}
The first step of the algorithm that computes the strong partial tree takes $O(n^3\log{n})$ time and $O(n^3)$ queries by \Cref{thm:strong-partial-tree}. We need to argue that the second step takes at most $O(n^3\log{n})$ time as well. Note that by \Cref{prop:sparse-cut}, the algorithm, with high probability, runs in $\tilde{O}(s^2)$ time and finds an $O(\sqrt{\log s})$-approximation of the sparsest cut $(A,B)$, where $s$ is the input size. In our case, for a single super-vertex of $\pTree$ with $\card{X}\leq 50000\log{n}$, the running time is therefore $O(\log^2{n})$ only. By \Cref{fct:num-splits-tree}, there are at most $O(\log{n})$ nodes in a binary tree with $O(\log{n})$ leaves, which implies $O(\log^3{n})$ running time for a single super-vertex. Finally, with at most $O(n)$ super-vertices in the tree, the second step only takes $O(n\cdot \log^3{n})=O(n^3\log{n})$ time, as desired.
\end{proof}

We now move to the approximation guarantee of the \Cref{alg:HC-Das-fast}.
\begin{lemma}
\label{lem:HC-das-sqrt-loglog-approx}
Conditioning on the high probability guarantees of \Cref{thm:strong-partial-tree}, \Cref{alg:HC-Das-fast} outputs an HC tree $\cT$ that achieves $O(\sqrt{\log\log{n}})$-approximation to the Dasgupta's objective.
\end{lemma}
\begin{proof}
Similar to the proof of \Cref{lem:HC-das-const-approx}, for any partial tree $\pTree$, we first partition the edges in to $E_{\text{cross}}$ and $E_{\text{same}}$ based on whether the edge $(u,v)\in E$ crosses different partial trees, i.e.,
\begin{enumerate}
\item $(u,v)\in E_{\text{cross}}$ iff $u\in X$ and $v\in Y$ for some super-vertices $X\neq Y$ in $\pTree$.
\item $(u,v)\in E_{\text{same}}$ iff $u,v \in X$ for some super-vertex $X$ in $\pTree$.
\end{enumerate}
Again, by using \Cref{obs:edge-partition}, we can show that $\OPTdas = \cost_{G}(\cTstar)= \cost_{G}(\cTstar, E_1)+\cost_{G}(\cTstar, E_2)$. The costs w.r.t. to $E_1$ and $E_2$ are therefore as follows.
\begin{enumerate}
\item For $E_{\text{cross}}$, we have that $\cost_{G}(\cT, E_{\text{cross}})=\cost_{G}(\cTstar, E_{\text{cross}})$ by using the same argument of \Cref{lem:HC-das-const-approx}. 
\item For $E_{\text{same}}$, we argue that $\cost_{G}(\cT, E_{\text{same}})\leq O(\sqrt{\log\log{n}})\cdot \cost_{G}(\cTstar, E_{\text{same}})$. Formally, since algorithm in \Cref{prop:sparse-cut} finds an $O(\sqrt{\log s})$-approximation for each super-vertex $X$, we can use \Cref{prop:const-recursive-HC} on $G[X]$ to argue that the $\cost_{G}(\cT, E_{\text{same}}[X])\leq O(\sqrt{\log\log{n}})\cdot \cost_{G}(\cTstar, E_{\text{same}}[X])$ since $s=O(\log{n})$. Therefore, we can apply the same argument to every super-vertex to get the desired approximation factor.
\end{enumerate}
We now use \Cref{obs:edge-partition} again on $E_{\text{cross}}$ and $E_{\text{same}}$ to bound that
\begin{align*}
\cost_{G}(\cT)&=\cost_{G}(\cT, E_{\text{cross}}) + \cost_{G}(\cT, E_{\text{same}}) \\
&\leq \cost_{G}(\cTstar, E_{\text{cross}}) + O(\sqrt{\log\log{n}})\cdot \cost_{G}(\cTstar, E_{\text{same}})\\
&\leq O(\sqrt{\log\log{n}})\cdot \paren{\cost_{G}(\cTstar, E_{\text{cross}}) + \cost_{G}(\cTstar, E_{\text{same}})}\\
&= O(\sqrt{\log\log{n}})\cdot\cost_{G}(\cTstar)=O(\sqrt{\log\log{n}})\cdot \OPTdas,
\end{align*}
as desired.
\end{proof}

\paragraph{Finalizing the proof of \Cref{thm:root-loglog-n4-time-HC}.} With \Cref{alg:HC-Das-fast}, we can obtain the $O(n^4)$ time efficiency from \Cref{lem:HC-das-n4-efficiency} with high probability. Therefore, we can apply a union bound on the high probability events of \Cref{lem:HC-das-n4-efficiency} and \Cref{lem:HC-das-const-approx} to get the desired efficiency and approximation guarantee.

\newcommand{\Elow}{\ensuremath{E_{\textnormal{low}}}\xspace}
\newcommand{\Esame}{\ensuremath{E^{\textnormal{same}}_{\textnormal{high}}}\xspace}
\newcommand{\Ediff}{\ensuremath{E^{\textnormal{diff}}_{\textnormal{high}}}\xspace}

\section{Near-linear Time Algorithms for Moseley-Wang Hierarchical Clustering Objective}
\label{sec:MW-objective-algs}
We introduce our algorithm for the Moseley-Wang HC objective in this section. The algorithm is based on the weakly consistent trees as in \Cref{thm:weak-partial-tree}, and our analysis is more involved compared to the results in \Cref{sec:dasgupta-objective} -- it requires a careful handling of the contributions of the `less significant edges' to the Moseley-Wang objective.

Our main result is a near-linear time algorithm for the Moseley-Wang objective that achieves $(1-o(1))$ multiplicative approximation.

\thmhcalgmw*

Our algorithm is simply as follows.
\begin{algorithm}
\caption{A near-linear time algorithm for the Moseley-Wang HC objective}\label{alg:HC-MW}
\KwIn{Input graph $G=(V,E,w)$; Splitting oracle $\cO$}
\KwOut{A hierarchical clustering tree $\cT$}
Run the weak partial tree approximation algorithm in \Cref{thm:weak-partial-tree} to obtain partial tree $\pTree$\;
\For{each super-vertex in $\pTree$}{ 
partition the leaves \emph{arbitrarily} to obtain an HC tree $\cT$\; 
}
\end{algorithm}
Since each super-vertex contains $O(\log{n})$ vertices, we can always partition the super-vertices in $\polylog{n}$ time\footnote{An arbitrary recursive balanced partition only requires a run time of $O(\log{n}\cdot \log\log{n})$}. There are at most $O(n/\log{n})$ such super-vertices, and the total runtime is at most $O(n\cdot \polylog{n})$, which is dominated by the runtime of the weakly-consistent partial tree construction as in \Cref{thm:weak-partial-tree}.

To prove the approximation guarantee of \Cref{thm:hc-alg-MW}, we will show the following technical lemma that lower bounds the number of non-leaves between $\cT$ and $\cTstar$.

\begin{lemma}
\label{lem:num-non-leaf-weak-tree}
Let $\cT$ be a hierarchical clustering tree obtained by \Cref{alg:HC-MW}, and let $u,v\in V$ be any two vertices. Then, conditioning on the high probability event that $\pTree$ is a partial tree that is weakly consistent with $\cTstar$, there is
\begin{itemize}
\item If $u$ and $v$ are in the same super-vertex $X$ of $\pTree$, then, there is
\begin{align*}
\card{\nonleaves{\cT}{\lcatree{u}{v}{\cT}}}\geq n - 50000\log{n}.
\end{align*}
\item If $u$ and $v$ are in different super-vertices $X$ and $Y$ of $\pTree$, then, there is
\begin{align*}
\card{\nonleaves{\cT}{\lcatree{u}{v}{\cT}}}\geq \card{\nonleaves{\cTstar}{\lcatree{u}{v}{\cTstar}}} - 50000\log{n}.
\end{align*}
Furthermore, if 
\[\leaves{\cTstar}{\lcatreeset{X}{\cTstar}}\cap \leaves{\cTstar}{\lcatreeset{Y}{\cTstar}}=\emptyset,\]
then, we additionally have 
\[\card{\nonleaves{\cT}{\lcatree{u}{v}{\cT}}}= \card{\nonleaves{\cTstar}{\lcatree{u}{v}{\cTstar}}}.\]
\end{itemize}
\end{lemma}
\begin{proof}
We prove the two cases separately as follows.
\begin{itemize}[leftmargin=5pt]
\item \textbf{$u$ and $v$ are in the same super-vertex $X$ of $\pTree$.} 
By our construction, every leaf that is outside the subtree induced by $\leaves{\cT}{\lcatreeset{X}{\cT}}$ counts as a non-leave of $(u,v)$. As such, since $\card{X}\leq 50000\log{n}$, it is straightforward to get that
\[\card{\nonleaves{\cT}{\lcatree{u}{v}{\cT}}}\geq n - 50000\log{n}.\]

\item \textbf{$u$ and $v$ are in different super-vertices $X,Y$ of $\pTree$.} Suppose w.log. that $u\in X$ and $v \in Y$. By the \textit{subtree preserving property}, we have $\leaves{\cT}{\lcatree{u}{v}{\cT}}=\leaves{\cTstar}{\lcatreeset{X\cup Y}{\cTstar}}$. We now discuss further two sub-cases:
\begin{enumerate}[label=\alph*).]
\item If $\leaves{\cTstar}{\lcatreeset{X}{\cTstar}}$ and $\leaves{\cTstar}{\lcatreeset{Y}{\cTstar}}$ are disjoint. In this case, we have exactly $\leaves{\cTstar}{\lcatreeset{X\cup Y}{\cTstar}}=\leaves{\cTstar}{\lcatree{u}{v}{\cTstar}}$, which implies $\leaves{\cT}{\lcatree{u}{v}{\cT}}=\leaves{\cTstar}{\lcatree{u}{v}{\cTstar}}$. Therefore, we have
\[\card{\nonleaves{\cT}{\lcatree{u}{v}{\cT}}}= \card{\nonleaves{\cTstar}{\lcatree{u}{v}{\cTstar}}}.\]
This proves the ``furthermore'' part of the second case.

\item If $\leaves{\cTstar}{\lcatreeset{X}{\cTstar}}$ and $\leaves{\cTstar}{\lcatreeset{Y}{\cTstar}}$ have intersections. In this case, note that by the \textit{weak contraction property}, we must have \emph{inclusion} relationship between $\leaves{\cTstar}{\lcatreeset{X}{\cTstar}}$ and $\leaves{\cTstar}{\lcatreeset{Y}{\cTstar}}$. Suppose without loss of generality, that $\leaves{\cTstar}{\lcatreeset{X}{\cTstar}} \supseteq \leaves{\cTstar}{\lcatreeset{Y}{\cTstar}}$, the differences between $\leaves{\cTstar}{\lcatreeset{X\cup Y}{\cTstar}}$ and $\leaves{\cTstar}{\lcatree{u}{v}{\cTstar}}$ is at most the set of $X$. Since $\card{X}\leq 50000\log{n}$, we have
\[\card{\nonleaves{\cT}{\lcatree{u}{v}{\cT}}}\geq \card{\nonleaves{\cTstar}{\lcatree{u}{v}{\cTstar}}} - 50000\log{n},\]
as desired.
\end{enumerate}
\end{itemize}
This concludes the proof of \Cref{lem:num-non-leaf-weak-tree}. 
\end{proof}

To prove the approximation guarantee of \Cref{thm:hc-alg-MW}, we will need a handful of structural observations for the optimal HC trees in the Moseley-Wang objective. We first observe that for any optimal HC tree $\cTstar$ necessarily has a ``monotone edge weight'' property between an internal node and its descendants.

\begin{figure}[!htb]
	\centering
	\includegraphics[width=0.8\textwidth]{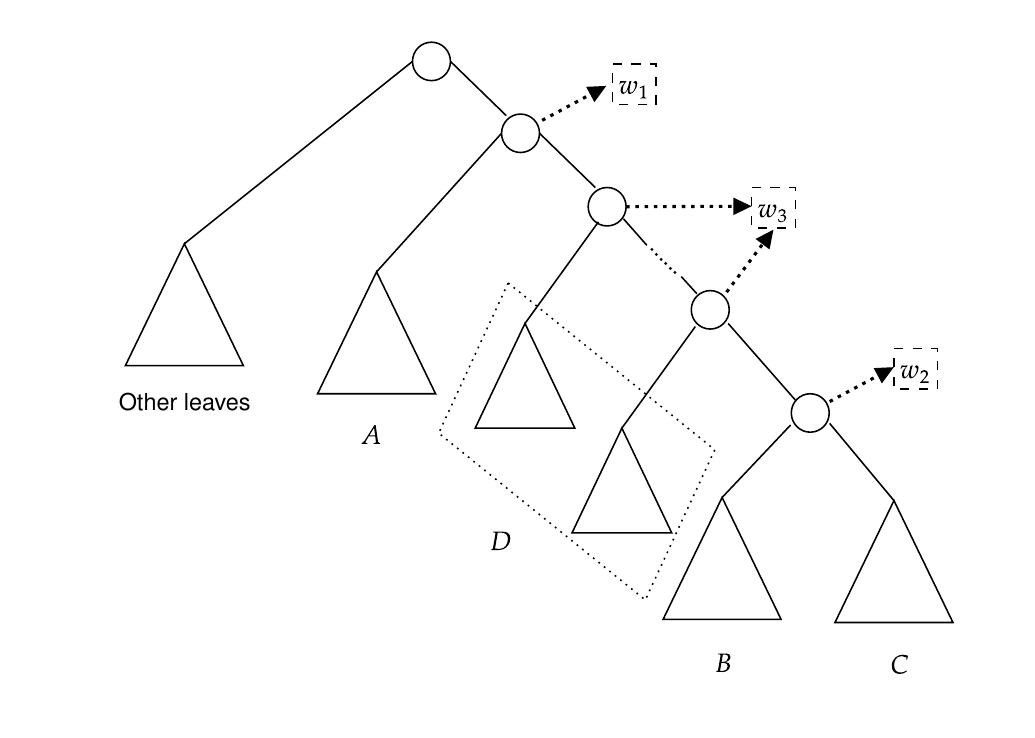}
	\caption{\label{fig:mw-OPT-structure} An illustration of the edge weights and leaves described in \Cref{clm:weight-property-internal-nodes}.}
\end{figure}

\begin{claim}
\label{clm:weight-property-internal-nodes}
Let $z_{1}$ and $z_{2}$ be any two internal nodes of the optimal HC tree $\cTstar$ under the Moseley-Wang objective, and let $z_{2}$ be a descendant node of $z_{1}$. We use $w_1$ and $w_2$ to denote the total weights of the edges induced by $z_1$ and $z_2$, respectively. Furthermore, let $w_3$ be the total weights of the edges induced on the internal nodes between $z_1$ and $z_2$.

Suppose the partition induced by $z_{1}$ is $S\rightarrow (A, S\setminus A)$, and let $B$ and $C$ be the set of vertices induced by $z_{2}$ such that $B\cup C\subseteq S\setminus A$. Furthermore, let $D=(S\setminus A)\setminus (B\cup C)$, and suppose $\max\{\card{B},\card{C}\}\geq \card{A}$. Then, there is
\begin{align*}
\frac{w_1-(\frac{\card{A}}{\card{D}}+1) \cdot w_3}{\card{A}+\card{D}}\leq \frac{w_2}{\max\{\card{B},\card{C}\}}.
\end{align*}
An illustration of the edges and leaves used in the statement can be found in \Cref{fig:mw-OPT-structure}.
\end{claim}
\begin{proof}
The claim is similar in spirit to the ``switching lemma'' under Dasgupta's objective as proved in \citep{HogemoBBPT21}. Suppose w.log. that $\card{B}\geq \card{C}$. For the purpose of this proof (and also that of \Cref{clm:weight-property-w-A-B}), we let $\rev_{G}(\cT, E_{1})$ be the revenue induced by a subset of edges $E_{1}$ with $\cT$ being the HC tree of $G$. We further observe some useful relationships between the weights $w_1$, $w_2$, $w_3$ and the weights $w(A,B)$, $w(A,C)$, $w(A,D)$, $w(B,C)$, $w(B,D)$, and $w(C,D)$ as follows.
\begin{align*}
w_1 = w(A,D)+w(A,B)+w(A,C) \qquad w_2=w(B,C) \qquad w_3 \geq w(B,D) + w(C,D).
\end{align*}
We first prove a self-contained structural claim that the edge weights between $A$ and $B$ cannot be too much bigger than $w_3$. More formally, the claim is as follows.
\begin{claim}
\label{clm:weight-property-w-A-B}
Let $A$, $B$, $C$, and $D$ be the set of vertices and $w_1$, $w_2$, and $w_3$ be the edge weights as prescribed in \Cref{clm:weight-property-internal-nodes}. Furthermore, let $E(A,B)$ be the edges between $A$ and $B$, and $w(A,B)$ be the weights of $E(A,B)$. Then, there is
\begin{align*}
w(A,B)\leq \frac{\card{A}}{\card{D}}\cdot w_{3}.
\end{align*}
\end{claim}
\begin{proof}
As consistent with the proof of \Cref{clm:weight-property-internal-nodes}, we assume w.log. that $\card{B}\geq \card{C}$. Let us construct a tree $\cT^{(1)}$ based on $\cTstar$ by switching the subtrees induced by $A$ and $D$. Note that in such change of HC tree, the only edges that will have a changed revenue are $E(A,B)$, $E(A,C)$, $E(A,D)$, $E(B,C)$, and edges accounted by $w_3$, which we call $E_3$ (including but not limited to $E(B,D)$ and $E(C,D)$), and edges in. We list these changes as follows. 
\begin{enumerate}
\item $E(A,B)$: we have that $\rev_{G}(\cT^{(1)}, E_{A,B})-\rev_{G}(\cTstar, E_{A,B})\geq \card{D}\cdot w(A,B)$.
\item $E(A,C)$: we have that $\rev_{G}(\cT^{(1)}, E_{A,C})-\rev_{G}(\cTstar, E_{A,C})\geq \card{D}\cdot w(A,C)$.
\item $E(A,D)$: we have that $\rev_{G}(\cT^{(1)}, E_{A,D})-\rev_{G}(\cTstar, E_{A,D}) = 0$.
\item $E(B,C)$: we have that $\rev_{G}(\cT^{(1)}, E_{B,C})-\rev_{G}(\cTstar, E_{B,C}) = 0$.
\item $E_3$: we have that $\rev_{G}(\cT^{(1)}, E_{3})-\rev_{G}(\cTstar, E_{3}) \geq - w_3 \cdot \card{A}$.
\end{enumerate}
To maintain the optimality of $\cTstar$, there should be $\rev_{G}(\cT^{(1)})-\rev_{G}(\cTstar)\leq 0$. As such, we have
\begin{align*}
&0 \geq \rev_{G}(\cT^{(1)})-\rev_{G}(\cTstar)\\
&= \paren{\rev_{G}(\cT^{(1)}, E_{A,B})-\rev_{G}(\cTstar, E_{A,B})} + \paren{\rev_{G}(\cT^{(1)}, E_{A,C})-\rev_{G}(\cTstar, E_{A,C})} 
\\ 
& \quad + \paren{\rev_{G}(\cT^{(1)}, E_{A,D})-\rev_{G}(\cTstar, E_{A,D})} + \paren{\rev_{G}(\cT^{(1)}, E_{B,C})-\rev_{G}(\cTstar, E_{B,C})} \\
& \quad + \rev_{G}(\cT^{(1)}, E_{3})-\rev_{G}(\cTstar, E_{3})\\
&\geq \paren{w(A,B)+w(A,C)}\cdot \card{D}-w_3 \cdot \card{A}.
\end{align*}
As such, we can move the terms around, and obtain that
\begin{align*}
w(A,B)\leq w(A,B)+w(A,C) \leq \frac{\card{A}}{\card{D}}\cdot w_3,
\end{align*}
as desired. 
\myqed{\Cref{clm:weight-property-w-A-B}}
\end{proof}

We now use \Cref{clm:weight-property-w-A-B} to prove \Cref{clm:weight-property-internal-nodes}. 
We again construct a tree $\cT^{(2)}$ based on $\cTstar$ by switching the subtrees induced by \emph{$A$ and $B$} (note that this is different from the proof of \Cref{clm:weight-property-w-A-B}). Observe again that only the edges that have different revenue contribution in $\cT^{(2)}$ vs. $\cTstar$ are $E(A,B)$, $E(A,C)$, $E(A,D)$, $E(B,C)$, and edges accounted by $w_3$. We again list all the changes of revenue induced on these edges.

\begin{enumerate}
\item $E(A,B)$: we have that $\rev_{G}(\cT^{(2)}, E_{A,B})-\rev_{G}(\cTstar, E_{A,B})=0$.
\item $E(A,C)$: we have that $\rev_{G}(\cT^{(2)}, E_{A,C})-\rev_{G}(\cTstar, E_{A,C})\geq (\card{B}+\card{D})\cdot w(A,C)$.
\item $E(A,D)$: we have that $\rev_{G}(\cT^{(2)}, E_{A,D})-\rev_{G}(\cTstar, E_{A,D}) \geq \card{B}\cdot w(A,D)$.
\item $E(B,C)$: we have that $\rev_{G}(\cT^{(2)}, E_{B,C})-\rev_{G}(\cTstar, E_{B,C}) \geq -(\card{A}+\card{D}) \cdot w(B,C)$.
\item $E_3$: we have that $\rev_{G}(\cT^{(2)}, E_{3})-\rev_{G}(\cTstar, E_{3}) \geq - \card{A}\cdot w_3$.
\end{enumerate}

Therefore, by the optimally of $\cTstar$, there should be $\rev_{G}(\cT^{(2)})-\rev_{G}(\cTstar)\leq 0$. As such, we have
\begin{align*}
&0 \geq \rev_{G}(\cT^{(2)})-\rev_{G}(\cTstar)\\
&= \paren{\rev_{G}(\cT^{(2)}, E_{A,B})-\rev_{G}(\cTstar, E_{A,B})} + \paren{\rev_{G}(\cT^{(2)}, E_{A,C})-\rev_{G}(\cTstar, E_{A,C})} 
\\ 
& \quad + \paren{\rev_{G}(\cT^{(2)}, E_{A,D})-\rev_{G}(\cTstar, E_{A,D})} + \paren{\rev_{G}(\cT^{(2)}, E_{B,C})-\rev_{G}(\cTstar, E_{B,C})} \\
& \quad + \paren{\rev_{G}(\cT^{(2)}, E_{3})-\rev_{G}(\cTstar, E_{3})}\\
&\geq (\card{B}+\card{D})\cdot w(A,C) + \card{B}\cdot w(A,D)  -(\card{A}+\card{D}) \cdot w(B,C) - w_3 \cdot \card{A}.
\end{align*}
As such, by moving the terms around, we can get that 
\begin{align*}
\card{B} \cdot \paren{w(A,C) +  w(A,D)} & \leq (\card{B}+\card{D})\cdot w(A,C) + \card{B}\cdot w(A,D) \\
& \leq (\card{A}+\card{D})\cdot w(B,C)+  \card{A} \cdot w_3.
\end{align*}
We can use the observation that $w_1 = w(A,B)+w(A,C)+w(A,D)$ to obtain that
\begin{align*}
\card{B} \cdot w_1 & = \paren{w(A,C) + w(A,B) +  w(A,D)} \\
& \leq (\card{A}+\card{D})\cdot w(B,C)+ \card{A}\cdot w_3 + \card{B}\cdot w(A,B)
\end{align*}
by adding $\card{B}\cdot w(A,B)$ on both sides. Now, we can apply \Cref{clm:weight-property-w-A-B} to obtain that
\begin{align*}
\card{B} \cdot w_1 & \leq (\card{A}+\card{D})\cdot w(B,C)+ \card{A}\cdot w_3 + \card{B}\cdot \frac{\card{A}}{\card{D}}\cdot w_3\\
&\leq (\card{A}+\card{D})\cdot w(B,C) + \card{B}\cdot w_3 + \card{B}\cdot \frac{\card{A}}{\card{D}}\cdot w_3 \tag{using $\card{A}\leq \card{B}$}
\end{align*}
Note that $w(B,C)=w_2$. As such, the above implies that
\begin{align*}
\card{B} \cdot \paren{w_1 - (1+\frac{\card{A}}{\card{D}})w_3} \leq (\card{A}+\card{D})\cdot w_2.
\end{align*}
Moving the turns around in the above inequality gives us the desired bound. \myqed{\Cref{clm:weight-property-internal-nodes}}
\end{proof}

We now use \Cref{clm:weight-property-internal-nodes} to show that the set of edges $(u,v)$ such that 
\begin{enumerate}[label=\alph*).]
\item has at most $O(\log^2{n})$ non-leaves in $\cTstar$; and
\item let $X$ and $Y$ be corresponding super-vertices that contain $u$ and $v$ in a weakly consistent partial tree $\pTree$; there is $\leaves{\cTstar}{\lcatreeset{X}{\cTstar}} \cap \leaves{\cTstar}{\lcatreeset{Y}{\cTstar}} \neq \emptyset$.
\end{enumerate}
can contribute to at most an $o(1)$ fraction of the optimal cost. This means that our estimation with non-leaves using \Cref{lem:num-non-leaf-weak-tree} would lead to a good approximation. The formal statement is as follows.

\begin{lemma}
\label{lem:MW-opt-structure}
Let $\pTree$ be an arbitrary partial HC tree that is weakly consistent with the optimal HC tree $\cTstar$ under the Moseley-Wang objective. Define $\Elow(\pTree) \subseteq V\times V$ as the edges such that for any $(u,v)\in \Elow(\pTree)$, there is
\begin{enumerate}[label=\alph*).]
\item $\card{\nonleaves{\cTstar}{\lcatree{u}{v}{\cTstar}}}\leq 50000\cdot \log^2{n}$.
\item Let $X$ and $Y$ be corresponding super-vertices that contain $u$ and $v$ in $\pTree$; there is 
\[\leaves{\cTstar}{\lcatreeset{X}{\cTstar}} \cap \leaves{\cTstar}{\lcatreeset{Y}{\cTstar}} \neq \emptyset.\]
\end{enumerate}
Then, for sufficiently large $n$, the total contribution of revenue from the edges in $\Elow(\pTree)$ is at most $50000\cdot \frac{\log^4{n}}{n}$ fraction of $\OPTmw$, i.e.,
\begin{align*}
\sum_{e=(u,v)\in \Elow(\pTree)} w(e)\cdot \card{\nonleaves{\cTstar}{\lcatree{u}{v}{\cTstar}}} \leq O(\frac{\log^4{n}}{n}) \cdot \OPTmw.
\end{align*}
\end{lemma}
\begin{proof}
For any two super-vertices $X$ and $Y$ in the partial HC tree $\pTree$, by the \emph{weak contraction property}, if $\leaves{\cTstar}{\lcatreeset{X}{\cTstar}} \cap \leaves{\cTstar}{\lcatreeset{Y}{\cTstar}}$ is not empty, the only possible case is to have \emph{inclusion} relationships between the two sets of leaves. Suppose w.log. that $\leaves{\cTstar}{\lcatreeset{X}{\cTstar}}\supseteq \leaves{\cTstar}{\lcatreeset{Y}{\cTstar}}$. 
Let $\tilde{Y}$ be the set of vertices induced by the sibling node of $X$ in $\pTree$, and it is straightforward to see that $Y\subseteq \tilde{Y}$. We further let $Y'$ be the \emph{larger} immediate child of the tree induced by $\tilde{Y}$.

By the conditions of $i).$ $\card{\nonleaves{\cTstar}{\lcatree{u}{v}{\cTstar}}}\leq 50000\cdot \log^2{n}$, $ii).$ $\card{X}\leq 5000\log{n}$, and $iii).$ the inclusion relationship between the leaves, there is
\[\card{\tilde{Y}}\geq n-50000\cdot (\log^2{n}+\log{n}); \qquad \card{Y'}\geq \frac{n-50000\cdot (\log^2{n}+\log{n})}{2}.\]

We now use \Cref{clm:weight-property-internal-nodes} inductively to bound the weights of $w(E_1)$ for the sets of edges $E_1\subseteq \Elow$ that are split by some internal nodes of $X$.
Let $E_2$ be the set of edges that is split in the subtree induced by $Y'$, and let $\tilde{X}$ be vertices that are split by $E_1$ and as the sibling of $\tilde{Y}$. Our induction hypothesis is that
\begin{align*}
w(E_1)< \frac{C \cdot \log{n} \cdot w(E_2)}{n-50000\cdot (\log^2{n}+\log{n})}
\end{align*}
for some absolute constant $C$.

To prove this statement, we first look at the base case. By the size bound on $X$, it is straightforward to see that
\[\card{X}-\card{\tilde{X}}\leq 50000\log{n}.\]
In the base case, we pick edges $E_1 \subseteq \Elow$ such that $E_2$ is split immediately after $E_1$, i.e., $E_1$ are the edges between $\tilde{X}$ and $\tilde{Y}$. Now, we can use \Cref{clm:weight-property-internal-nodes} with $A\gets \tilde{X}$, $B\gets Y'$, and $D\gets \emptyset$ ($w_3=0$) to argue that
\begin{align*}
\frac{w(E_1)}{\card{X}}\leq \frac{w(E_1)}{\card{\tilde{X}}}< \frac{w(E_2)}{\card{Y'}}.
\end{align*}
By the size upper bound of $X$ and the size lower bound of $Y'$, the above inequality implies that
\begin{align*}
\frac{w(E_1)}{50000\log{n}}< \frac{2 \cdot w(E_2)}{n-50000\cdot (\log^2{n}+\log{n})}<\frac{C \cdot \log{n} \cdot w(E_2)}{n-50000\cdot (\log^2{n}+\log{n})}
\end{align*}
for any $C>2$, which proves the base case.

For the inductive step, let us suppose the statement holds until some internal node $z$, and we look into $E_1\subseteq \Elow$ that is induced on $\pa{z}$. We again use \Cref{clm:weight-property-internal-nodes} by letting $A\gets \tilde{X}$, $B\gets Y'$, and $D\gets X\setminus \tilde{X}$. Define $E_3$ as the set of edges that are split between $z$ and $\lcatreeset{\tilde{Y}}{\cTstar}$. By the induction hypothesis, for every node between $z$ and $\lcatreeset{\tilde{Y}}{\cTstar}$, the total induced weights is at most
\begin{align*}
\frac{C \cdot \log{n} \cdot w(E_2)}{n-50000\cdot (\log^2{n}+\log{n})}.
\end{align*}
Furthermore, since $\frac{\card{X}}{\card{\tilde{X}}}\leq \card{X} \leq 50000 \log{n}$, and $z$ and $\pa{\lcatreeset{\tilde{Y}}{\cTstar}}$ are in the same $X$ of $\pTree$, we can bound the total weights in $E_3$ as follows
\begin{align*}
w(E_3)\leq 50000\cdot \log{n}\cdot \frac{C \cdot \log{n} \cdot w(E_2)}{n-50000\cdot (\log^2{n}+\log{n})}.
\end{align*}
Furthermore, by the size upper bound of $X$, we have $\card{A}+\card{D}\leq 50000\cdot \log{n}$ in \Cref{clm:weight-property-internal-nodes}. Combining the above gives us that for sufficiently large $n$, we have
\begin{align*}
\frac{1}{2}\cdot \frac{w(E_1)}{\card{X}}\leq \frac{w(E_1)-(50000 \log{n} + 1)\cdot w(E_3)}{\card{X}} \leq \frac{w(E_1)-(\frac{\card{X}}{\card{\tilde{X}}}+1) \cdot w(E_3)}{\card{X}} < \frac{w(E_2)}{\card{Y'}},
\end{align*}
where the first inequality follows from the fact that $w(E_3)=O(\frac{\log^2{n}}{n})\cdot w(E_2)$, the second inequality follows from that $\card{X}/\card{\tilde{X}}\leq \card{X} \leq 50000 \log{n}$, and the last inequality follows from \Cref{clm:weight-property-internal-nodes}.
Therefore, we can again obtain that 
\begin{align*}
\frac{w(E_1)}{100000\log{n}}< \frac{4 \cdot w(E_2)}{n-50000\cdot (\log^2{n}+\log{n})},
\end{align*}
which means $w(E_1)< \frac{C \cdot \log{n} \cdot w(E_2)}{n-50000\cdot (\log^2{n}+\log{n})}$ for a sufficiently large constant $C$ ($C=400000$ suffices). This concludes our inductive proof for the weight bound of any $E_1 \subseteq \Elow$ in any $X$.

Using $w(E_1)\leq \frac{C\cdot \log{n}}{n-50000\cdot (\log^2{n}+\log{n})} \cdot w(E_2)$ for any $(u,v)\in \Elow$, we note that a trivial lower bound of the optimal revenue is
\[\OPTmw = \Omega(\log{n}\cdot w(E_2)),\]
since this is the revenue induced on the edge set $E_2$ only. On the other hand, since we need the inclusion relationship between the leaves, and by the fact that the number of non-leaves is at most $\log^2{n}$, there are at most $50000(\log^2{n}+\log{n})$ edges in $\Elow$, i.e.,
\[\card{\Elow}\leq 50000(\log^2{n}+\log{n}).\]
Therefore, the total contribution of revenue by $\Elow$ is at most 
\begin{align*}
& \sum_{e=(u,v)\in \Elow(\pTree)} w(e)\cdot \card{\nonleaves{\cTstar}{\lcatree{u}{v}{\cTstar}}} \\
&\leq \sum_{e=(u,v)\in \Elow(\pTree)} w(e)\cdot \log^2{n} \tag{by the number of non-leaves}\\
&\leq \max\{w(E_1)\mid E_1 \subseteq \Elow \}\cdot 50000(\log^2{n}+\log{n}) \cdot \log^2{n} \tag{uniform upper bound}\\
&\leq \frac{C\cdot \log{n}}{n-50000\cdot (\log^2{n}+\log{n})} \cdot w(E_2) \cdot 50000(\log^2{n}+\log{n}) \cdot \log^2{n} \tag{by the relationship between $w(E_2)$ and $w(E_1)$ for $E_1\subseteq \Elow$}\\
&\leq O(\frac{\log^4{n}}{n}\cdot \OPTmw), \tag{by the lower bound of $\OPTmw$}
\end{align*}
as desired. 
\end{proof}

\paragraph{Finalizing the proof of \Cref{thm:hc-alg-MW}. } We have discussed that the algorithm enjoys $\tilde{O}(n^2)$ running time and $O(n^2)$ query efficiency. For the approximation guarantee, note that each edge $(u,v)$ gains a revenue of $w_{u,v}\cdot \card{\nonleaves{\cTstar}{\lcatree{u}{v}{\cTstar}}}$ in the optimal tree $\cTstar$. For edges $(u,v)\not\in \Elow$, let $\Esame$ be the set of edges where $u$ and $v$ are in the same super-vertex, and $\Ediff$ be the set of edges where $u$ and $v$ are in different super-vertices. We now have
\begin{itemize}
\item For edges in $\Esame$, we have $\card{\nonleaves{\cT}{\lcatree{u}{v}{\cT}}}\geq n - 50000\log{n}$ that by \Cref{lem:num-non-leaf-weak-tree}. On the other hand, any vertex pair has at most $n$ non-leaves in the graph. Therefore, we have that
\begin{align*}
\rev_{G\cap \Esame}{(\cT)} &= \sum_{e\in \Esame} w(e)\cdot \card{\nonleaves{\cT}{\lcatree{u}{v}{\cT}}}\\
&\geq \sum_{e\in \Esame} w(e)\cdot n - 50000\log{n}\\
&\geq  \sum_{e\in \Esame} w(e)\cdot n \cdot (1-\frac{50000\log{n}}{n})\\
&\geq (1-\frac{50000\log{n}}{n}) \cdot \rev_{G\cap \Esame}{(\cTstar)}.
\end{align*}
\item For edges in $\Ediff$, we have 
\begin{align*}
\card{\nonleaves{\cT}{\lcatree{u}{v}{\cT}}} & \geq \card{\nonleaves{\cTstar}{\lcatree{u}{v}{\cTstar}}}-50000\log{n} \\
& \geq (1-O(\frac{1}{\log{n}}))\cdot \card{\nonleaves{\cTstar}{\lcatree{u}{v}{\cTstar}}},
\end{align*}
where the first inequality follows from \Cref{lem:num-non-leaf-weak-tree}, and the second inequality is from the fact that $\card{\nonleaves{\cTstar}{\lcatree{u}{v}{\cTstar}}}\geq 50000\cdot \log^2{n}$ for every $(u,v)\not\in \Elow$. As such, we have that
\begin{align*}
\rev_{G\cap \Ediff}{(\cT)} &= \sum_{e\in \Ediff} w(e)\cdot \card{\nonleaves{\cT}{\lcatree{u}{v}{\cT}}}\\
&\geq \paren{1-O(\frac{1}{\log{n}})} \cdot  \sum_{e\in \Ediff} w(e)\cdot \card{\nonleaves{\cTstar}{\lcatree{u}{v}{\cTstar}}}\\
&= \paren{1-O(\frac{1}{\log{n}})} \cdot \rev_{G\cap \Ediff}{(\cTstar)}.
\end{align*}
\end{itemize}
Therefore, by additionally using \Cref{lem:MW-opt-structure}, we have that
\begin{align*}
\rev_{G}(\cT)&\geq \rev_{G\cap \Esame}{(\cT)} + \rev_{G\cap \Ediff}{(\cT)}\\
&\geq (1-O(\frac{1}{\log{n}})) \cdot \rev_{G\cap \Esame}{(\cTstar)} + \rev_{G\cap \Ediff}{(\cTstar)}\\
&\geq (1-O(\frac{1}{\log{n}})) \cdot \paren{\rev_{G\cap \Esame}{(\cTstar)} + \rev_{G\cap \Ediff}{(\cTstar)}}\\
&\geq (1-O(\frac{1}{\log{n}})) \cdot \paren{1-O(\frac{\log^4{n}}{n})} \cdot \rev_{G}(\cTstar) \tag{using \Cref{lem:MW-opt-structure}} \\
&\geq (1-o(1))\cdot \rev_{G}(\cTstar),
\end{align*}
as desired.

\begin{remark}
\label{rmk:MW-obj-strong-tree}
We can observe that if we run \Cref{alg:HC-MW} with the \emph{strongly consistent} partial HC tree, we can get a similar (and even stronger) approximation guarantee, albeit with worse efficiency ($\Otilde(n^3)$ time). Concretely, note that if we use the strongly consistent partial HC tree, the additive error again only happens on $E^{\text{same}}$, and we do \emph{not} need \Cref{lem:MW-opt-structure} to bound the contributions of the edges in $\Elow(\pTree)$ (since $=\emptyset$). In this way, we can further decrease the $o(1)$ term to $O(\frac{\log{n}}{n})$. However, the sacrifice of the running time is too significant, and we skip the details of this algorithm.
\end{remark}

\section{Learning-Augmented Sublinear Algorithms for Hierarchical Clustering}
\label{sec:hc-oracle-sublinear-algs}
In this section, we explore \emph{sublinear} algorithms for hierarchical clustering with the splitting oracle. 
Among these is a \emph{semi-streaming} algorithm, capable of computing an $O(1)$ approximation of Dasgupta's HC objective within \emph{polynomial time}. Additionally, we introduce a PRAM algorithm that achieves a $(1-o(1))$ approximation of the Moseley-Wang objective, utilizing $\tilde{O}(n^2)$ work and $\polylog{n}$ depth. From a technical standpoint, these algorithms represent straightforward extensions of the results outlined in \Cref{thm:hc-alg-MW,thm:constant-poly-time-HC,thm:root-loglog-n4-time-HC}. Despite their simplicity, these algorithms demonstrate the advantages of the splitting oracle in modern sublinear computation models. Specifically, we compare our sublinear algorithms with previous results as follows:
\begin{enumerate}
\item In the streaming setting, \citep{AssadiCLMW22,AgarwalKLP22} designed single-pass streaming algorithms that achieve $\tilde{O}(n)$ memory usage and $O(1)$-approximation to Dasgupta's objective, albeit in \emph{exponential time}. By improving the time efficiency to polynomial time, our streaming result echoes a similar narrative in the offline setting, demonstrating significantly more efficient constructions with the splitting oracle.
\item For the parallel setting, \citep{agarwal2024parallel} (cf. \citep{AgarwalKLP22}) provided parallel algorithms (in the PRAM and the similar MPC settings, see \Cref{sec:pram-mpc-models} for details of these models) for Dasgupta's objective with $\tilde{O}(n^2)$ work and $\polylog{n}$ depth that achieve $\polylog{n}$ approximation. Since the objectives are different, their result is not directly comparable to ours; however, the conceptual message here is still that the splitting oracle is able to significantly improve the approximation guarantee and the efficiency.
\end{enumerate}

We discuss these algorithms and their analysis for the rest of this section.

\subsection{A single-pass semi-streaming HC algorithm for Dasgupta's objective}
\label{subsec:hc-alg-streaming}
We present our algorithm for the semi-streaming algorithm in this part. To begin with, we need to define the \emph{model} for streaming hierarchical clustering with the splitting oracle.

\paragraph{Graph streaming with offline splitting oracle.} We focus on the (dynamic) graph streaming model with the \emph{offline} splitting oracle. In this model, the edges of the graph are inserted and deleted (together with their edge weights), and the algorithm is asked to output an HC tree by the \emph{end} of the stream. Additionally, the algorithm is given an offline splitting oracle $\cO$ \emph{before} the graph stream, and the algorithm is allowed to make unlimited computations before the stream starts. The total memory cost is the total number of bits the algorithm maintains at any point, including those dedicated to edge representation and those generated through offline computations. 

For polynomial-time efficiency, we assume that the stream itself is of $\poly(n)$ length since otherwise, the stream itself will take super-polynomial time to complete. Under the above model and setting, the formal statement of our result is as follows.
\begin{theorem}
\label{thm:hc-das-streaming}
There exists a single-pass (dynamic) streaming algorithm that given a weighted undirected graph $G=(V,E,w)$ in a $\poly(n)$-length dynamic stream and an offline splitting oracle $\cO$, with high probability, in $O(n\cdot \log^{3}{n})$ (bits) space and polynomial time computes a hierarchical clustering tree $\cT$ such that $\cost_{G}(\cT)\leq O(1)\cdot \OPTdas(G)$, where $\OPTdas(G)$ is the cost of the optimal hierarchical clustering tree $\cTstar$, i.e., $\OPTdas(G)=\cost_{G}(\cTstar)$.
\end{theorem}

\begin{proof}
The algorithm is to simply first construct a strong partial tree $\pTree$ with the algorithm of \Cref{def:strongly-consistent-tree} \emph{before} the stream starts, and store edges only inside the same super-vertices during the stream. By the end of the stream, we compute recursive sparsest cuts on the vertices induced on super-vertices of $\pTree$, and output in the same manner of \Cref{alg:HC-Das}. The formal description is as follows.

\begin{algorithm}
\caption{A polynomial-time single-pass semi-streaming algorithm for the Dasgupta's HC objective} \label{alg:HC-streaming-poly-time}
\KwIn{Input graph $G=(V,E,w)$; Splitting oracle $\cO$}
\KwOut{A hierarchical clustering tree $\cT$}
\textbf{Before} the start of the stream, run the strong partial tree approximation algorithm in \Cref{thm:strong-partial-tree} to obtain partial tree $\pTree$\;
\For{each edge $(u,v)$ with during the insertion/deletion stream}{
If $u,v\in X$ for any super-vertex $X$ in $\pTree$, update $(u,v)$ with the same insertion/deletion update\;
Otherwise, ignore the edge\;
}
\For{each super-vertex in $\pTree$ after the stream, run recursive sparsest cut as follows}{
On input vertex set $S$, \emph{exhaustively search} the sparsest cut $(A,B)$ on the induced subgraph $G[S]$.\;
Partition the vertices as $S\rightarrow (A,B)$, and recurse on $G[A]$ and $G[B]$.\;
}
\end{algorithm}


We first prove the time and space efficiency. Essentially, \Cref{alg:HC-streaming-poly-time} computes a strong partial tree $\pTree$ before the stream starts, and only main edges inside each super-vertex of $O(\log{n})$ size. By \Cref{thm:strong-partial-tree}, the pre-processing part takes polynomial ($O(n^3\log{n})$) time and $O(n\log{n})$ space. Furthermore, for each super-vertex, we maintain at most $O(\log^{2}n)$ edges; each edge could have at most $\poly(n)$ updates, which means an additive space of $O(\log{n})$ bits suffices for each edge. Therefore, we record at most $O(\log{n})$ bits for each edge. There are at most $n$ super-vertices, which means the algorithm maintains the information of $O(n\log^2{n})$ edges, which takes $O(n\log^3 {n})$ bits. After the stream, we partition the super-vertices with the edges stored, and write down the rest of the HC tree. The time efficiency of the post-processing part is exactly as \Cref{lem:HC-das-poly-efficiency}.

For the approximation guarantee, note that we are essentially simulating \Cref{alg:HC-Das} with the same input and output guarantees. As such, we can simply use \Cref{lem:HC-das-const-approx} to argue the approximation guarantee.
\end{proof}

\subsection{Parallel HC algorithms for Moseley-Wang Objective}
We now move the PRAM hierarchical clustering algorithm for the Moseley-Wang objective with near-linear $\tilde{O}(n^2)$ work and $\polylog(n)$ depth. The formal statement of the algorithm is as follows.

\begin{theorem}
\label{thm:hc-MW-PRAM}
There exists a PRAM algorithm that given a weighted undirected graph $G=(V,E,w)$ and a splitting oracle $\cO$, with high probability, in $O(n^2 \cdot \polylog{n})$ work and $\polylog{n}$ depth computes a hierarchical clustering tree $\cT$ such that $\rev_{G}(\cT)\geq (1-o(1))\cdot \OPTmw(G)$, where $\OPTmw(G)$ is the revenue of the optimal hierarchical clustering tree $\cTstar$, i.e., $\OPTmw(G)=\rev_{G}(\cTstar)$.
\end{theorem}
\begin{proof}
The algorithm is to run the PRAM weak partial tree algorithm in \Cref{cor:parallel-weak-partial-tree} to obtain $\pTree$, and arbitrarily partition the vertices in the super-vertices of $\pTree$. The formal description of the algorithm is as follows.

\begin{algorithm}
\caption{A near-linear work, poly-logarithmic depth PRAM algorithm for the Moseley-Wang HC objective} \label{alg:HC-MW-PRAM}
\KwIn{Input graph $G=(V,E,w)$; Splitting oracle $\cO$}
\KwOut{A hierarchical clustering tree $\cT$}
Run the PRAM weak partial tree approximation algorithm in \Cref{cor:parallel-weak-partial-tree} to obtain partial tree $\pTree$\;
\For{each super-vertex in $\pTree$}{
Partition the leaves \emph{arbitrarily} to obtain an HC tree $\cT$\; 
}
\end{algorithm}

By \Cref{cor:parallel-weak-partial-tree}, the first step that computes the weak partial tree $\pTree$ only takes $\tilde{O}(n^2)$ work and $\polylog{n}$ depth. For the second step, since the partition is arbitrary, we can simply perform arbitrary balanced cuts on the vertices for each super-vertex $X$. For an individual $X$, this procedure can be done in $O(\log{n}\cdot \log\log{n})$ work and $O(\log\log{n})$ depth. By accounting for all the super-vertices, we blow up the work by at most an $O(n)$ factor and the depth remains the same since we can partition super-vertices in parallel. Therefore, the second step takes $O(n\cdot \log^2{n})$ work and $O(\log\log{n})$ depth. Therefore, the entire procedure takes $\tilde{O}(n^2)$ work and $\polylog{n}$ depth.

The output of \Cref{thm:hc-MW-PRAM} follows exactly the same rules of \Cref{alg:HC-MW}; therefore, the approximation guarantee follows from \Cref{thm:hc-alg-MW}.
\end{proof}

\begin{remark}
By the reduction of \Cref{prop:PRAM-to-MPC}, the result of \Cref{thm:hc-MW-PRAM} also implies a fully-scalable Massively Parallel Computation (MPC) algorithm that computes the HC tree $\cT$ with $(1-o(1))\cdot \OPTmw(G)$ revenue in $\tilde{O}(n^2)$ total memory and $\polylog{n}$ depth. The memory on each machine here is allowed to be $O(n^{\alpha})$ for any $\alpha\in (0,1)$, and we use $\tilde{O}(n^{2-\alpha})$ total machines in the MPC algorithm.
\end{remark}

\section{Preliminaries for Partial HC Tree Algorithms}
\label{sec:partial-tree-prelims}
We will discuss the construction of partial HC trees for the rest of this paper (\Cref{sec:partial-tree-prelims,sec:strong-partial-tree,sec:weak-partial-tree}). In this section, we first introduce the several notions that are essential in the \emph{analysis} of the partial HC tree algorithms. We use these notions extensively in the analysis of both \Cref{thm:strong-partial-tree} (\Cref{sec:strong-partial-tree}) and \Cref{thm:weak-partial-tree} (\Cref{sec:weak-partial-tree}). 

In the high-level overview, we slightly abused the notation to use $V$ to denote the subset of vertices. In our formal description of the algorithms, we will use $\Vp \subseteq V$ as the set of vertices of the current recursion step, and use $\np=\card{\Vp}$ as the size of the set. 

\subsection{Composable vertex sets and restricted subtrees}
To continue, we first introduce the notions of \emph{composable vertex sets} for a graph and the HC tree restricted to the sets.

\begin{definition}[Composable vertex sets]
\label{def:comp-set}
Let $G=(V,E)$ be a $n$-vertex graph and $\cT$ be a hierarchical clustering of $G$. For a subset $\Vp\subseteq V$, we say $\Vp$ is a \emph{composable (vertex) set} of $(G, \cT)$ if $\Vp$ can be written in a disjoint union of the leaves of \emph{maximal trees}, i.e., a union of vertices $\Vp=\cup_{i} \Vp_{i}$, such that \emph{all} $\Vp_{i}$ satisfies that 
\[\leaves{\cT}{\lcatreeset{\Vp_{i}}{\cT}} = \Vp_{i}.\]
In the special case, we say $\Vp$ is a \emph{single composable (vertex) set} if there is \emph{only one} such maximal tree, i.e., $\leaves{\cT}{\lcatreeset{\Vp}{\cT}} = \Vp$.
\end{definition}

In other words, if a vertex set $\Vp$ is composable, it means if we only look at the leaves of $\Vp$, they still form subtrees of $\cT$. An illustration of composable sets can be found in \Cref{fig:composable-sets}.

\begin{figure}[!htb]
	\centering
	\includegraphics[scale=0.6]{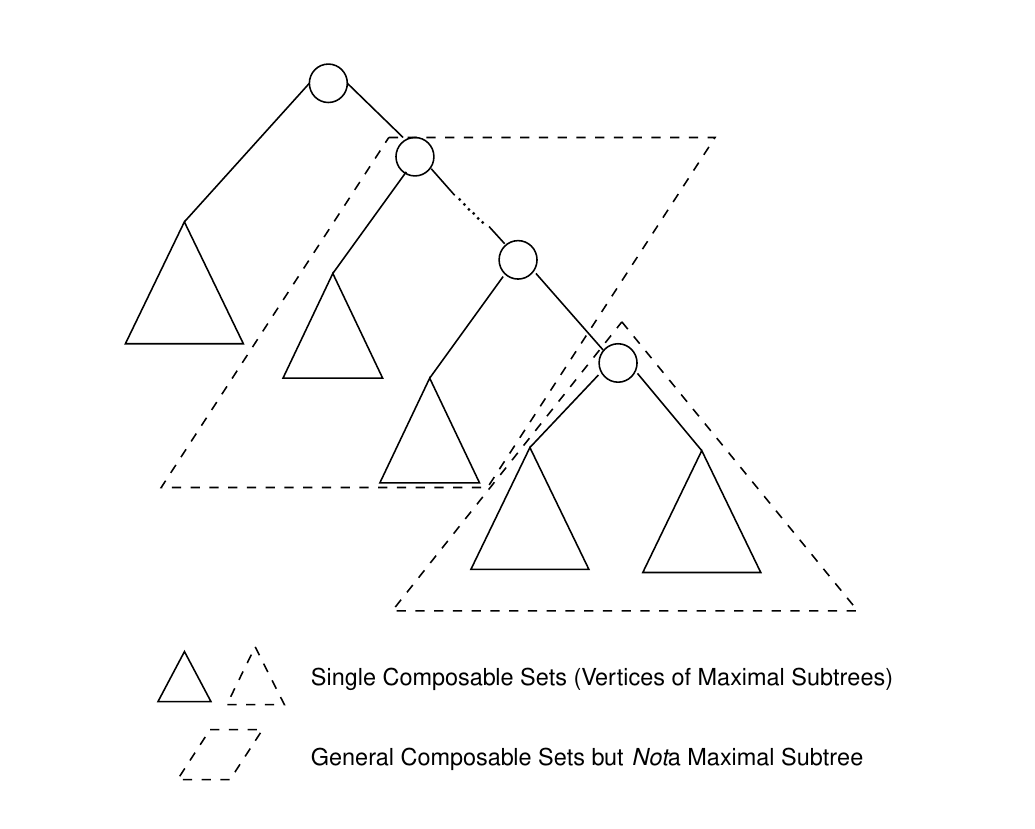}
	\caption{\label{fig:composable-sets} An illustration of the composable vertex sets and maximal trees as in \Cref{def:comp-set}.}
\end{figure}

We can define the HC tree restricted to composable subsets of vertices as follows.

\FloatBarrier
\begin{definition}[Hierarchical clustering tree restricted to subset]
\label{def:restrict-tree}
Let $G=(V,E)$ be a $n$-vertex graph and $\cT$ be a hierarchical clustering of $G$. For any \emph{composable} subset $S\subseteq V$ such that $\card{S}\geq 2$, we call $\cT(S)$ as \emph{$\cT$ restricted to $S$} if $\cT(S)$ is a new binary tree constructed with the following process
\begin{tbox}
\begin{enumerate}
\item Extract $\cT'$ from $\cT$ by taking the induced subtree of the internal node $\lcatreeset{S}{\cT}$.
\item Remove all subtrees whose leaves contain \emph{only} vertices \emph{not} in $S$ from $\cT'$.
\item If there exists an internal node $x$ that only has \emph{a single} child, contract $x$ and its child to one node. Repeat until all internal node has two children nodes.
\end{enumerate}
\end{tbox}
\end{definition}
\FloatBarrier

Note that the algorithm in \Cref{def:restrict-tree} is a \emph{thought process}, and it only serves the purpose of analysis. The third line eventually terminates since there exists at least one subtree with two leaves for any \emph{composable} set $S$ with at least two vertices. An illustration of HC trees restricted to subsets can be shown in \Cref{fig:HC-restricted-to-subset}.
\begin{figure}[!htb]
	\centering
	\includegraphics[scale=0.6]{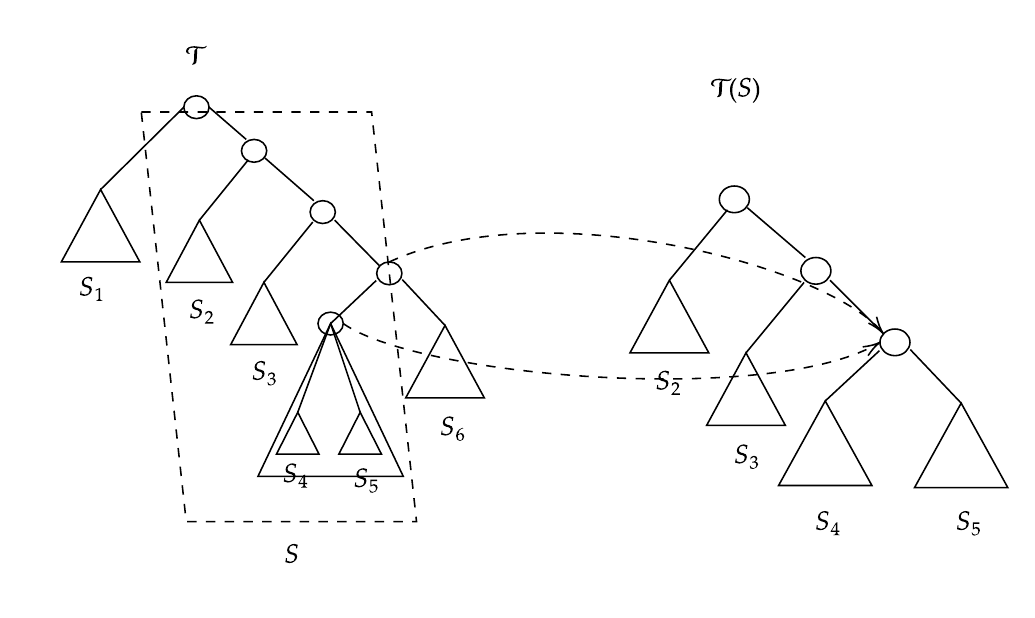}
	\caption{\label{fig:HC-restricted-to-subset} An illustration of the HC tree $\cT$ to be restricted on a subset of vertices $S$.}
\end{figure}

We now provide the following observation that the new tree $\cT(S)$ preserves the \emph{relative order} of split away vertices of $\cT$. 
\begin{observation}
\label{obs:split-away-preserv}
For any triplet of vertices $(u,v,w)$, the orders of split-away for $(u,v,w)$ are the same in $\cT(S)$ and $\cT$, i.e., $w$ splits away from $(u,v)$ in $\cT(S)$ \emph{if and only if} $w$ splits away from $(u,v)$ in $\cT$.
\end{observation}
\begin{proof}
We first observe that the orders of split away for any triplet $(u,v,w)$ are the same in $\cT$ and $\cT'$ since $\cT'$ is a \emph{subtree} of $\cT$. Furthermore, since $S$ is a \emph{conposeable set}, every subtree we remove from $\cT'$ is necessarily a \emph{maximal} subtree, i.e., fix the removed vertex set $S$, the leaves of $\lcatreeset{S}{\cT'}$ is $S$ itself.

Let $x$ be the lowest common ancestor of $(u,v)$ in $\cT'$, and $x'$ be the lowest common ancestor of $(u,v,w)$ in $\cT'$. By our definition, we have $\levelr{\cT'}{x'}>\levelr{\cT'}{x}$. In $\cT(S)$, if none of $x$ and $x'$ is contracted in line 3, then the order of splits away trivially remains the same as in $\cT'$. 

On the other hand, if at least one of $x$ and $x'$ is contracted, and let the new internal nodes be $y$ and $y'$, we claim that there is
still $\levelr{\cT(S)}{y'}>\levelr{\cT(S)}{y}$. This is simply because every removed tree is a maximal tree, and the only case that $y$ and $y'$ are merged is that the induced vertices of $y'$ become empty, which contradicts the fact that $(u,v)$ is not removed. Therefore, the split order between $(w,u,v)$ is the same as in $\cT'$, which in turn is the same in $\cT$. 
\end{proof}

\subsection{Small-tree splitting order}
We now introduce the following notion of \emph{small-tree split order}, which we frequently use in our analysis.
\begin{definition}[Small-tree split order]
\label{def:small-tree-split}
Consider any subset $S\subseteq V$ and the hierarchical clustering tree $\cTstar(S)$ restricted to $S$, and consider the following process that divides the \emph{leaves} of $S$:
\begin{tbox}
\begin{enumerate}
\item Starting from the root $r$, let the split be $S \rightarrow (S^1_{l}, S^1_{r})$, and assume w.log. $\card{S^1_{l}}\leq \card{S^1_{r}}$. Let $\Vsmall_{1}$ be the vertices (leaves) induced by the ``smaller subtree'' $S^1_{l}$.
\item Starting from the lowest common ancestor of $S^1_{r}$, recursively define $\Vsmall_{\ell}$ as the vertices contained in the ``smaller subtree'' of level $\ell$ (by splitting from the ``larger subtree'' of level $\ell-1$). 
\end{enumerate}
\end{tbox}
\noindent
For the convenience of notation, we define $\Vsmall_{\ell}=\emptyset$ when $\ell$ is larger than the depth of the tree. For any integer $N \in [0, \nh]$, we can define the \emph{first $N$ vertices in the small-tree split order} by taking the first $N$ vertices in the order of $\Vsmall_{1}$, $\Vsmall_{2}$, etc.. Inside each set $\Vsmall_{i}$, we pick vertices in an arbitrarily \emph{fixed} order. We use the notation $\Vbefore{N}$ to denote the \emph{first} $N$ vertices in the small-tree split order, and we use $\Vafter{N}$ to denote the \emph{last} $N$ vertices in the small-tree split order. 
\end{definition}

Intuitively, we can think of the small-tree split order as we always write the ``smaller'' tree on the left-hand side, and recurse on the ``larger'' tree for this procedure, then take leaves from the left-most vertices. Note that \emph{every} vertex has to belong to $\Vsmall_{i}$ for \emph{some} $i$. An illustration can be found in \Cref{fig:small-tree-split-order}. 

Based on the notion of small-tree split order, we can define the notion of \emph{induced leaves} of the first or last $N$ vertices in the small-tree split order as follows.

\begin{definition}[Induced leaves of the $N$ first split vertices]
\label{def:small-tree-induced-leaves}
Let $\cTstar(S)$ and $S$ be the hierarchical clustering tree and the set of leaves. For any integers $N \in [0, \nh]$, we define the induced leaves of $\Vbefore{N}$ as the union of the $\cup_{i=1}^{\ell} \Vsmall_{i}$, where $\ell$ is the \emph{maximum} level such that $\Vsmall_{\ell}$ contains at least one vertex in $\Vbefore{N}$.
\end{definition}
\noindent
An illustration of the induced leaves in \Cref{def:small-tree-induced-leaves} can be found in \Cref{fig:small-tree-induced-leaves}.

\begin{figure*}[!htb]
\centering
\begin{subfigure}[t]{0.48\textwidth}
  \includegraphics[scale=0.40]{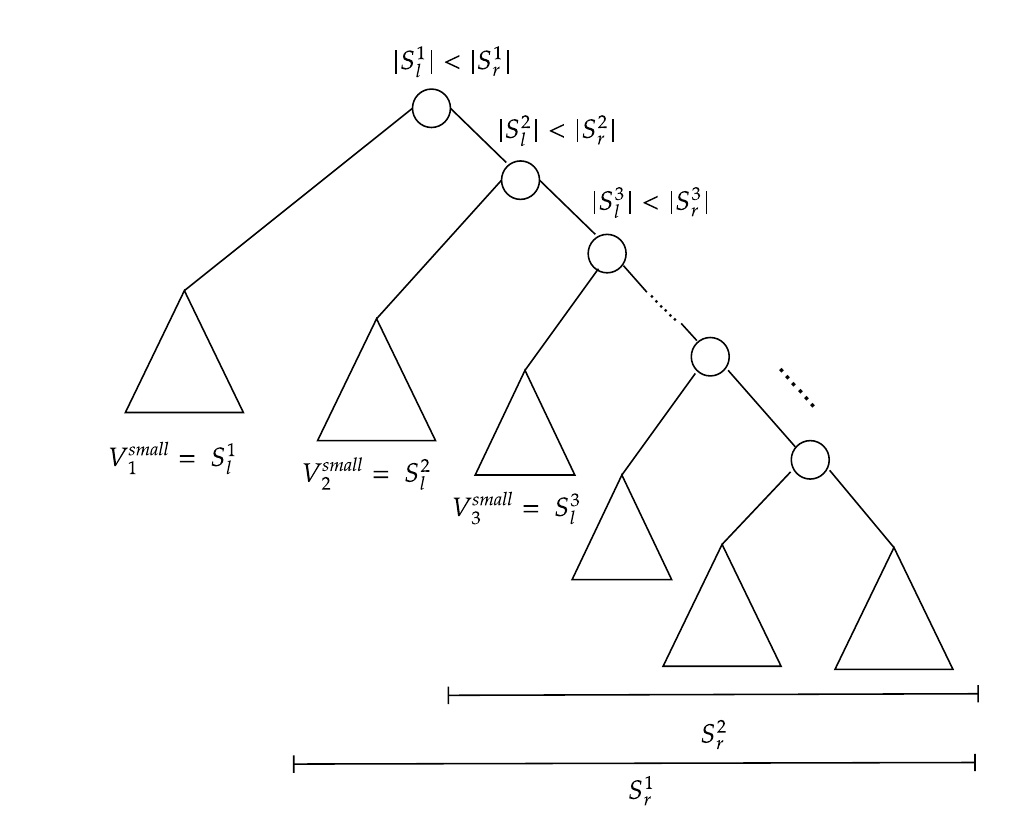}
  \label{fig:small-tree-split-order}
\end{subfigure}
\begin{subfigure}[t]{0.48\textwidth}
  \includegraphics[scale=0.40]{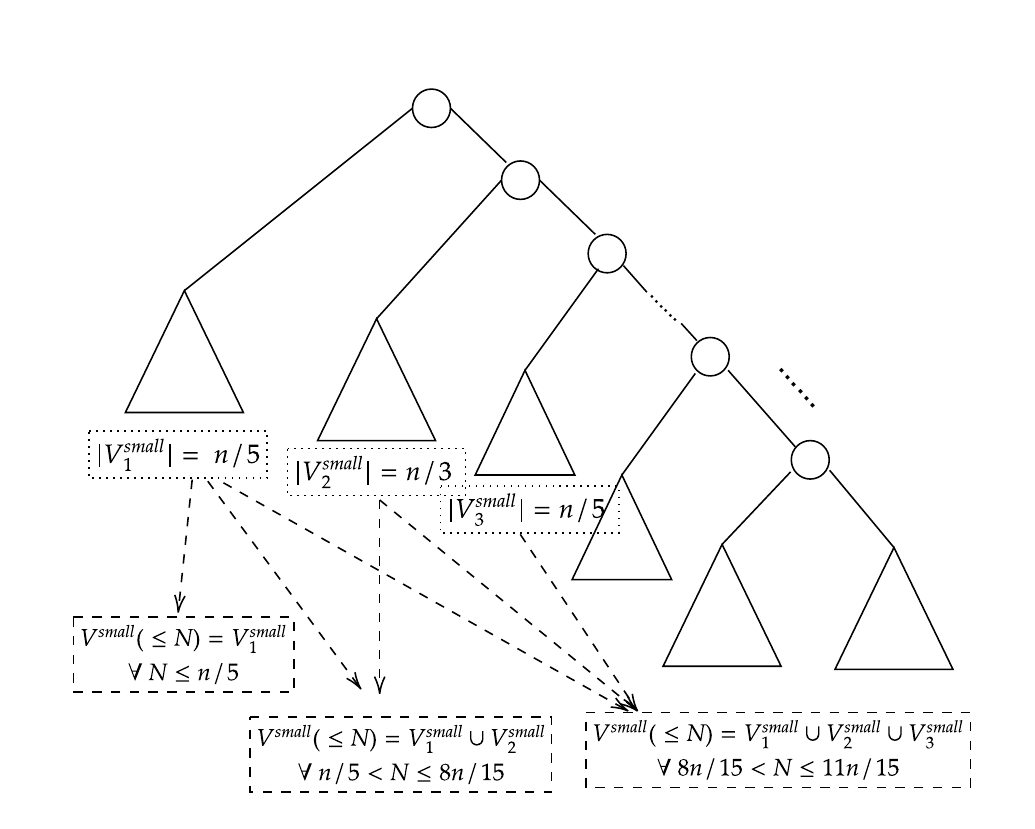}
  \label{fig:small-tree-induced-leaves}
\end{subfigure}
\caption{An illustration of the notion of small-tree split order \Cref{def:small-tree-split} and the Induced leaves of the $N$ first split vertices \Cref{def:small-tree-induced-leaves}.}
\label{fig:small-tree-split-technical-tools}
\end{figure*}

\FloatBarrier

\section{The Algorithm for the Strongly Consistent Partial HC Tree: Proof of \texorpdfstring{\Cref{thm:strong-partial-tree}}{strongp}}
\label{sec:strong-partial-tree}
We now give an algorithm for partial trees that are strongly consistent with the optimal HC tree $\cT^*$. We first remind the readers of our main result on the algorithm for strongly consistent partial HC trees as follows.

\thmstrongpartialtree*

We refer the readers to \Cref{sec:tech-overview} for a high-level overview of the algorithm. We directly give the algorithm and the analysis in this section.

\paragraph{The algorithm.} We introduce our main algorithm for the construction of strongly consistent partial HC trees. To begin with, we give a `helper function' \texttt{counterpart-tester-strong} as follows. As we have discussed in the high-level overview, this function tests the `sibling vertices' of a small tree that is `early enough' in the small-tree splitting order.

\begin{algorithm}
\caption{$\testalgstrong{u}{t}{\textnormal{threshold}_1}{\textnormal{threshold}_2}{\Vinput}$: an algorithm to test whether $t$ is among the ``counterpart'' of $u$} \label{alg:tester-strong}
\KwIn{Vertex set $\Vinput$; the splitting oracle $\cO$; baseline vertex $u$; test vertex $v$; $\text{threshold}_1\in (0,\np)$, $\text{threshold}_2\in (0,\np)$}
\KwOut{Whether $v$ is a ``counterpart'' of $u$}
Initialize counters $c_1\gets 0$, $c_2 \gets 0$\;
\For{Every vertex $t\in \Vinput$}{
If $v$ splits away from $(u,t)$, increase $c_1$ by 1\;
If $t$ splits away from $(u,v)$, increase $c_2$ by 1\;
}
\If{$c_1\leq \textnormal{threshold}_1$ and $c_2 \leq \textnormal{threshold}_2$}{
Return ``$v$ is a counterpart of $u$''.
}
\end{algorithm}
\FloatBarrier

Our main algorithm for the strongly consistent partial HC tree construction is as follows.
\begin{algorithm}
\caption{$\strongptalg{\Vp}{\cO}$: an algorithm for strongly consistent partial tree} \label{alg:strong-partial-tree}
\KwIn{Vertex set $\Vp$ of size $\np$; the splitting oracle $\cO$; parameter $\eps<1$ a sufficiently small constant}
\KwOut{A partial tree $\cT_{\Vp}$ that is strongly consistent with $\cTstar(\Vp)$}
\If{$\card{\Vp}\leq 50000 \log{n}$}{
Return a super-vertex\;
}
\For{Each $u \in \Vp$}{
Initialize $T \gets \emptyset$\;
Sample a set $S$ of $s=20\log{n}/\eps^2$ vertices from $\Vp$\;
\For{$v \in \Vp$}{
\label{line:test-good-T-new} Run $\testalgstrong{u}{v}{(3/5-\eps)\cdot s}{(1/6-\eps)\cdot s}{S}$\; 
Add $v$ to $T$ if $v$ is a ``counterpart'' of $u$\;
}
Record the size $\card{T}$ for this choice of $u$\;
}
\label{line:root-cut} Pick $(T^*, \Vp\setminus T^*)$ such that $T^*\gets T$ is with the largest size (breaking ties arbitrarily)\; 
Recursively call
\begin{align*}
\cT_{\Vp\setminus T^*} \gets \strongptalg{\Vp\setminus T^*}{\cO} \qquad \cT_{T^*} \gets \strongptalg{T^*}{\cO}.
\end{align*}

Connect the two trees with a common ancestor as the root\;

\end{algorithm}
\FloatBarrier

We first observe that the algorithm takes at most$\tilde{O}(n^3)$ queries to $\cO$ and $\tilde{O}(n^3/\eps^2)$ time. The formal statement and analysis are as follows. 
\begin{lemma}
\label{lem:strong-tree-run-time}
The algorithm $\strongptalg{V}{\cO}$ takes at most $O(n^3)$ queries to $\cO$ and $\tilde{O}(n^3\log{n}/\eps^2)$ running time.
\end{lemma}
\begin{proof}
The $O(n^3)$ query upper bound is trivial since there are at most $\binom{n}{3}$ such comparisons, and the algorithm can store the answers for reuse. For the running time, we claim that each recursive call on $\strongptalg{\Vp}{\cO}$ with $\card{\Vp}=\np$ vertices takes at most $O(\np^2\log{\np})$ time. To see this, note that each call of Line~\ref{line:test-good-T-new} takes at most $O(\log{n}/\eps^2)$ time, and there are at most $O(\np^2)$ calls of the \texttt{counterpart-tester-strong} function. By \Cref{fct:num-splits-tree}, there are at most $O(n)$ such recursion calls induced on internal nodes, which results in at most $O(n^3\log{n}/\eps^2)$ running time.
\end{proof}

We now prove the correctness of the algorithm. To this end, we first establish the split-away relationships between the vertex set $\Vp$ and the sampled set $S$. 
\begin{lemma}
    \label{lem:sampling-approx-counterpart-test}
    Consider quantities $c_1$, $c_2$, $\tilde{c}_1$, and $\tilde{c}_2$ obtained by the following processes of running \Cref{alg:tester-strong}:
    \begin{itemize}
        \item Let $c_1$ and $c_2$ be the counter returned by $\testalgstrong{u}{v}{(3/5-\eps)\cdot s}{(1/6-\eps)\cdot s}{S}$ (i.e., by running line~\ref{line:test-good-T-new}). 
        \item Let $\tilde{c}_1$ and $\tilde{c}_2$ be the counter obtained by running $\testalgstrong{u}{v}{3\np/5}{\np/6}{\Vp}$. 
    \end{itemize}
    Then, with high probability, the following statements are true.
    \begin{itemize}
        \item If $\tilde{c}_1\geq 3\np/5 $ and $\tilde{c}_2 \geq \np/6\cdot $, then $c_1\geq (3/5-\eps)\cdot s$ and $c_2 \geq (1/6-\eps)\cdot s$.
        \item If $\tilde{c}_1< (3/5-2\eps) \cdot \np$ and $\tilde{c}_2 < (1/6-2\eps) \cdot \np$, then $c_1< (3/5-\eps)\cdot s$ and $c_2 < (1/6-\eps)\cdot s$.
    \end{itemize}
\end{lemma}
\begin{proof}
    The lemma is a direct application of Chernoff bound, and we prove the quantities for $c_1$ and $\tilde{c}_1$ only since the relationships between $c_2$ and $\tilde{c}_2$ follow from the same argument. For each fixed pair of $(u,v)$ as the inputs to \Cref{alg:tester-strong}, let $\tilde{T} \subset \Vp$ be the set of vertices in $\Vp$ reporting ``$v$ splits away from $(u,t)$''. It is clear that $\card{\tilde{T}}=\tilde{c}_1$. Furthermore, define $X_t$ as the indicator random variable for the sampled vertex $t\in \tilde{T} \cap S$ and $X$ as the total number of answers reporting ``$v$ splits away from $(u,t)$'' in $S$. We have $\expect{X}=\sum_{x\in S} \Pr(X_t = 1) = \frac{s\cdot \tilde{c}_1}{\np}$.
    
    If $\tilde{c}_1\geq 3\np/5$, we have that $\expect{X}\geq \frac{3s}{5}$, and $X$ is summation of independent indicator random variables. Therefore, by Chernoff bound we have
    \begin{align*}
        \Pr\paren{X<(\frac{3}{5}-\eps)\cdot s}\leq \Pr\paren{X-\expect{X}<\eps\cdot s}&\leq \exp\paren{-2\cdot \frac{\eps^2 s^2}{s}} \leq \exp\paren{-40 \log{n}} \leq \frac{1}{10}\cdot \frac{1}{n^3},
    \end{align*}
    where the third inequality is by the choice that $s=20\log{n}/\eps^2$. On the other hand, if $\tilde{c}_1 < (3/5-2\eps)\cdot \np$, we have that $\expect{X}<(3/5-2\eps)\cdot s$. Therefore, by Chernoff bound, we have
    \begin{align*}
        \Pr\paren{X\geq \frac{3s}{5}}\leq \Pr\paren{X-\expect{X}\geq \eps\cdot s}&\leq \exp\paren{-2\cdot \frac{\eps^2 s^2}{s}} \leq \exp\paren{-40 \log{n}} \leq \frac{1}{10}\cdot \frac{1}{n^3},
    \end{align*}
    where the third inequality is by the choice that $s=20\log{n}/\eps^2$.
\end{proof}

Using \Cref{lem:sampling-approx-counterpart-test}, we now present the following key lemma for the behavior of the ``counterpart-test'' of \Cref{alg:strong-partial-tree}.
\begin{lemma}
\label{lem:counterpart-test-strong}
For any composable $\Vp$ such that $\card{\Vp} \geq 50000 \log{n}$, let $\ell^*$ be the maximal level that $\Vsmall_{\ellstar}$ contains a vertex in $\Vbefore{\np/5}$. With high probability, Line~\ref{line:test-good-T-new} in \Cref{alg:strong-partial-tree} satisfies the following properties.
\begin{enumerate}[label=\alph*).,leftmargin=25pt]
\item \label{line:split-balance-new-later} For any vertex $u\in \cup_{i=\ellstar+1}^{\infty} \Vsmall_{i}$, there are
\begin{itemize}
\item No vertex $v \in \cup_{i=\ellstar+1}^{\infty} \Vsmall_{i}$ can be added to $T$ by Line~\ref{line:test-good-T-new} of \Cref{alg:strong-partial-tree}.
\item No vertex $v \in \cup_{i=1}^{\ellstar-1} \Vsmall_{i}$ can be added to $T$ by Line~\ref{line:test-good-T-new} of \Cref{alg:strong-partial-tree}.
\end{itemize}
\item \label{line:split-balance-new-middle} For every level $\ell\leq \ellstar$ and vertices $u\in \Vsmall_{\ell}$, with high probability, there are
\begin{itemize}
\item No vertex $v\in \Vsmall_{\ell}$ can be added to $T$ by Line~\ref{line:test-good-T-new} of \Cref{alg:strong-partial-tree}.
\item No vertex $v\in \Vsmall_{k}$ for $k<\ell$ can be added to $T$ by Line~\ref{line:test-good-T-new} of \Cref{alg:strong-partial-tree}.
\end{itemize}
\item \label{line:split-balance-new-earlier} There exists a $\tilde{\ell}\leq \ellstar$ such that $\card{\cup_{i=1}^{\tilde{\ell}} \Vsmall_{i}}\geq \frac{\np}{25}$, and for any $\ell\leq \tilde{\ell}$ and any vertex $u\in \Vsmall_{\ell}$, with high probability, all vertices $v\in \Vsmall_{k}$ for $k>\ell$ are added to $T$ by Line~\ref{line:test-good-T-new} of \Cref{alg:strong-partial-tree}.
\end{enumerate}
\end{lemma}
\begin{proof}
By \Cref{lem:sampling-approx-counterpart-test}, we only need to prove the split-away properties on $\Vp$, and the split-away properties on $S$ follows.
Let $\ellstar$ be the maximal level that the $\Vsmall_{\ellstar}$ contains a vertex in $\Vbefore{\np/5}$, and suppose the split on this level results in $(S^{\ellstar}_{l}, S^{\ellstar}_{r})$. We assume w.log. that $\card{S^{\ellstar}_{l}}\leq \card{S^{\ellstar}_{r}}$ so that $S^{\ellstar}_{l}$ becomes $\Vsmall_{\ellstar}$. We first observe a structural property.
\begin{observation}
\label{obs:split-structure}
The size of the vertices in $\Vh\setminus \Vbefore{\np/5}$ (and equivalently $\cup_{i=\ellstar+1}^{\infty} \Vsmall_{i}$) satisfies 
\[\frac{2\np}{5} < \card{\cup_{i=\ellstar+1}^{\infty} \Vsmall_{i}} \leq \frac{4\np}{5}.\]
\end{observation}
\begin{proof}
The upper bound follows from $\card{\Vbefore{\np/5}}\geq \np/5$, as otherwise $\Vbefore{\np/5}$ will not have enough vertices. For the lower bound, note that by level $\ellstar-1$, the size there are still \emph{more than} $\frac{4\np}{5}$ vertices remain to be decided. Also, note that  $\cup_{\ellstar+1}^{\infty} \Vsmall_{i}$ collectively forms the `larger subtree' w.r.t. $\Vsmall_{\ellstar}$. As such, its size has to be \emph{more than} $\frac{1}{2}\cdot \frac{4\np}{5} = \frac{2\np}{5}$. 

\end{proof}
\noindent
Note that by the definition of small trees, we also have
$\card{\Vsmall_{\ell}}\leq \frac{\np}{2}$
for any $\ell$.
We now proceed one-by-one to the proofs of \Cref{line:split-balance-new-later}, \Cref{line:split-balance-new-middle}, and \Cref{line:split-balance-new-earlier}. 

\paragraph{Proof of \Cref{line:split-balance-new-later}} Fix any vertex $u\in \cup_{i=\ellstar+1}^{\infty} \Vsmall_{i}$. For the first statement, note that for every $v\in \cup_{i=\ellstar+1}^{\infty} \Vsmall_{i}$, all the vertices in $\cup_{i=1}^{\ellstar} \Vsmall_{i}$ are \emph{splitting away} from $(u,v)$. Define $X_{u,v,t}$ as the indicator random variable answered by $\cO$ as ``$t$ splits away from $(u,v)$'', and define $X_{u,v}=\sum_{t\in\Vp}X_{u,v,t}$ as the total number of vertex split away from $(u,v)$. By \Cref{obs:split-structure}, the expected number of split away reported by oracle $\cO$ is at least
\begin{align*}
\expect{X_{u,v}}\geq \frac{9}{10}\cdot \frac{\np}{5} = \frac{9}{50}\cdot \np.
\end{align*}
Note also that $X_{u,v}$ is a summation of $0/1$ independent random variables. As such, we can apply Chernoff bound to get that
\begin{align*}
\Pr\paren{X_{u,v}\leq \frac{\np}{6}} & = \Pr\paren{X_{u,v}\leq \frac{25}{27}\cdot \expect{X_{u,v}}} \\
& \leq \exp\paren{-\frac{(2/27)^2}{3}\cdot \expect{X_{u,v}}}\\
& \leq \frac{1}{10}\cdot \frac{1}{n^3}.
\tag{$\frac{9}{50}\cdot \np \geq 9000 \log{n}$ using the lower bound on the size of $\Vp$}
\end{align*}
For the second statement, note that for any $t\in \cup_{i=\ellstar}^{\infty} \Vsmall_{i}$, $v$ actually splits away from $(u,t)$. As such, there is at least $\frac{4\np}{5}$ vertices such that ``$v$ splits away from $(u,t)$'' -- this is because the LCA between $(u,v)$ has to be higher than the LCA between $(u,t)$ for any $t\in \cup_{i=\ellstar}^{\infty} \Vsmall_{i}$. As such, define $Y_{u,v}$ as the number of total answers of ``$v$ splits away from $(u,t)$'' by $\cO$, we again have
\begin{align*}
\expect{Y_{u,v}} \geq \frac{9}{10}\cdot \frac{4\np}{5}=\frac{18}{25}\cdot \np.
\end{align*}
As such, we again have
\begin{align*}
\Pr\paren{Y_{u,v}\leq \frac{3\np}{5}} &= Pr\paren{Y_{u,v}\leq \frac{15}{18}\cdot \expect{X_{u,v}}}\\
&\leq \frac{1}{10}\cdot \frac{1}{n^3}. \tag{$\frac{9}{50}\cdot \np \geq 9000 \log{n}$ using the lower bound on the size of $\Vp$}
\end{align*}
By a union bound over at most $n$ vertices, no vertex $v\in \cup_{i=1}^{\ellstar-1} \Vsmall_{i}$ or $v\in \cup_{i=\ellstar+1}^{\infty} \Vsmall_{i}$ can pass the test and be recorded as a `counterpart'. Furthermore, by \Cref{lem:sampling-approx-counterpart-test}, if $X_{u,v}>\np/6$ and $Y_{u,v}>3\np/5$, then such $u$ cannot pass the test on the sampled set $S$, which is as desired.
This part of the proof can be visualized as in \Cref{fig:counterpart-test-later}.

\paragraph{Proof of \Cref{line:split-balance-new-middle}}
We now show that vertices $v \in \cup_{i=1}^{\ell} {\Vsmall_{i}}$ cannot pass the ``counterpart'' test. For vertices in $v \in \Vsmall_{\ell}$, we claim that there are too many vertices $t$ that split away from $(u,v)$. To see this, let us define $A_{u,v}$ as the total number of answers of ``$t$ splits away from $(u,v)$'' answered by $\cO$. Note that for any vertex $t \in \cup_{i=\ellstar+1}^{\infty} \Vsmall_{i}$, the answer is always ``yes'' since $\ell^{*}\geq \tilde{\ell}\geq \ell$. Therefore, we can lower bound the expectation of $A_{u,v}$ with $\expect{A_{u,v}}\geq \frac{2}{5}\np\cdot \frac{9}{10} \geq \frac{9}{25}\cdot \np$. Therefore, we can apply Chernoff bound to get
\begin{align*}
\Pr\paren{A_{u,v}\leq \frac{\np}{6}} & \leq \exp\paren{-\frac{(1/2)^2}{3}\cdot \expect{A_{u,v}}}\leq \frac{1}{10}\cdot \frac{1}{n^3}.
\end{align*}
We now turn to the vertices that are in $\cup_{i=1}^{\ell-1} {\Vsmall_{i}}$. Note that by definition, there is $\card{\cup_{i=\ell}^{\infty} {\Vsmall_{i}}} \geq \card{\cup_{i=\ellstar}^{\infty} {\Vsmall_{i}}} \geq \frac{4}{5}\cdot \np$. For any $t \in \cup_{i=\ellstar}^{\infty} {\Vsmall_{i}}$, $v$ splits away from $(u,t)$. As such, define $B_{u,v}$ as the number of answers by $\cO$ with ``$v$ splits away from $(u,t)$'', there is $\expect{B_{u,v}}\geq \frac{18}{25}\cdot \np$. Once again, by applying Chernoff bound, there is 
\begin{align*}
\Pr\paren{B_{u,v}\leq \frac{3\np}{5}} & \leq \exp\paren{-\frac{(1/5)^2}{3}\cdot \expect{B_{u,v}}}\leq \frac{1}{10}\cdot \frac{1}{n^3}.
\end{align*}
Applying a union bound over the cases and all $\np\leq n$ vertices and using \Cref{lem:sampling-approx-counterpart-test} would lead to the desired statements.
This part of the proof can be visualized as in \Cref{fig:counterpart-test-middle}.

\paragraph{Proof of \Cref{line:split-balance-new-earlier}} 
We fix $\tilde{\ell}$ as the minimum level such that the union of $\cup_{i=1}^{\tilde{\ell}}\Vsmall_{i}$ has at least $\np/25$ vertices. Let $\ell\leq \tilde{\ell}$, and for any vertex $v$ in $\Vsmall_{k}$ for $k>\ell$, define the following two random variables: $X_{u,v}$ as the number of answers that ``$t$ splits away from $(u,v)$'', and $Y_{u,v}$ as the number of answers that ``$v$ splits away from $(u,t)$''. We shall show that both terms are not large and $v$ can always pass the test. (Note that $X$ and $Y$ are overloaded and are \emph{not} related to their meaning in the proof of Line~\ref{line:split-balance-new-middle}.)



Note that by definition, we have $\card{\cup_{i=1}^{\tilde{\ell}-1} \Vsmall_{i}}<\frac{\np}{25}$. Since $\ell\leq \tilde{\ell}$, for any vertex $u \in \cup_{i=1}^{\ell}\Vsmall_{i}$, only vertices $t\in \cup_{i=1}^{\ell-1} \Vsmall_{i}$ are \emph{actually} splitting away from $(u,v)$ for $v \in \cup_{i=\ell}^{\infty} \Vsmall_{i}$. As such, we have 
$\expect{X_{u,v}}\leq (\frac{1}{10}+\frac{1}{25})\cdot \np = \frac{7}{50}\cdot \np$ by the correct probability of $\cO$.
For $X_{u,v}$, we also note that if the answer is at most $(30000-100000\eps)\cdot \log{n}$, then by \Cref{lem:sampling-approx-counterpart-test}, the vertex passes the test with high probability. Therefore, we assume $X_{u,v}\geq \expect{X_{u,v}}\geq (8000-100000\eps)\cdot \log{n}$, and we again apply Chernoff for the tail bound:
\begin{align*}
\Pr\paren{X_{u,v}\geq \paren{1/6-2\eps}\cdot \np} & \leq \exp\paren{-\frac{50\cdot (2/75-2\eps)^2}{7\cdot 3}\cdot \expect{Y_{u,v}}}\\
& \leq \frac{1}{10}\cdot \frac{1}{n^3} \tag{using $\expect{Y_{u,v}}\geq (8000-100000\eps)\cdot \log{n}$ and pick $\eps$ sufficiently small}. 
\end{align*}
For $Y_{u,v}$, note that only $t\in \Vsmall_{\tilde{\ell}}$ can report ``$v$ splits away from $(u,t)$'' since the LCA between $(u,t)$ is at least $\tilde{\ell}$ for any other $v$. Therefore, the number of signals we can possibly get is $\np/2$ plus the noise induced by the oracle $\cO$. As such, we again have 
$\expect{Y_{u,v}}\leq 11\np/20$.
Also, note that if $Y_{u,v}$ is less than $(30000-100000\eps)\cdot \log{n}$, the vertex $v$ would always pass the test, which allows us to assume w.log. that $\expect{Y_{u,v}}\geq (30000-100000\eps)\cdot \log{n}$. Therefore, we can again apply Chernoff bound to show
\begin{align*}
\Pr\paren{Y_{u,v}\geq (\frac{3}{5}-2\eps)\cdot \np} & \leq \exp\paren{-\frac{20\cdot (1/20-2\eps)^2}{3}\cdot \expect{Y_{u,v}}}\\
& \leq \frac{1}{10}\cdot \frac{1}{n^3} \tag{using $\expect{Z_{u,v}}\geq (30000-100000\eps)\cdot \log{n}$ and pick $\eps$ sufficiently small}.
\end{align*}
Therefore, we could apply \Cref{lem:sampling-approx-counterpart-test} to argue that with high probability, the vertex $v$ would pass the test. 
This part of the proof can be visualized as in \Cref{fig:counterpart-test-earlier}.

Combining the analysis of \Cref{line:split-balance-new-later}, \Cref{line:split-balance-new-middle}, and \Cref{line:split-balance-new-earlier} gives us the desired proof of \Cref{lem:counterpart-test-strong}.
\end{proof}

\begin{figure}
\centering
\begin{subfigure}{0.45\textwidth}
  \includegraphics[scale=0.40]{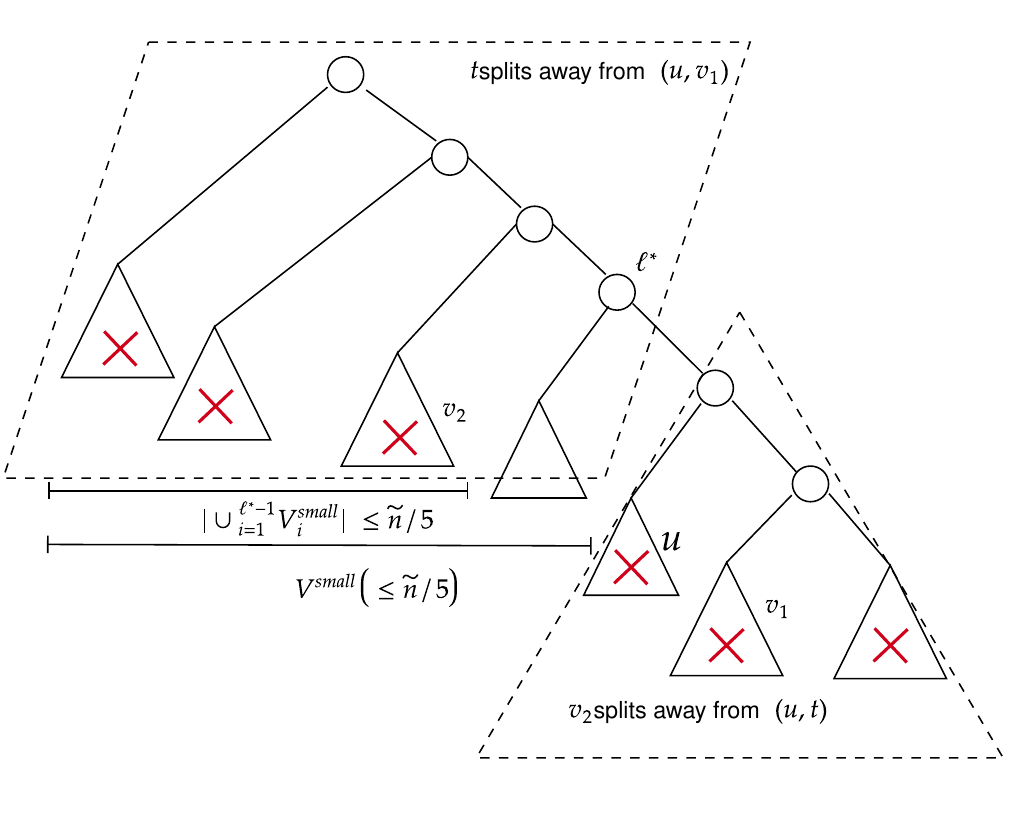}
  \caption{\centering Illustration of \Cref{line:split-balance-new-later} of \Cref{lem:counterpart-test-strong} with $u\in \cup_{i=\ellstar+1}^{\infty} \Vsmall_{i}$.}
  \label{fig:counterpart-test-later}
\end{subfigure}
\begin{subfigure}{0.45\textwidth}
  \includegraphics[scale=0.40]{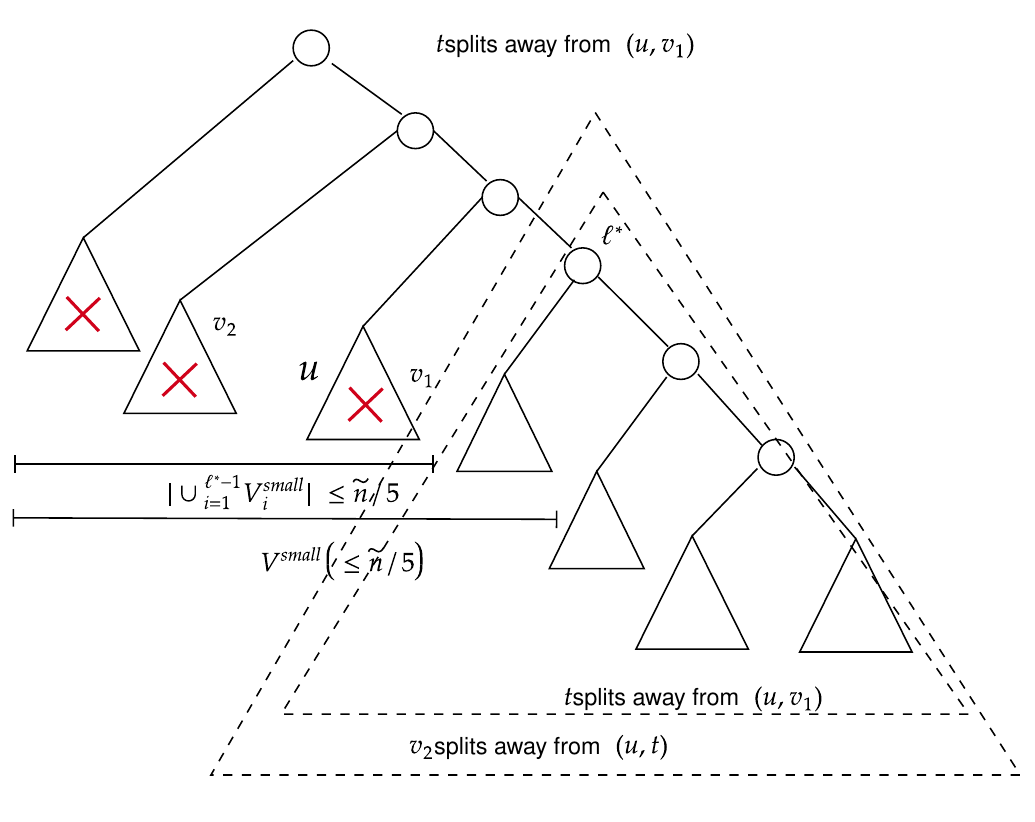}
  \caption{\centering Illustration of \Cref{line:split-balance-new-middle} of \Cref{lem:counterpart-test-strong} with $u\in \Vsmall_{\ell}$ for some $\ell\leq \ellstar$.}
  \label{fig:counterpart-test-middle}
\end{subfigure}
\begin{subfigure}{0.45\textwidth}
  \includegraphics[scale=0.40]{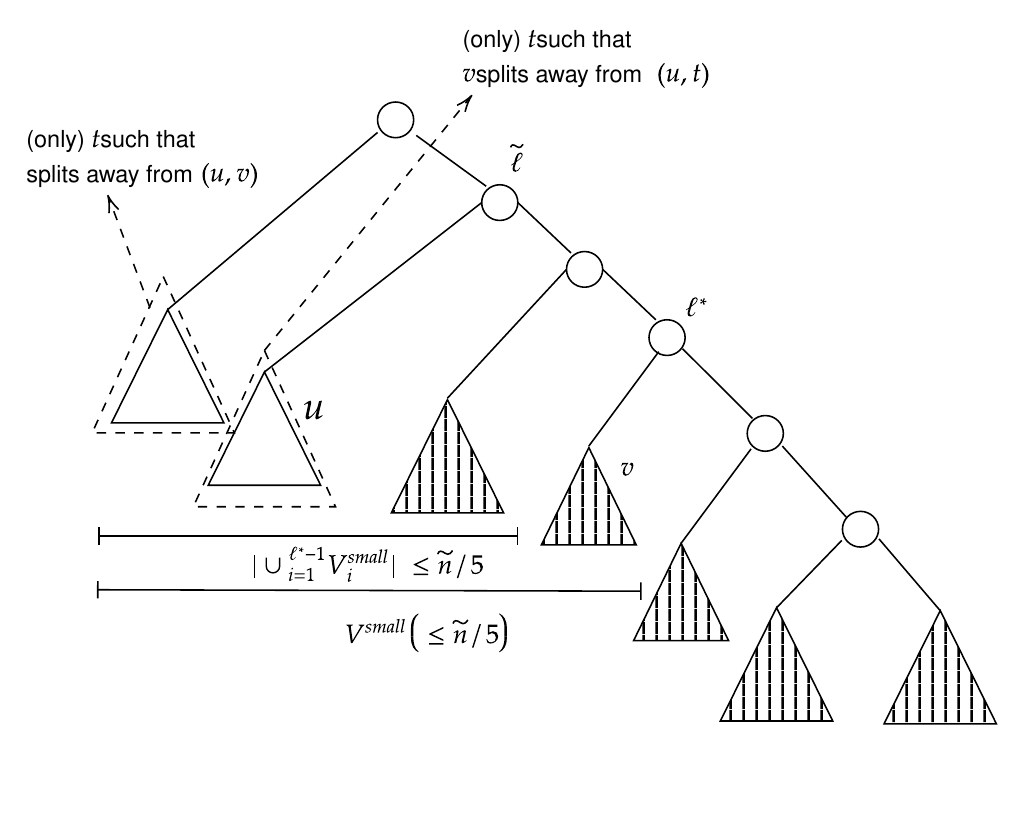}
  \caption{\centering Illustration of \Cref{line:split-balance-new-earlier} of \Cref{lem:counterpart-test-strong} with $u\in \Vsmall_{\ell}$ for $\ell\leq \tilde{\ell}$ as prescribed in \Cref{lem:counterpart-test-strong}.}
  \label{fig:counterpart-test-earlier}
\end{subfigure}
\begin{subfigure}{0.45\textwidth}
  \includegraphics[scale=0.40]{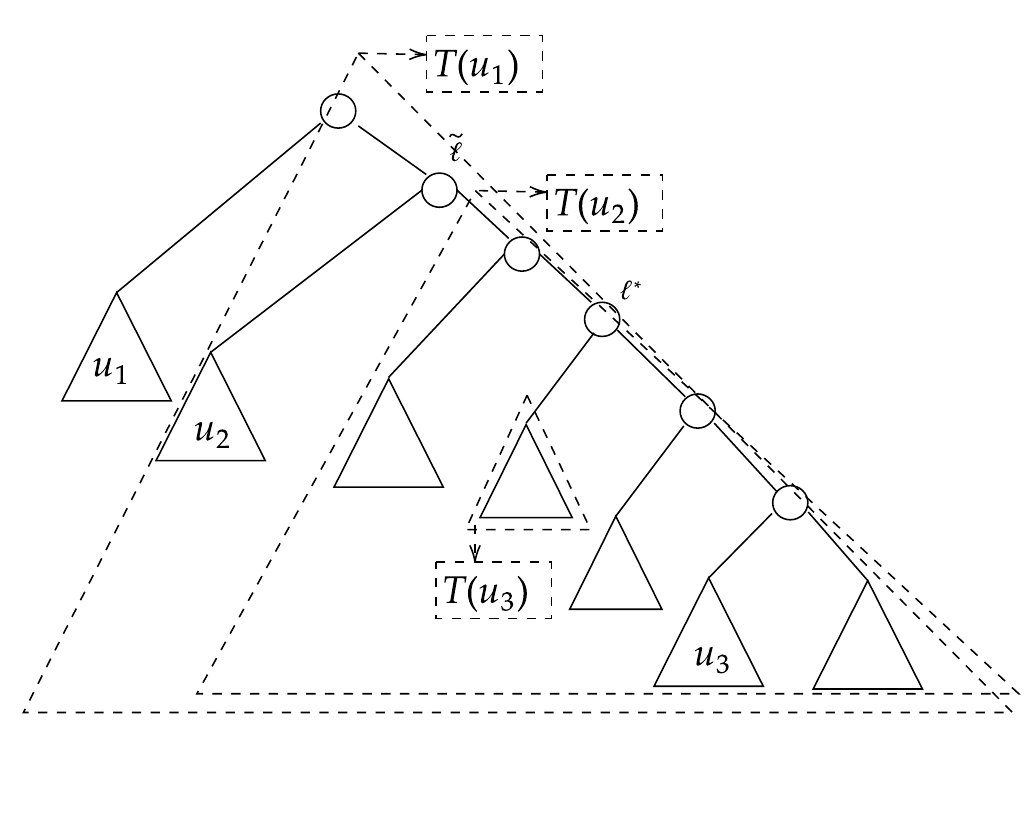}
  \caption{\centering The ``counterpart'' trees generated by vertices $u_1$, $u_2$, and $u_3$. We can simply take the largest partition to guarantee the root cut.}
  \label{fig:counterpart-test-consequence}
\end{subfigure}
\caption{An illustration of the analysis we used in the proof of \Cref{lem:counterpart-test-strong}.}
\label{fig:counterpart-test}
\end{figure}

Using \Cref{lem:counterpart-test-strong}, we can now argue that with high probability, the algorithm always correctly identifies the vertex split from the root as long as the size is at least $50000\log{n}$. More formally, we have
\begin{lemma}
\label{lem:root-cut-correct-identify}
For any composable $\Vp$ such that $\card{\Vp}\geq 50000 \log{n}$, let $\cTstar(\Vp)$ be the optimal HC tree restricted to $\Vp$, and suppose the root split of $\cTstar(\Vp)$ is $(S^*_{l}, S^*_{r})$. Then, the output $(T^*, \Vp\setminus T^*)$ of Line~\ref{line:root-cut} is exactly $(S^*_{l}, S^*_{r})$.
\end{lemma}
\begin{proof}
Assume without the loss the generality that $S^*_{l}\leq S^*_{r}$. Note that for any vertex $u\in S^*_{l}$, it is also among the vertex of $\Vbefore{\np/25}$. As such, by \Cref{lem:counterpart-test-strong}, the entire vertex set of $S^*_{r}$ is returned. On the other hand, for any vertex $u\in S^*_{r}$, we claim the induced set $T$ is necessarily smaller than $S^*_{r}$. To see this, note that by \Cref{lem:counterpart-test-strong}, the only possible case for Line~\ref{line:test-good-T-new} to \emph{not} return a \emph{subset} of $S^*_{r}$ is if $u\in \Vsmall_{\ell}$ such that $\ell>\ell^*$ and $\ell^*=1$. However, in such a case, the algorithm can at most return the set $S^*_{l}$, which is necessarily smaller than the size of $S^*_{r}$. As such, the maximum size of the set is attained by picking the vertex from $S^*_{l}$, which implies that the return rule of Line~\ref{line:root-cut} gives the partition $(S^*_{l}, S^*_{r})$.
\end{proof}

In essence, \Cref{lem:root-cut-correct-identify} is the very natural consequence of \Cref{lem:counterpart-test-strong}. This can be visualized in \Cref{fig:counterpart-test-consequence} of \Cref{fig:counterpart-test} -- the $T$ with the largest size is always the set of vertices induced by $u_1$, which is exactly what we want.

\paragraph{Finalizing the proof of \Cref{thm:strong-partial-tree}.}
In the beginning, we have $\Vp = V$, and the algorithm identifies the partition root partition for $V$ on $\cTstar$ with probability at least $1-\frac{1}{n^2}$. Subsequently, on each level of the recursion, as long as $\card{\Vp}\geq 50000 \log{n}$, the guarantee will hold. By \Cref{fct:num-splits-tree}, there are at most $n$ internal nodes in such a tree. Therefore, the high probability event holds for every split in \Cref{alg:strong-partial-tree}. 

Conditioning on the high probability event, any super-vertex has to induce a maximal subtree in $\cTstar$, which satisfies the requirement for the tree to be \emph{strongly consistent} with $\cTstar$. By \Cref{lem:strong-tree-run-time}, \Cref{alg:strong-partial-tree} deterministically takes at most $O(n^3)$ queries and $O(n^3\log{n}/\eps^2) = O(n^3\log{n})$ time by the choice of $\eps=\Theta(1)$.

Finally, we observe that \Cref{alg:strong-partial-tree} only takes $O(n\log{n})$ memory to implement. 
To see this, observe that for any recursive call of \Cref{alg:strong-partial-tree}, for any fixed $u$ and $v$, the subroutine \texttt{counterpart-tester-strong} can be implemented in $O(n\log{n})$ bits of space; furthermore, the space can be re-used for different $u$ and $v$ vertices. 
As such, we only need to maintain the set $T$ and its size for every $u$, which takes $O(n \log^2{n})$ space. In fact, we can maintain this in $O(n\log{n})$ space by keeping the set $T$ with the largest size and re-using the space on the fly. After each recursive call, we can free up the space of $\Vp$ to store $\Vp\setminus T^*$ and $T^*$ instead. On a separate $O(n\log{n})$-sized space, we can keep writing down the HC tree with 
\begin{enumerate}[label=\roman*).]
\item A name of the node that takes $O(\log{n})$ bits of memory;
\item The set of vertices induced by this node.
\end{enumerate}
In each recursion, we re-use the space of $ii)$ by erasing the set of vertices induced by an internal node $x$ once the children of $x$ are written down. In the end, the sets of vertices are only written on the super-vertices. As such, the entire partial HC tree can be stored in $O(n\log{n})$ memory as well.

\section{The Algorithm for the Weakly Consistent Partial HC Tree: Proof of \texorpdfstring{\Cref{thm:weak-partial-tree}}{weakp}}
\label{sec:weak-partial-tree}
We present the considerably more involved algorithm for the weakly consistent partial tree construction in \emph{near-linear} time. We first remind the readers of the main theorem of this result.

\thmweakpartialtree*

Our main effort is to show the algorithm that satisfies the guarantee of \Cref{thm:weak-partial-tree}.
As we have discussed in the high-level overview, the algorithm is divided into the ``split'' and the ``merging'' parts. We first show an algorithm that given a composable subset of vertices $\Vp \subseteq V$ such that $\card{\Vp}\geq \Omega(\log{n})$, return a set $T$ that induces a \emph{maximal subtree} in $\cT(\Vp)$. As such, we can recurse on $T$ and $\Vp \setminus T$ to gradually reduce the sizes to at most $O(\log{n})$. Then, in the second part, we show how to `glue' the subtrees together to eventually form a partial tree that is consistent with $\cTstar$. The overall structure of the algorithm can be shown as \Cref{alg:partial-tree}.

\begin{algorithm}
\caption{$\ptalg{\Vp, \Vh}{\cO}$: an algorithm for partial tree construction} \label{alg:partial-tree}
\KwIn{Input vertex set $\Vp$; input horizon vertex set $\Vh$; a splitting oracle $\cO$ as prescribed in \Cref{def:orl-model}}
\KwOut{A partial hierarchical clustering tree for $\Vp$.}
Initialize $\Vp=V$, $\Vh=V$.\;
If $\card{\Vp}<50000 \log{n}$, return a super-vertex (defined as in \Cref{def:partial-tree}).\;
Get the partition of $T$ such that $\cTree{T}$ is a complete subtree, i.e. $(T, \Vp\setminus T, \ustar, \Vh) \gets \splitalg{\Vp, \Vh}{\cO}$ (\Cref{alg:split})\;
Recurse on $T$ and $\Vp\setminus T$, i.e., \[\cT_{T} \gets \ptalg{T, T}{\cO}; \qquad \cT_{\Vp\setminus T} \gets \ptalg{\Vp\setminus T, \Vh}{\cO}\]
for the respective HC trees\;
Run the merging algorithm $\cT_{\Vp}\gets \treemerge{\cT_{T}}{\cT_{\Vp\setminus T}}{\ustar}$ (\Cref{alg:merge}).
\end{algorithm}

We will introduce and analyze the splitting and the merging algorithms in the subsequent sections.

\subsection{An algorithm to split the vertices}
\label{subsec:split-alg}
We first discuss our algorithm that produces complete subtrees for a given set of vertices $\Vp$. For the clarity of presentation, we use $\np$ to denote the size of $\Vp$, i.e., $\np=\card{\Vp}$.

\subsubsection{The algorithm} 
We now present the actual algorithm. We first give two ``tester'' subroutines that serves as a more general version of the \texttt{counterpart-tester-strong} algorithm we used in \Cref{alg:strong-partial-tree}.

\FloatBarrier
\begin{algorithm}
\caption{$\testalg{u}{t}{\textnormal{threshold}_1}{\textnormal{threshold}_2}$: an algorithm to test whether $t$ is among the ``counterpart'' of $u$} \label{alg:tester}
\KwIn{Vertex set $\Vh$; a splitting oracle $\cO$; baseline vertex $u$; test vertex $v$; $\text{threshold}_1\in (0,\np)$, $\text{threshold}_2\in (0,\np)$}
\KwOut{Whether $v$ is a ``counterpart'' of $u$}
Initialize counters $c_1\gets 0$, $c_2 \gets 0$\;
\For{Every vertex $t\in \Vh$}{
If $v$ splits away from $(u,t)$, increase $c_1$ by 1\;
If $t$ splits away from $(u,v)$, increase $c_2$ by 1\;
}
\If{$c_1\leq \textnormal{threshold}_1$ and $c_2 \leq \textnormal{threshold}_2$}{
Return ``$v$ is a counterpart of $u$''.
}
\end{algorithm}

\Cref{alg:tester} is very similar to \Cref{alg:tester-strong}, and the difference is that instead of testing on $\Vp$ set, we query on the $\Vh$ set instead. This is more than a simple notation change: it allows us to separate the vertices ``to be tested'' (e.g., vertices in $\Vp$) vs. the set we ``used in the tests'' (e.g., vertices in $\Vh$).

Next, we introduce the tester that given vertices $u\in \Vsmall_{\ell_{1}}$ and $t \in \Vsmall_{\ell_{2}}$, returns whether $\ell_2 < \ell_1$, i.e., whether $t$ is ``split earlier'' in the small-tree split order of $u$. 

\FloatBarrier
\begin{algorithm}
\caption{$\testalgprev{u}{t}{\textnormal{threshold}}$: an algorithm to test whether $t$ is among the ``predecessor'' of $u$} \label{alg:tester-prev}
\KwIn{Vertex set $\Vh$; a splitting oracle $\cO$; baseline vertex $u$; test vertex $t$; $\text{threshold}\in (0,\np)$}
\KwOut{Whether $t$ is a ``predecessor'' of $u$}
Initialize counters $c\gets 0$ \;
\For{Every vertex $s\in \Vh$}{
If $t$ splits away from $(u,s)$, increase $c$ by 1\;
}
\If{$c\geq \textnormal{threshold}$}{
Return ``$t$ is a predecessor of $u$''.
}
\end{algorithm}

Note that unlike \Cref{alg:tester}, the vertex $t$ passes the test in \Cref{alg:tester-prev} if the number of ``split away'' is lower-bounded. With the $\testalgprev{u}{t}{\textnormal{threshold}}$ algorithm, we can now define the following algorithm that tests whether a root split has happened in the previous iteration.

\FloatBarrier
\begin{algorithm}
\caption{$\testorphan$: an algorithm to obtain any vertex that is split earlier than the orphaned vertex set in the small-tree split order} \label{alg:orphan-test}
\KwIn{Vertex set $\Vp$ of size $\np$; the horizon set $\Vh$ of size $\nh$; a splitting oracle $\cO$}
\KwOut{Vertex $X^*$ that is either empty or contain ``predecessor'' of orphaned sets}
Sample a set $U'$ of $100 \log{n}$ vertices from $\Vp$ using \emph{fresh randomness}\;
\For{$u' \in U'$}{
\For{$v \in \Vh$}{
\If{$\nh\leq 3\np$}{
$\textnormal{threshold-pred} \gets 3\np/2$\;
}
\Else{
$\textnormal{threshold-pred} \gets 2\nh/3$\;
}
Add $x$ to $X$ if $\testalgprev{u}{v}{\textnormal{threshold-pred}}$ (\Cref{alg:tester-prev}) reports ``$v$ is a predecessor of $u$''\label{line:pred-test}\;
}
Record the size $\card{X}$ for this choice of $u'$\;
}
Return $X^*$ as the largest size of $X$\;
\end{algorithm}

We now continue to the presentation of our tree-split algorithm. The full algorithm is as \Cref{alg:split}.

\FloatBarrier
\begin{algorithm}
\caption{$\splitalg{\Vp, \Vh}{\cO}$: an algorithm that splits $\Vp$ to $T$ and $\Vp\setminus T$} \label{alg:split}
\KwIn{Vertex set $\Vp$ of size $\np$; the horizon set $\Vh$ of size $\nh$; a splitting oracle $\cO$}
\KwOut{Vertex sets $T$, $\Vp\setminus T$; new horizon set $\Vh$}
\label{line:split-sample}Sample a set $U$ of $500 \log{n}$ vertices from $\Vp$\;
\For{$u \in U$}{
Initialize $T \gets \emptyset$\;
\For{$v \in \Vh$}{ \label{line:test-good-T} 
Run $\testalg{u}{v}{3\nh/5}{\nh/6}$\; 
Add $v$ to $T$ if $v$ is a ``counterpart'' of $u$\;
}
Record the size $\card{T}$ for this choice of $u$\;
}
Pick $(T^*, \Vp\setminus T^*)$ such that $T^*\gets T$ is with the largest size\; 
\If{$T^* \cap \Vp = \emptyset$}{
\nonl\algcomment{Test whether a root split happened in the last iteration}\;
$X^* \gets \testorphan$  (using \Cref{alg:orphan-test})\;
\If{$X^* \neq \emptyset$}{\label{line:large-orphan}
\nonl\algcomment{The case of non-root split}\;
Pick an arbitrary $\ustar \in X^*$\;
Initialize $T^* \gets \emptyset$\;
\For{$v \in \Vh$}{
Run $\testalg{\ustar}{v}{3\nh/5}{\nh/6}$\; 
Add $v$ to $T^*$ if $v$ is a ``counterpart'' of $\ustar$\;
}
Output with the same rules of the $T^*\cap \Vp \neq \emptyset$ case\;
}
\Else{ 
\nonl\algcomment{Case for root split}\;
Let $\Vh\gets \Vp$, run and output with $\splitalg{\Vp, \Vp}{\cO}$ \label{line:recursion-call}\;
Enforce a \emph{single level} of recursion call (return ``FAIL'' if $\Vp=\Vh$ and the algorithm enters the above line again with $\Vp=\Vh$)\;
}
}
\Else{ \label{line:output-lines}
\nonl\algcomment{Keep the current horizon}\;
Output $(T^*\cap \Vp, \Vp\setminus T^*)$ as the partition and keep the same $\Vh$\; 
Output $\ustar$ as the vertex $u \in U$ corresponding to the output $T^*$\; 
}
\end{algorithm}
\FloatBarrier

\noindent
Note that in Line~\ref{line:split-sample}, we sample $500 \log{n}$ as opposed to $500 \log{\np}$ vertices (this is on purpose as opposed to being a typo). 
The formal guarantee of the tree split algorithm is as follows.

\begin{lemma}
\label{lem:tree-split}
Suppose $\Vp$ is a composable set with size $\np$ that satisfies the \emph{weak contraction property} as prescribed in \Cref{def:weakly-consistent-tree}, and suppose $\np \geq 50000 \log{n}$. Furthermore, let $\Vh$ be a single composable set (e.g., $\Vh$ induces a maximal tree in $\cTstar$) of $\nh$ vertices, and suppose 
\[\Vp =\cup_{i=1}^{\ell} \Vsmall_{i}\] 
for some $\ell$ in the small-tree split order of $\Vh$. Then, given a splitting oracle $\cO$ with correct probability at least $9/10$, \Cref{alg:split} with high probability outputs composable sets $T^{*}\cap \Vp$ and $\Vp\setminus T^{*}$ and a vertex $\ustar$ such that 
\begin{enumerate}[leftmargin=12pt]
\item\label{eff-split:T} $T^*$ induces a \emph{single} maximal subtree in $\cTstar(\Vp)$. 
\item\label{eff-split:out-degree} $\Vp\setminus T^{*}$ satisfies the \emph{weak contraction property} as prescribed in \Cref{def:weakly-consistent-tree}, and in particular, the subtrees induced by $\Vp\setminus T^{*}$ have
\begin{enumerate}
\item one edge connected to the node $\pa{\lcatreeset{\Vp}{\cTstar}}$ (if it exist).
\item one edge connected to the node that induces $T^{*}$ in $\cTstar$.
\end{enumerate}
\item\label{eff-split:easy-to-merge} The lowest common ancestor between the nodes in $T$ is an (immediate) child node of the lowest common ancestor between $T\cup \{\ustar\}$ in $\cTstar(\Vp)$, i.e.,
\[\lcatreeset{T\cup \{\ustar\}}{\cTstar(\Vp)} = \pa{\lcatreeset{T}{\cTstar(\Vp)}}. \]
\item \label{eff-split:size-bounds} Size properties: \emph{at least one} of the following guarantees hold.
\begin{enumerate}
\item \label{case:large-orphan} The new set $\Vp\setminus T^*$ has at least $\frac{99}{100}$ fraction of vertices that are orphaned vertices, i.e., the new $\Vorphan$ set accounts of at least $\frac{99}{100}$ fraction of vertices in $\Vp\setminus T^*$;
\item \label{case:large-T-star} The size of $T^*$ satisfies the following properties:
\begin{enumerate}
    \item\label{eff-split:size-lb} The size of $T^{*}\cap \Vp$ is at least $\frac{1}{200}\cdot \np$, i.e., 
    \[\card{T^{*}\cap \Vp} \geq \frac{1}{200}\cdot \np.\]
    \item\label{eff-split:size-ub} If $\Vh=\Vp$, the size of $T^{*}\cap \Vp$ is at most $(1-\frac{1}{10000 \log^2{\np}})\cdot \np$, i.e.,
    \[\card{T^{*}\cap \Vp} \leq (1-\frac{1}{10000 \log^2{\np}})\cdot \np.\]
\end{enumerate}
\end{enumerate}
\end{enumerate}
Furthermore, case \cref{case:large-T-star} always happens if $\Vh = \Vp$, and the algorithm runs in time $O(\nh^2 \cdot \log n)$.
\end{lemma}

We now proceed to the analysis to prove \Cref{lem:tree-split}.





\subsubsection*{The analysis} 
To begin with, we define the notion of \emph{orphaned vertices} from our split procedure. 
\begin{definition}[Orphaned vertices]
\label{def:orphan-vertex}
Let $\Vp$ be a composable set of vertices, and let $T= \cup_{i=\ell+1}^{\infty} \Vsmall_{i}$. We call $\Vsmall_{\ell} \subseteq \Vp$ the set of orphaned vertices in $\Vsmall_{\ell}$ with respect to $T$. We denote the orphaned set of vertices as $\Vorphan_{T, \ell}$, and we write $\Vorphan$ as the simplified notation when the context is clear.
\end{definition}

One can refer to \Cref{fig:horizon-sets} (in \Cref{sec:tech-overview}) for a visualization of the orphaned vertices (we used $V$ as opposed to $\Vp$ in the figure). When the context is clear, we ignore the dependence on $\Vsmall_{\ell}$ and $T$ when talking about orphaned vertices, and simply denote the orphaned vertices as $\Vorphan$. 
In our proof of correctness, we will show ``inductively'' that the set $\Vp$ only has a \emph{single} orphaned set $\Vorphan$ -- a key to guarantee property \Cref{eff-split:out-degree} of \Cref{lem:tree-split}. 

We now proceed with the relatively straightforward analysis of the running time of a \emph{single} level of recursion.
\begin{lemma}
\label{lem:efficient-split-run-time}
The algorithm $\splitalg{\Vp, \Vh}{\cO}$ runs in time $O(\nh^2 \cdot \log n)$. 
\end{lemma}
\begin{proof}
Each run of the \texttt{counterpart-tester} takes $O(\nh)$ time. As such, the procedure that finds $\ustar$ takes $O(\nh^2 \cdot \log {n})$ time since it involves $O(\nh \log{n})$ calls of \texttt{counterpart-tester}. Similarly, the procedure that finds $X^*$ takes at most $O(\nh \log{n})$ calls of \texttt{predecessor-tester}, which again result in $O(\nh^2 \cdot \log \np)$ time. Furthermore, if $X^* \neq \emptyset$, we only make $\nh$ number of \texttt{counterpart-tester} calls, which takes $O(\nh^2)$ time. On the other hand, if $X^* = \emptyset$, since we insist on a single level of recursion call, the runtime overhead is at most $O(\nh^2 \cdot \log \np)$. Adding up the time complexity of the above two procedures gives us the desired bound.
\end{proof}

We proceed with the proof of correctness for the algorithm. To this end, we first observe that by \Cref{obs:split-away-preserv} and since $\Vh$ is composable, the answers for the splitting oracle $\cO$ on $(u,v,w)\in \Vh$ is fully preserved in $\cTstar(\Vh)$. 
We now give a technical lemma that characterizes the return set $T$ by running \texttt{counterpart-tester} on different sets of vertices -- this is essentially the same argument we used in \Cref{lem:counterpart-test-strong}, albeit we switched $\Vp$ to $\Vh$. We provide the lemma and the analysis for the purpose of completeness.

\begin{lemma}[Cf. \Cref{lem:counterpart-test-strong}]
\label{lem:counterpart-test}
For any composable $\Vh$ such that $\card{\Vh} \geq 50000 \log{n}$, let $\ell^*$ be the maximal level that $\Vsmall_{\ellstar}$ contains a vertex in $\Vbefore{\nh/5}$. With high probability, Line~\ref{line:test-good-T} in \Cref{alg:split} satisfies the following properties.
\begin{enumerate}[label=\alph*).]
\item \label{line:split-balance-later} For any vertex $u\in \cup_{i=\ellstar+1}^{\infty} \Vsmall_{i}$, there are
\begin{itemize}
\item No vertex $v \in \cup_{i=\ellstar+1}^{\infty} \Vsmall_{i}$ can be added to $T$ by Line~\ref{line:test-good-T} of \Cref{alg:split}.
\item No vertex $v \in \cup_{i=1}^{\ellstar-1} \Vsmall_{i}$ can be added to $T$ by Line~\ref{line:test-good-T} of \Cref{alg:split}.
\end{itemize}
\item \label{line:split-balance-middle} For every level $\ell\leq \ellstar$ and vertices $u\in \Vsmall_{\ell}$, with high probability, there are 
\begin{itemize}
\item No vertex $v\in \Vsmall_{\ell}$ can be added to $T$ by Line~\ref{line:test-good-T} of \Cref{alg:split}.
\item No vertex $v\in \Vsmall_{k}$ for $k<\ell$ can be added to $T$ by Line~\ref{line:test-good-T} of \Cref{alg:split}.
\end{itemize}
\item \label{line:split-balance-earlier} There exists a $\tilde{\ell}$ such that $\card{\cup_{i=1}^{\tilde{\ell}} \Vsmall_{i}}\geq \frac{\nh}{25}$, and for any $\ell\leq \tilde{\ell}$ and any vertex $u\in \Vsmall_{\ell}$, with high probability, all vertices $v\in \Vsmall_{k}$ for $k>\ell$ are added to $T$ by Line~\ref{line:test-good-T} of \Cref{alg:split}.
\end{enumerate}
\end{lemma}
\begin{proof}
Let $\ellstar$ be the maximal level that the $\Vsmall_{\ellstar}$ contains a vertex in $\Vbefore{\nh/5}$, and suppose the split on this level results in $(S^{\ellstar}_{l}, S^{\ellstar}_{r})$. We again assume w.log. that $\card{S^{\ellstar}_{l}}\leq \card{S^{\ellstar}_{r}}$, which means $S^{\ellstar}_{l}$ becomes $\Vsmall_{\ellstar}$. The bounds in \Cref{obs:split-structure} still hold with parameter $\nh$, i.e., $\frac{2\nh}{5} < \card{\cup_{i=\ellstar+1}^{\infty} \Vsmall_{i}} \leq \frac{4\nh}{5}$.
We now proceed to the proofs of \Cref{line:split-balance-later}, \Cref{line:split-balance-middle}, and \Cref{line:split-balance-earlier}, respectively. 

\paragraph{Proof of \Cref{line:split-balance-later}} Fix any vertex $u\in \cup_{i=\ellstar+1}^{\infty} \Vsmall_{i}$. For the first statement, note that for every $v\in \cup_{i=\ellstar+1}^{\infty} \Vsmall_{i}$, all the vertices in $\cup_{i=1}^{\ellstar} \Vsmall_{i}$ are \emph{splitting away} from $(u,v)$. Define $X_{u,v,t}$ as the indicator random variable answered by $\cO$ as ``$t$ splits away from $(u,v)$'', and define $X_{u,v}=\sum_{t\in\Vp}X_{u,v,t}$ as the total number of vertex split away from $(u,v)$. By \Cref{obs:split-structure}, the expected number of split away reported by oracle $\cO$ is at least
\begin{align*}
\expect{X_{u,v}}\geq \frac{9}{10}\cdot \frac{\nh}{5} = \frac{9}{50}\cdot \nh.
\end{align*}
Note also that $X_{u,v}$ is a summation of $0/1$ independent random variables. As such, we can apply the Chernoff bound to get that
\begin{align*}
\Pr\paren{X_{u,v}\leq \frac{\nh}{6}} & = \Pr\paren{X_{u,v}\leq \frac{25}{27}\cdot \expect{X_{u,v}}} \\
& \leq \exp\paren{-\frac{(2/27)^2}{3}\cdot \expect{X_{u,v}}}\\
& \leq \frac{1}{10}\cdot \frac{1}{n^3} \tag{$\frac{9}{50}\cdot \nh \geq 9000 \log{n}$ using the lower bound on the size of $\Vh$}. 
\end{align*}
For the second statement, note that for any $t\in \cup_{i=\ellstar}^{\infty} \Vsmall_{i}$, $v$ actually splits away from $(u,t)$. As such, there is at least $\frac{4\nh}{5}$ vertices such that ``$v$ splits away from $(u,t)$'' -- this is because the LCA between $(u,v)$ has to be higher than the LCA between $(u,t)$ for any $t\in \cup_{i=\ellstar}^{\infty} \Vsmall_{i}$. As such, define $Y_{u,v}$ as the number of total answers of ``$v$ splits away from $(u,t)$'' by $\cO$, we again have
\begin{align*}
\expect{Y_{u,v}} \geq \frac{9}{10}\cdot \frac{4\nh}{5}=\frac{18}{25}\cdot \nh.
\end{align*}
As such, we again have
\begin{align*}
\Pr\paren{Y_{u,v}\leq \frac{3\nh}{5}} &= Pr\paren{Y_{u,v}\leq \frac{15}{18}\cdot \expect{X_{u,v}}}\\
&\leq \frac{1}{10}\cdot \frac{1}{n^3} \tag{$\frac{9}{50}\cdot \nh \geq 9000 \log{n}$ using the lower bound on the size of $\Vh$}.
\end{align*}
By a union bound over at most $n$ vertices, no vertex $v\in \cup_{i=1}^{\ellstar-1} \Vsmall_{i}$ or $v\in \cup_{i=\ellstar+1}^{\infty} \Vsmall_{i}$ can pass the test and be recorded as a `counterpart'.

\paragraph{Proof of \Cref{line:split-balance-middle}}
We now show that vertices $v \in \cup_{i=1}^{\ell} {\Vsmall_{i}}$ cannot pass the ``counterpart'' test. For vertices in $v \in \Vsmall_{\ell}$, we claim that there are too many vertices $t$ that split away from $(u,v)$. To see this, let us define $A_{u,v}$ as the total number of answers of ``$t$ splits away from $(u,v)$'' answered by $\cO$. Note that for any vertex $t \in \cup_{i=\ellstar+1}^{\infty} \Vsmall_{i}$, the answer is always ``yes'' since $\ell^{*}\geq \tilde{\ell}\geq \ell$. Therefore, we can lower bound the expectation of $A_{u,v}$ with $\expect{A_{u,v}}\geq \frac{2}{5}\nh\cdot \frac{9}{10} \geq \frac{9}{25}\cdot \nh$. Therefore, we can apply Chernoff bound to get
\begin{align*}
\Pr\paren{A_{u,v}\leq \frac{\nh}{6}} & \leq \exp\paren{-\frac{(1/2)^2}{3}\cdot \expect{A_{u,v}}}\leq \frac{1}{10}\cdot \frac{1}{n^3}.
\end{align*}
We now turn to the vertices that are in $\cup_{i=1}^{\ell-1} {\Vsmall_{i}}$. Note that by definition, there is $\card{\cup_{i=\ell}^{\infty} {\Vsmall_{i}}} \geq \card{\cup_{i=\ellstar}^{\infty} {\Vsmall_{i}}} \geq \frac{4}{5}\cdot \nh$. For any $t \in \cup_{i=\ellstar}^{\infty} {\Vsmall_{i}}$, $v$ splits away from $(u,t)$. As such, define $B_{u,v}$ as the number of answers by $\cO$ with ``$v$ splits away from $(u,t)$'', there is $\expect{B_{u,v}}\geq \frac{18}{25}\cdot \nh$. Once again, by applying Chernoff bound, there is 
\begin{align*}
\Pr\paren{B_{u,v}\leq \frac{3\nh}{5}} & \leq \exp\paren{-\frac{(1/5)^2}{3}\cdot \expect{B_{u,v}}}\leq \frac{1}{10}\cdot \frac{1}{n^3}.
\end{align*}
Applying a union bound over the cases and all $\nh\leq n$ vertices gives us the desired statements.

\paragraph{Proof of \Cref{line:split-balance-earlier}} 
We fix $\tilde{\ell}$ as the minimum level such that the union of $\cup_{i=1}^{\tilde{\ell}}\Vsmall_{i}$ has at least $\nh/25$ vertices. Let $\ell\leq \tilde{\ell}$, and for any vertex $v$ in $\Vsmall_{k}$ for $k>\ell$, define the following two random variables: $X_{u,v}$ as the number of answers that ``$t$ splits away from $(u,v)$'', and $Y_{u,v}$ as the number of answers that ``$v$ splits away from $(u,t)$''. We shall show that both terms are not large and $v$ can always pass the test. (Note that $X$ and $Y$ are overloaded and are \emph{not} related to their meaning in the proof of Line~\ref{line:split-balance-later}.)



Note that by definition, we have $\card{\cup_{i=1}^{\tilde{\ell}-1} \Vsmall_{i}}<\frac{\nh}{25}$. Since $\ell\leq \tilde{\ell}$, for any vertex $u \in \cup_{i=1}^{\ell}\Vsmall_{i}$, only vertices $t\in \cup_{i=1}^{\ell-1} \Vsmall_{i}$ are \emph{actually} splitting away from $(u,v)$ for $v \in \cup_{i=\ell}^{\infty} \Vsmall_{i}$. As such, we have 
$\expect{X_{u,v}}\leq (\frac{1}{10}+\frac{1}{25})\cdot \nh = \frac{7}{50}\cdot \nh$ by the correct probability of $\cO$.
For $X_{u,v}$, we also note that if the answer is at most $3000 \log{n}$, then the vertex trivially passes the test. Therefore, we assume $X_{u,v}\geq \expect{X_{u,v}}\geq 3000 \log{n}$, and we again apply Chernoff for the tail bound:
\begin{align*}
\Pr\paren{X_{u,v}\geq \frac{\nh}{6}} & \leq \exp\paren{-\frac{(4/21)^2}{3}\cdot \expect{Y_{u,v}}}\\
& \leq \frac{1}{10}\cdot \frac{1}{n^3} \tag{using the condition $\expect{Y_{u,v}}\geq 3000 \log{n}$}. 
\end{align*}
For $Y_{u,v}$, note that only $t\in \Vsmall_{\tilde{\ell}}$ can report ``$v$ splits away from $(u,t)$'' since the LCA between $(u,t)$ is at least $\tilde{\ell}$ for any other $v$. Therefore, the number of signals we can possibly get is $\nh/2$ plus the noise induced by the oracle $\cO$. As such, we again have 
$\expect{Y_{u,v}}\leq 11\nh/20$.
Also, note that if $Y_{u,v}$ is less than $30000 \log{n}$, the vertex $v$ would always pass the test, which allows us to assume w.log. that $\expect{Y_{u,v}}\geq 30000 \log{n}$. Therefore, we can again apply Chernoff bound to show
\begin{align*}
\Pr\paren{Y_{u,v}\geq \frac{3\nh}{5}} & \leq \exp\paren{-\frac{(1/11)^2}{3}\cdot \expect{Y_{u,v}}}\\
& \leq \frac{1}{10}\cdot \frac{1}{n^3} \tag{using the condition $\expect{Z_{u,v}}\geq 30000 \log{n}$}. 
\end{align*}
\myqed{\Cref{lem:counterpart-test}}
\end{proof}

A direct corollary of \Cref{lem:counterpart-test} is that the set $T$ induced by any $u\in \cup_{i=1}^{\tilde{\ell}} \Vsmall_{i}$ is larger than the $T$ induced by $u\in \cup_{i=\tilde{\ell}+1}^{\infty} \Vsmall_{i}$, and the ``higher level'' vertices in $\cup_{i=1}^{\tilde{\ell}} \Vsmall_{i}$ induces larger sets. More formally, we can summarize this observation as follows.
\begin{lemma}
\label{lem:counterpart-set-larger}
Conditioning on the high-probability event of \Cref{lem:counterpart-test}, the following statements are true:
\begin{itemize}
\item Let $T_1$ be the vertex set induced by $u_1 \in \cup_{i=1}^{\tilde{\ell}} \Vsmall_{i}$ from Line~\ref{line:test-good-T} in \Cref{alg:split}, and let $T_2$ be the vertex set induced by $u_2 \in \cup_{i=\tilde{\ell}+1}^{\infty} \Vsmall_{i}$ from Line~\ref{line:test-good-T} in \Cref{alg:split}. We have $\card{T_1}\geq \card{T_2}$.
\item Let $\ell_{1}\leq \ell_{2}\leq \tilde{\ell}$. Let $T_1$ be the vertex set induced by $u_1 \in \Vsmall_{\ell_{1}}$ from Line~\ref{line:test-good-T} in \Cref{alg:split}, and let $T_2$ be the vertex set induced by $u_2  \in \Vsmall_{\ell_{2}}$ from Line~\ref{line:test-good-T} in \Cref{alg:split}. We have $\card{T_1}\geq \card{T_2}$.
\end{itemize}
\end{lemma}
\begin{proof}
We prove the second bullet first since the conclusion can be used to prove the first bullet. Note that conditioning on the high-probability event of \Cref{lem:counterpart-test}, if $\ell_{1}\leq \ell_{2}$, we have $T_{2}\subseteq T_{1}$ by \Cref{line:split-balance-earlier}. Therefore, we have $\card{T_1}\geq \card{T_2}$.

For the first bullet, note that conditioning on the high-probability event of \Cref{lem:counterpart-test}, the set $T_{2}$ can either be $\Vsmall_{\tilde{\ell}}$ (by \Cref{line:split-balance-later}) or $\cup_{i=\tilde{\ell}+1}^{\infty} \Vsmall_{i}$ (by \Cref{line:split-balance-middle}). In either case, the size of such a set is at most $\cup_{i=\tilde{\ell}+1}^{\infty} \Vsmall_{i}$, which is the set $T_1$ generated by $u\in \Vsmall_{\tilde{\ell}}$. Furthermore, by the result in the second bullet, if $u\in \Vsmall_{i}$ for some $i\leq \tilde{\ell}$, the induced set $T_1$ can only be larger. This proves the first bullet. 
\end{proof}

We now show that conditioning on \Cref{lem:counterpart-test} (resp. \Cref{lem:counterpart-set-larger}), if the size of the orphaned set is relatively small, we will not need the subroutine in Line~\ref{line:large-orphan}, and all the guarantees in \Cref{lem:tree-split} will be satisfied.
\begin{lemma}
\label{lem:small-orphan-case}
Let $\Vp$ and $\Vh$ be as prescribed by \Cref{lem:tree-split}, and suppose the size of the orphaned set is at most $\frac{99}{100}\cdot \np$, i.e., $\card{\Vorphan}\leq \frac{99}{100}\cdot \np$. Then, with high probability, we have $T^* \cap \Vp \neq \emptyset$, and the resulting $T^*$ and $\Vp\setminus T^*$ satisfy the properties of \Cref{lem:tree-split}.
\end{lemma}
\begin{proof}
We discuss the cases based on whether $\Vh=\Vp$ and whether there exists a $\Vsmall_{\ell}$ such that $\card{\Vsmall_{\ell}}\geq \frac{99}{100}\cdot \np$.
\begin{enumerate}
\item \textbf{If $\Vh=\Vp$.} In this case, we show that with high probability, the guarantees in \Cref{eff-split:T}, \Cref{eff-split:out-degree}, and \Cref{eff-split:easy-to-merge} of \Cref{lem:tree-split} always hold, and the \cref{case:large-T-star} case of \Cref{lem:tree-split} is going to happen. To see this, note that with high probability, we will sample a vertex that is among the first $\frac{\np}{25}$ vertices in the small-tree split order: the probability for us to not sample a vertex from $\Vbefore{\frac{\np}{25}}$ is at most
\begin{align*}
(\frac{24}{25})^{500\log{n}}\leq \frac{1}{10}\cdot \frac{1}{n^2}.
\end{align*}
As such, we can condition on a vertex $v'$ among $\Vbefore{\frac{\np}{25}}$ is sampled. By \Cref{lem:counterpart-test} and \Cref{lem:counterpart-set-larger}, the counterpart set induced by $v' \in \Vbefore{\frac{\np}{25}}$ is among $\cup_{i=1}^{\tilde{\ell}} \Vsmall_{i}$, which necessarily induces a larger set than any other $u\in \cup_{i=\tilde{\ell}+1}^{\infty} \Vsmall_{i}$. Therefore, the induced set $T^*$ must be from the vertex $v'$.

We now use this to verify the desired properties. The proofs of \Cref{eff-split:T}, \Cref{eff-split:out-degree}, and \Cref{eff-split:easy-to-merge} are straightforward as follows.
\begin{itemize}
\item For \Cref{eff-split:T}, note that by \Cref{lem:counterpart-test}, the induced set $T^*$ is always $\cup_{i=\ell}^{\infty} \Vsmall_{i}$, which forms a maximal subtree in $\cTstar(\Vh)$. Similarly, $T\cap \Vp$ induces a maximal subtree in $\cTstar(\Vp)$.
\item For \Cref{eff-split:out-degree}, note that as long as $T^*$ includes the orphaned set $\Vorphan$, there will be only one edge connecting to the node induces $T^*$ in $\cTstar$. In this case, there is no orphaned set, and \Cref{eff-split:out-degree} holds trivially.
\item \Cref{eff-split:easy-to-merge} directly follows from \Cref{lem:counterpart-test} since $T^*$ is induced by $u^*$.
\end{itemize}

For the size upper and lower bounds (\Cref{eff-split:size-bounds}), we verify that the guarantees for case \ref{case:large-T-star} always holds. Note that conditioning on the high-probability event of \Cref{lem:counterpart-test}, the size is at least $\frac{2\nh}{5}=\frac{2\np}{5}$ (\Cref{obs:split-structure}). Therefore, the set $T^*\cap \Vp=T^*$ has size at least $\frac{2\np}{5}\geq \frac{1}{200}\cdot \np$, which proves the lower bound (\Cref{eff-split:size-lb}). For the upper bound (\Cref{eff-split:size-ub}), we note that for the size of $T^*\cap \Vp$ to be more than $(1-\frac{1}{10000 \log^2{\np}})\cdot \np$, a \emph{necessary} condition is to sample a vertex $u \in \Vbefore{\np/10000\log^2{\np}}$. Since we sample $500\log{n}$ vertices, define $X$ as the random variable for the number of vertices sampled from $\Vbefore{\np/10000\log^2{\np}}$, we have
\begin{align*}
\expect{X}\leq \frac{1}{20\log n}.
\end{align*}
Since $X$ is a summation of independent random variables supported on $[0,1]$, we can apply Chernoff bound to show that
\begin{align*}
\Pr\paren{X\geq 1} &= \Pr\paren{X\geq 50\log{n}\cdot \expect{X}}\\
& \leq \exp\paren{-\frac{2500\log^2{n}\cdot \frac{1}{20\log{n}}}{2+20\log{n}}}\\
&\leq \frac{1}{10}\cdot \frac{1}{n^2}.
\end{align*}
Therefore, we can apply a union bound to show that with high probability, the size of $T^*\cap \Vp$ will \emph{not} be larger than $(1-\frac{1}{10000 \log^2{\np}})\cdot \np$, as desired.

\item \textbf{If $\Vh\neq \Vp$.} We need to handle this case with more care. We first show that at least one vertex that is in $\Vp \setminus \Vorphan$ can be sampled with high probability. To see this, note that by the size bound on $\card{\Vorphan}$, the probability for a vertex in $\Vp\setminus \Vorphan$ to \emph{not} be sampled is at most $99/100$. Therefore, the probability for no vertices in $\Vp\setminus \Vorphan$ to be sampled is at most
\begin{align*}
(99/100)^{500 \log{n}}\leq \frac{1}{10}\cdot \frac{1}{n^2}.
\end{align*}
We condition on the high-probability event that at least one vertex from $\Vp \setminus \Vorphan$ is sampled for the rest of the proof. We now discuss two sub-cases.
    \begin{enumerate}[label=\alph*).]
    \item \textbf{If there exists a $\Vsmall_{\ell}$ such that $\card{\Vsmall_{\ell}} \geq \frac{99}{100}\cdot \np$ and no vertex from $\cup_{i=1}^{\ell-1} \Vsmall_{i}$ is sampled.} In this case, we show that with high probability, \Cref{eff-split:T}, \Cref{eff-split:out-degree}, and \Cref{eff-split:easy-to-merge} of \Cref{lem:tree-split} always hold, and \cref{case:large-orphan} in \Cref{lem:tree-split} is going to happen. Note that in this case, there exist vertices of $\Vsmall_{\ell}$ that are among $\Vbefore{\np/100}$, and it is of the size at least $\np/100$. As such, the probability for us to sample at least one vertex from $\Vsmall_{\ell}$ is at least
    \begin{align*}
    1-(\frac{1}{100})^{500\log{n}} \geq 1-\frac{1}{10}\cdot \frac{1}{n^5}. 
    \end{align*}
    Let $v \in \Vsmall_{\ell}$ be the sampled vertex. Note that since $\Vh\neq \Vp$, the vertex we sample from $\Vsmall_{\ell}$ is among the vertices of $\cup_{i=1}^{\tilde{\ell}}\Vsmall_{i}$ in \Cref{lem:counterpart-test}. Therefore, by the same argument of the $\Vh= \Vp$ case, if we pick $T^*$ with the largest size, the entire set of $\cup_{i=\ell+1}^{\infty} \Vsmall_{i}$ is going to be included in $T^*$. Therefore, the properties of \Cref{eff-split:T}, \Cref{eff-split:out-degree}, and \Cref{eff-split:easy-to-merge} follow from the same argument of the $\Vh=\Vp$ case. 
    
    Furthermore, by the condition that no vertex from $\cup_{i=1}^{\ell-1} \Vsmall_{i}$ is sampled, we cannot have larger such $T^*$ sets (see~\Cref{lem:counterpart-set-larger}). As such, the set we will pick is necessarily the set $T$ corresponds to $v \in \Vsmall_{\ell}$, and $\Vsmall_{\ell}$ becomes the new $\Vorphan$ set of the next iteration.
    
    Finally, note that we have $\card{\Vsmall_{\ell}} \geq \frac{99}{100}\cdot \np$. And after we remove $T^*$ from $\Vp$, we have $(\Vp \gets \Vp\setminus T^*)$, which means $(\np \gets \np - C)$ for some $C>0$. As such, for the next iteration, we must have $\card{\Vorphan}\geq \frac{99}{100}\cdot \np$. 
    
    \item \textbf{If there exists a $\Vsmall_{\ell}$ such that $\card{\Vsmall_{\ell}} \geq \frac{99}{100}\cdot \np$ and a vertex from $\cup_{i=1}^{\ell-1} \Vsmall_{i}$ is sampled.}
    In this case, we show that \Cref{eff-split:T}, \Cref{eff-split:out-degree}, and \Cref{eff-split:easy-to-merge} of \Cref{lem:tree-split} always hold, and \cref{case:large-T-star} of \Cref{lem:tree-split} is going to happen with high probability. Note that since a vertex $v'\in \cup_{i=1}^{\ell-1} \Vsmall_{i}$ is sampled, and since $v'$ is among $\cup_{i=1}^{\tilde{\ell}}\Vsmall_{i}$ in \Cref{lem:counterpart-test}, the \emph{entire set} of $\Vsmall_{\ell}$ is going to be included in $T^*$. The properties as prescribed by \Cref{eff-split:T}, \Cref{eff-split:out-degree}, and \Cref{eff-split:easy-to-merge} follow from the argument in the $\Vp=\Vh$ case, and the size lower bound becomes 
    \[\card{T^*\cap \Vp}\geq \frac{99}{100}\cdot \np\geq \frac{1}{200}\cdot \np,\]
    as desired. Finally, note that we do not need to guarantee the size upper bound (\Cref{eff-split:size-ub}) since we will not be able to meet the $\Vh=\Vp$ condition.

    \item \textbf{If $\card{\Vsmall_{\ell}} < \frac{99}{100}\cdot \np$ for all $\ell$ among $\Vp$.} In this case, we show that \Cref{eff-split:T}, \Cref{eff-split:out-degree}, and \Cref{eff-split:easy-to-merge} of \Cref{lem:tree-split} always hold, and \cref{case:large-T-star} case of \Cref{lem:tree-split} will happen. We first show the size lower bound of \cref{case:large-T-star}: note that with high probability, we can sample one vertex that is among the first $1/200$ vertices to be split in $\Vp$ in the small-tree split order: the probability for us to \emph{not} sample any vertex among the first $n/200$ vertices in the small-tree split order is at most
    \begin{align*}
    \paren{\frac{199}{200}}^{500\log{n}}\leq \frac{1}{3}\cdot \frac{1}{n^2}.
    \end{align*}
    Therefore, we condition on the event that a vertex $v'$ in $\Vbefore{\np/200}$ is sampled. Since we have the condition that $\card{\Vsmall_{\ell}} < \frac{99}{100}\cdot \np$ for all $\ell$ among $\Vp$, a vertex among $\Vbefore{\np/200}$ can induce at most $\frac{99\np}{100}+\frac{\np}{200}=\frac{199}{200}\cdot \np$ vertices. As such, let $\ell'$ be the level in the small-tree split order of $v'$, we have 
    \[\card{\cup_{i=\ell'+1}^{\infty}\Vsmall_{i}}\geq \frac{1}{200}\cdot \np.\]
    Furthermore, as in the case analysis of $\Vh\neq \Vp$, which means $v'$ is among the vertices of $\cup_{i=1}^{\tilde{\ell}}\Vsmall_{i}$ in \Cref{lem:counterpart-test}. Therefore, the properties of \Cref{eff-split:T}, \Cref{eff-split:out-degree}, and \Cref{eff-split:easy-to-merge} follow from the same argument as in the $\Vh=\Vp$ case.
    
    For the size bounds of $T^*\cap \Vp$, by \Cref{lem:counterpart-set-larger}, if we pick the largest $T^*$ by the subroutine of line Line~\ref{line:test-good-T}, at least the entire set of $\cup_{i=\ell'+1}^{\infty}\Vsmall_{i}$ is going to be included, which means the size of $T^*\cap \Vp$ is of size at least $\frac{1}{200}\cdot \np$. This gives us the size lower bound. 
    
    Finally, we again note that we do \emph{not} need to guarantee the size upper bound (\Cref{eff-split:size-ub}) since we will not be able to meet the $\Vh=\Vp$ condition.
    \end{enumerate}
\end{enumerate}
\end{proof}

We now handle the case when $\Vorphan$ becomes large. We first note that if we happen to sample a vertex $u\in \Vp \setminus \Vorphan$, we can still guarantee $T^* \cap \Vp \neq \emptyset$ and obtain the properties as prescribed by \Cref{lem:tree-split}.

\begin{lemma}
\label{lem:large-orphan-good-sample}
Let $\Vp$ and $\Vh$ be as prescribed by \Cref{lem:tree-split}, and suppose the size of the orphaned set is more than $\frac{99}{100}\cdot \np$, i.e., $\card{\Vorphan}> \frac{99}{100}\cdot \np$. Furthermore, suppose $U \cap (\Vp\setminus \Vorphan) \neq \emptyset$. Then, with high probability, we have $T^* \cap \Vp \neq \emptyset$, and the resulting $T^*$ and $\Vp\setminus T^*$ satisfy the properties of \Cref{lem:tree-split}.
\end{lemma}
\begin{proof}
In the lemma statement, we have already conditioned on a vertex sampled from $\Vp\setminus \Vorphan$. Furthermore, we can again show that the probability for us to sample a vertex among the first $\nh/25$ in the small-tree split order is at least
\[1-(\frac{24}{25})^{500\log{n}}\geq 1-\frac{1}{10}\cdot \frac{1}{n^2}.\]
The events of $U \cap (\Vp\setminus \Vorphan) \neq \emptyset$ and $U\cap \Vbefore{\nh/25} \neq \emptyset$ are \emph{not} independent. Nevertheless, we can still apply a \emph{union bound} and argue that both events happen with high probability. 

Conditioning on the high-probability events as above, we can argue by \Cref{lem:counterpart-test} and \Cref{lem:counterpart-set-larger} that the entire set of $\Vorphan$ is going to be included by the subroutine as defined in line~Line~\ref{line:test-good-T} of \Cref{alg:split}, which gives us the size lower bound. We do \emph{not} need to guarantee the size upper bound since we cannot meet the condition of $\Vh=\Vp$.

Finally, for properties of \Cref{eff-split:T}, \Cref{eff-split:out-degree}, and \Cref{eff-split:easy-to-merge}, we can repeat our proofs in \Cref{lem:small-orphan-case}. We provide the analysis again for the purpose of self-contained proof.
\begin{itemize}
\item For \Cref{eff-split:T}, note that by \Cref{lem:counterpart-test}, the induced set $T^*$ is always $\cup_{i=\ell}^{\infty} \Vsmall_{i}$. Therefore, $T\cap \Vp$ induces a maximal subtree in $\cTstar(\Vp)$.
\item For \Cref{eff-split:out-degree}, note that as long as $T^*$ includes the orphaned set $\Vorphan$, there will be only one edge connecting to the node induces $T^*$ in $\cTstar$. This is exactly what we proved in the lemma.
\item \Cref{eff-split:easy-to-merge} directly follows from \Cref{lem:counterpart-test} since $T^*$ is induced by $u^*$.
\end{itemize}
This concludes the proof of \Cref{lem:large-orphan-good-sample}. 
\end{proof}

By \Cref{lem:large-orphan-good-sample}, the only case of concern now is when $T^* \cap \Vp = \emptyset$, i.e., $\card{\Vorphan}> \frac{99}{100}\cdot \np$ and $U$ does \emph{not} contain samples from $\Vp \setminus \Vorphan$. We now show that our procedure in Line~\ref{line:large-orphan} can effectively distinguish between the cases of $\Vp \setminus \Vorphan=\emptyset$ (root cut) and $\Vp \setminus \Vorphan$ being small.

\begin{lemma}
\label{lem:large-orphan-bad-sample}
Let $\Vp$ and $\Vh$ be as prescribed by \Cref{lem:tree-split}, and suppose the size of the orphaned set is more than $\frac{99}{100}\cdot \np$, i.e., $\card{\Vorphan}> \frac{99}{100}\cdot \np$. Furthermore, suppose $T^* \cap \Vp = \emptyset$ in Line~\ref{line:large-orphan} of \Cref{alg:split}. Then, the following statements are true.
\begin{enumerate}[label=(\roman*).]
\item\label{item:large-orphan-non-root} If $\Vp \setminus \Vorphan \neq \emptyset$, with high probability, we have $X^*\neq \emptyset$ and $X^* \subseteq (\Vp \setminus \Vorphan)$, i.e., $X^*$ only contains vertices in $\Vp$ but not in $\Vorphan$.
\item\label{item:large-orphan-root} If $\Vp \setminus \Vorphan = \emptyset$, with high probability, we have $X^* = \emptyset$.
\end{enumerate}
\end{lemma}
\begin{proof}
Consider the small-tree splitting order of $\Vh$, and let $u\in \Vsmall_{\ell}$ for some $\ell\leq \tilde{\ell}$, where $\Vsmall_{\tilde{\ell}}=\Vorphan$. We prove that with high probability, $i).$ no vertices $v \in \Vp \cap (\cup_{i=\ell}^{\infty} \Vsmall_{\ell})$ can be added to $X$ by Line~\ref{line:pred-test}; and $ii).$ all vertices in $v\in \Vp \cap (\cup_{i=1}^{\ell-1} \Vsmall_{\ell})$ are added to $X$ by Line~\ref{line:pred-test}. (Note that this is why we name the subroutine as a ``predecessor'' test.)

We first observe that by our definition, there is $\Vorphan \subseteq \Vp \cap (\cup_{i=\ell}^{\infty} \Vsmall_{\ell})$. Too see $i)$, note that any $v\in \Vp \cap (\cup_{i=\ell}^{\infty} \Vsmall_{\ell})$ can only split away from $(u,t)$ for $t \in \Vsmall_{\ell}$: this is true since for every $t\not\in \Vsmall_{\ell}$, the lowest common ancestor between $(u,t)$ induces a subtree in $\Vh$ that contains $v$. Moreover, since $\Vsmall_{\ell}\subseteq \Vp$, we have $\card{\Vsmall_{\ell}}\leq \np$. Define $X_{v}$ as the number of answers ``$v$ splits away from $(u,t)$'' for $t\in \Vh$ from $\cO$. By our choice of the parameter \emph{threshold-pred}, we assume w.log. $X_{v}\geq \np$ since otherwise $v$ will not join $X$ anyway. If $\nh\leq 3\np$, we have $\expect{X_{v}}\leq \np+\frac{\nh}{10}=\frac{13}{10}\cdot \np$ in expectation.  Since $X_{v}$ is a summation of independent random variables supported on $\{0,1\}$, we can apply Chernoff bound to obtain that
\begin{align*}
\Pr\paren{X_{v}\geq \frac{3}{2}\cdot \np} &= \Pr\paren{X_{v}\geq \frac{15}{13}\cdot \expect{X_{v}}}\\
&\leq \exp\paren{-\frac{(2/13)^2}{3}\cdot \expect{X_{v}}}\\
&\leq \frac{1}{10}\cdot \frac{1}{n^3}. \tag{using the condition on the lower bound of $X_{v}$}
\end{align*}
On the other hand, if $\nh> 3\np$, we have $\expect{X_{v}}\leq \frac{\nh}{3}+\frac{\nh}{10}=\frac{13}{30}\cdot \nh$. We again assume w.log. that $X_{v}\geq \np$, and we can apply Chernoff bound to obtain that
\begin{align*}
\Pr\paren{X_{v}\geq \frac{2}{3}\cdot \nh} &= \Pr\paren{X_{v}\geq \frac{20}{13}\cdot \expect{X_{v}}}\\
&\leq \exp\paren{-\frac{(8/13)^2}{3}\cdot \expect{X_{v}}}\\
&\leq \frac{1}{10}\cdot \frac{1}{n^3}. \tag{using the condition on the lower bound of $X_{v}$}
\end{align*}
Therefore, we can apply a union bound and argue that with high probability, no vertices $v \in \Vp \cap (\cup_{i=\ell}^{\infty} \Vsmall_{\ell})$ can be added to $X$ by Line~\ref{line:pred-test}, as desired by $i)$.

We now proceed to show $ii).$ all vertices in $v\in \Vp \cap (\cup_{i=1}^{\ell-1} \Vsmall_{\ell})$ are added to $X$ by Line~\ref{line:pred-test}. Note that $v\in \Vp \cap (\cup_{i=1}^{\ell-1} \Vsmall_{\ell})$ implies $v\in \Vp\setminus \Vorphan$. Therefore, $v$ splits from $(u,t)$ for every $t$ in the orphaned set \emph{and} for every $t$ as the sibling of $\Vorphan$, i.e., $t \in \cup_{i=\tilde{\ell}+1}^{\infty} \Vsmall_{i}$. Define $Y_{v}$ as the number of answers ``$v$ splits away from $(u,t)$'' for $t \in \Vh$ from $\cO$. By our choice of the parameter \emph{threshold-pred}, if $\nh\leq 3\np$, since we have $\card{\Vorphan}\geq \frac{99}{100}\np$ and $\card{\cup_{i=\tilde{\ell}+1}^{\infty} \Vsmall_{i}}\geq \card{\Vorphan}$, we have $\expect{Y_{v}}\geq \frac{9}{10}\cdot \frac{199}{100}\cdot \np$ in expectation. Since $Y_{v}$ is a summation of independent random variables supported on $\{0,1\}$, we can apply Chernoff bound to obtain that
\begin{align*}
\Pr\paren{Y_{v}\leq \frac{3}{2}\cdot \np} &= \Pr\paren{Y_{v}\leq \frac{17}{15}\cdot \expect{Y_{v}}}\\
&\leq \exp\paren{-\frac{(2/15)^2}{3}\cdot \expect{Y_{v}}}\\
&\leq \frac{1}{10}\cdot \frac{1}{n^3}. \tag{using $\expect{Y_{v}}\geq \frac{17}{10}\cdot \np$}
\end{align*}
On the other hand, if $\nh> 3\np$, we have at least $\nh-\frac{1}{100}\np \geq \frac{299}{300}\nh$ vertices $t$ such that $v$ splits away from $(u,t)$. Therefore, we have $\expect{Y_{v}}\geq \frac{9}{10}\cdot \frac{299}{300}\cdot \nh\geq \frac{4}{5}\cdot \nh$ in expectation. Therefore, we can again apply Chernoff bound to obtain that
\begin{align*}
\Pr\paren{Y_{v}\leq \frac{2}{3}\cdot \nh} &= \Pr\paren{Y_{v}\geq \frac{5}{6}\cdot \expect{Y_{v}}}\\
&\leq \exp\paren{-\frac{(1/6)^2}{3}\cdot \expect{Y_{v}}}\\
&\leq \frac{1}{10}\cdot \frac{1}{n^3}. \tag{using $\expect{Y_{v}}\geq \frac{4}{5}\cdot \nh\geq \frac{4}{5}\cdot \np$}
\end{align*}
Therefore, we can apply a union bound to obtain the desired statement on $ii)$.

By our statements in $i)$ and $ii)$ as above, we can already conclude that $X^* \subseteq (\Vp \setminus \Vorphan)$. Therefore, \Cref{item:large-orphan-root} of \Cref{lem:large-orphan-bad-sample} follows straightforwardly since if $\Vp \setminus \Vorphan=\emptyset$, any of its subset can only be empty as well. For \Cref{item:large-orphan-non-root}, what remains to show is that with high probability, there is $X^* \neq \emptyset$. Note that if $u\in \Vorphan$ and $\Vp \setminus \Vorphan \neq \emptyset$, then by our statements in $i)$ and $ii)$ above, the set $X^*$ will \emph{not} be empty. Since we assume $\card{\Vorphan}\geq \frac{99}{100}\cdot \card{\Vp}$, the probability for $U'$ to \emph{not} have any vertex $u\in \Vorphan$ is at most
\begin{align*}
(\frac{1}{100})^{100\log{n}}\leq \frac{1}{10}\cdot \frac{1}{n^3},
\end{align*}
as desired. Thus, with high probability, $X^*$ is not empty, which proves \Cref{item:large-orphan-non-root} and concludes the proof of \Cref{lem:large-orphan-bad-sample}.
\end{proof}

\begin{proof}[Finalizing the proof of \Cref{lem:tree-split}]
By \Cref{lem:efficient-split-run-time}, the algorithm runs in $\nh^2 \log{n}$ time. For the set of vertices $\Vp$, we either have $\card{\Vorphan}< \frac{99}{100}\cdot \np$ or $\card{\Vorphan}\geq \frac{99}{100}\cdot \np$. In the former case, we apply \Cref{lem:small-orphan-case}, and all guarantees in \Cref{lem:tree-split} are satisfied. Otherwise, if $\card{\Vorphan}\geq \frac{99}{100}\cdot \np$ \emph{and} $U\cap (\Vp \setminus \Vorphan) \neq \emptyset$, by \Cref{lem:large-orphan-good-sample}, we can still satisfy the properties as prescribed by \Cref{lem:tree-split}. 

The only remaining case is if $\card{\Vorphan}\geq \frac{99}{100}\cdot \np$ \emph{and} $U\cap (\Vp \setminus \Vorphan) = \emptyset$. In such a case, the algorithm will enter Line~\ref{line:large-orphan}. If $\Vp\setminus \Vorphan = \emptyset$, then by \Cref{item:large-orphan-root} of \Cref{lem:large-orphan-bad-sample}, the algorithm uses $\Vh = \Vp$ for a single level of recursion call, and the properties of \Cref{lem:tree-split} are satisfied by the guarantees of \Cref{lem:small-orphan-case} (since now $\Vorphan = \emptyset$). Otherwise, if $\Vp\setminus \Vorphan \neq \emptyset$, note that by \Cref{item:large-orphan-non-root} of \Cref{lem:large-orphan-bad-sample}, any arbitrary vertex $u\in U^*$ belongs to $\Vsmall_{\ell}$ for some $\ell<\tilde{\ell}$, such that $\Vsmall_{\tilde{\ell}}=\Vorphan$. Since $\card{\Vorphan}\geq \frac{99}{100}\np$ and $\nh\geq \np$, the vertex $u$ is among $\Vbefore{\nh/100}$. As such, by \Cref{lem:counterpart-test}, with high probability, the vertex set $T^*$ contains all vertices in $\Vorphan$, and the size is sufficiently large to guarantee \Cref{eff-split:size-lb} of \Cref{lem:tree-split}. We do not need to guarantee \Cref{eff-split:size-ub} since $\Vorphan\neq \emptyset$, and we will not meet the $\Vp=\Vh$ condition. The guarantees of \Cref{eff-split:T} and \Cref{eff-split:out-degree} are similarly satisfied since we remove a set $\cup_{i} \Vsmall_{i}$ that contains $\Vorphan$. Finally, by a similar argument as we used in \Cref{lem:small-orphan-case} and \Cref{lem:large-orphan-good-sample}, \Cref{eff-split:easy-to-merge} is satisfied as desired.
\end{proof}


\subsection{An algorithm to merge two subtrees}
\label{subsec:glue-alg}

We now move to the algorithm that merges two partial trees. Note that this task is \emph{not} trivial: we use the thought process of small-tree split order in the proof of \Cref{lem:tree-split}, but the \emph{algorithm} $\splitalg{\Vp}{\cO}$ does \emph{not} immediately tell us which node did we ``extract'' the set $T$. As such, it still takes considerable work to merge the two trees on the ``right'' internal node.

Exactly here is why we need the split algorithm $\splitalg{\Vp}{\cO}$ to return the vertex $\ustar$. Note that our goal is essentially to find the lowest common ancestor $x$ between $\ustar$ and $T$, and ``stitch'' the tree $\cT_{T}$ to the node. As such, a natural strategy is to ask whether in $\cTstar$ (resp. $\cTstar(\Vp)$), whether a vertex $v\in \Vp$ \emph{splits away} from $(\ustar,x)$ for $x\in T$. If $\cO$ is to answer the queries correctly, all vertices that are ``outside'' $\lcatreeset{T \cup \{\ustar\}}{\cTstar(\Vp)}$ would answer \emph{yes}, and all vertices that are among the leaves of $\lcatreeset{T \cup \{\ustar\}}{\cTstar(\Vp)}$ would answer \emph{no}. We then use the fact that $T$ is always large enough to beat the noise from $\cO$.

The formal description of the algorithm is as \Cref{alg:merge}.

\FloatBarrier
\begin{algorithm}
\caption{$\treemerge{\cT_{T}}{\cT_{\Vp\setminus T}}{\ustar}$: an algorithm to merge partial trees $\cT_{T}$ and $\cT_{\Vp\setminus T}$.} \label{alg:merge}
\KwIn{Vertex set $\Vp$ of size $\np$; a splitting oracle $\cO$; Partial trees $\cT_{T}$ and $\cT_{\Vp\setminus T}$ constructed by $\splitalg{\Vp, \Vh}{\cO}$; vertex $\ustar$ by $\splitalg{\Vp, \Vh}{\cO}$}
\KwOut{A partial tree on $\cT_{\Vp}$}
Initialize $S' \gets \emptyset$\;
\For{$s \in \Vp\setminus T$}{
Initialize a counter $c_s \gets 0$\;
\For{Every vertex $t \in T$}{
If $s$ splits away from $(\ustar, t)$, increase $c_s$ by 1\;
}
\If{$c_s\leq \frac{1}{2}\cdot \card{T}$}{
Add $s$ to $S'$\;
}
}
Take the lowest common ancestor $x=\lcatreeset{S'}{\cT_{\Vp\setminus T}}$\;
If $S'$ does \emph{not} induces a \emph{maximal tree} in $\cT_{\Vp\setminus T}$, i.e., $\leaves{\cT_{\Vp\setminus T}}{S'}\neq S'$, abort the algorithm and report ``fail''\;
If the algorithm does not fail, split node $x$ into two nodes: the left node induces the subtree of $x$, and the right node induces the subtree $\cT_{T}$.
\end{algorithm}

\FloatBarrier

We now present the guarantees of the tree-merging algorithm.

\begin{lemma}
\label{lem:tree-merge}
Given any composable vertex set $\Vp \subseteq V$ such that $\card{\Vp} = \np \geq 50000 \log{n}$, a splitting oracle $\cO$ with correct probability $9/10$, and suppose $\cT_{\Vp\setminus T}$, $\cT_{T}$, and $\ustar$ are obtained by \Cref{alg:split}. Furthermore, assume that
\begin{enumerate}[label=\roman*).]
\item The high probability events of \Cref{lem:tree-split} happens;
\item The partial tree $\cT_{\Vp\setminus T}$ is \emph{(weakly) consistent with} $\cTstar(\Vp\setminus T)$, $\cT_{T}$ is \emph{(weakly) consistent with} $\cTstar(T)$.
\end{enumerate}
Then, with high probability, \Cref{alg:merge} runs in time $O({\np}^2)$, and outputs a partial tree $\cT_{\Vp}$ that is weakly consistent with $\cTstar(\Vp)$.
\end{lemma}
\begin{proof}
We first remind the readers of the definition of a partial tree $\pTree$ \emph{weakly consistent} with another tree $\cT$. For $\pTree$ to be consistent with $\cT$, there should be
\begin{enumerate}[label=\alph*).]
\item each super-vertex of $\pTree$ corresponding to a connected subtree in $\cT$ with out-degree at most $2$, and each of the edge connects to either a parent or a sibling node; and 
\item for leaves $(x,y)$ in $\pTree$, let $X$ and $Y$ be the corresponding leaves in $\cT$, the subtrees induced by $\lcatree{x}{y}{\cT}$ and $\lcatreeset{X\cup Y}{\cT}$ contain exactly the same set of leaves. 
\end{enumerate}

We now show that with high probability, the algorithm $\treemerge{\cT_{T}}{\cT_{\Vp\setminus T}}{\ustar}$ outputs a partial tree that satisfies $a)$ and $b)$ w.r.t. $\cTstar(\Vp)$. For $a)$, we note that the super-vertices in $\cTstar(\Vp)$ and $(\cTstar(\Vp\setminus T), \cTstar(T))$ are exactly the same. Furthermore, by the high probability event of \Cref{lem:tree-split}, both $\cTstar(\Vp)$ and $(\cTstar(\Vp\setminus T), \cTstar(T))$ satisfy the \emph{weakly consistent} property. As such, the guarantee of $a)$ follows.

The main work here is to prove the guarantee prescribed by $b)$. To this end, we first observe that for any set of vertices $X$ and $Y$, if $\lcatreeset{X\cup Y}{\cTstar(\Vp)}$ only contain vertices in $T$ (resp. $\Vp\setminus T$), then the assumptions of $\cT_{\Vp\setminus T}$ being consistent with $\cTstar(\Vp\setminus T)$ and $\cT_{T}$ being consistent with $\cTstar(T)$ is sufficient to prove the leaves induced by $\lcatreeset{X\cup Y}{\cT_{\Vp}}$ is the same as the leaves of $\lcatreeset{X\cup Y}{\cTstar(\Vp)}$. This is evident by using the procedure that constructs $\cTstar(\Vp)$ as in \Cref{def:restrict-tree}. 

The final missing piece is the vertex sets $X$ and $Y$ that induce vertices in both $T$ and $\Vp\setminus T$, which is the place where we need to show that the merging algorithm finds the ``correct'' node to merge. Let $S=\Vorphan$ be the orphaned vertices by removing $T$ from $\Vp$. Since we condition on the high probability event of \Cref{lem:tree-split}, there must be an internal vertex $z$, such that $z=\lcatreeset{T\cup\{\ustar\}}{\cTstar(\Vp)}$, and nodes $r_{T}$ and $r_{S}$ such that $i).$ $z=\pa{r_{T}}$ and $z=\pa{r_{S}}$ in $\cTstar(\Vp)$ and $ii).$ the induced leaves of $r_{T}$ is $T$ and the induced leaves of $r_{S}$ is $S$ such that $\ustar\in S$. Furthermore, we have $\card{T}\geq 50000\log{n}\cdot \frac{1}{200} \geq 200\cdot \log{n}$ by \Cref{lem:tree-split}. We now claim that by running $\treemerge{\cT_{T}}{\cT_{\Vp\setminus T}}{\ustar}$, the set $S'$ we recover is exactly the leaves of $S$ (the $\Vorphan$ set of vertices). The detailed analysis is as follows.
\begin{itemize}
\item For each $s\in S$, observe that in $\cTstar$ (and $\cTstar(\Vp)$), $s$ does \emph{not} split away from $(\ustar, t)$ for $t\in T$. Therefore, define $C_{s}$ as the random variable that records ``$s$ split away from $(\ustar, t)$ for $t\in T$'' from $\cO$, we have in expectation $\expect{C_{s}}\leq \frac{1}{10}\cdot \card{T}$. If $C_{s}\leq 20\cdot \log{n}$, it trivially fails the test. Otherwise, if $C_{s}\geq 20\cdot \log{n}$, we can apply Chernoff bound to get that
\begin{align*}
\Pr\paren{C_{s}\geq \frac{1}{2}\cdot \card{T}} &= \Pr\paren{C_{s} \geq 5\cdot \expect{C_{s}}}\\
&\leq \exp\paren{-\frac{5^2}{3}\cdot \expect{C_{s}}} \\
& \leq \frac{1}{5}\cdot \frac{1}{n^3}. \tag{using $C_{s}\geq 20\cdot \log{n}$}
\end{align*}
\item For each $d \not\in S$, observe that in $\cTstar$ (and $\cT(\Vp)$), $d$ \emph{does} split away from $(\ustar, t)$ for $t\in T$. Therefore, define $C_{d}$ as the random variable that records ``$d$ split away from $(\ustar, t)$ for $t\in T$'' from $\cO$, we have in expectation $\expect{C_{d}}\geq \frac{9}{10}\cdot \card{T}$. Therefore, we can apply Chernoff bound to get that
\begin{align*}
\Pr\paren{C_{d}\leq \frac{1}{2}\cdot \card{T}} &= \Pr\paren{C_{d} \leq \frac{5}{9}\cdot \expect{C_{d}}}\\
&\leq \exp\paren{-\frac{(4/9)^2}{3}\cdot \expect{C_{d}}} \\
& \leq \frac{1}{5}\cdot \frac{1}{n^3}. \tag{using $C_{d}\geq 200\cdot \log{n}$}
\end{align*}
\end{itemize}
We can then apply a union bound over at most $\np\leq n$ vertices in $\Vp\setminus T$ to get the desired statement.

Observe that any internal node that induces leaves in both $T$ and $\Vp\setminus T$ has to at least include the whole set of $S$ and $T$. For any leaves $x$ and $y$ in $T_{\Vp}$, let the induced set of vertices in the leaves of $\cTstar(\Vp)$ be $Z=S\cup T\cup P$. By the above argument, $S\cup T$ should be returned in the set of vertices. Furthermore, by the consistency between $\cT_{\Vp\setminus T}$ and $\cTstar(\Vp\setminus T)$, the set $P$ should also be returned. This concludes the proof. 
\end{proof}

\subsection{Finalizing the proof of \texorpdfstring{\Cref{thm:weak-partial-tree}}{thm:weak-partial-tree}}
We now move to prove \Cref{thm:weak-partial-tree} for our partial tree construction. We first bound the number of possible internal nodes as $n$ by \Cref{fct:num-splits-tree}. Therefore, we can apply a union bound on \emph{all} splits for \Cref{lem:tree-split}, and argue that with high probability, the events of \Cref{lem:tree-split} hold for every partition. Furthermore, conditioning on \Cref{lem:tree-split} always holds across the splits, we can again apply a union bound to show that \Cref{lem:tree-merge} holds across all internal nodes in the partial tree construction with high probability. We condition on the high probability events of \Cref{lem:tree-split} and \Cref{lem:tree-merge} across the internal nodes for the rest of the proof.

\paragraph{Proof of efficiency.} We now show that the depth of recursive calls on \Cref{lem:tree-split} is $O(\log^3{n})$, i.e., the longest sequence of recursive calls induced by any fixed $\Vp$ is at most $O(\log^3{n})$ in \Cref{alg:partial-tree}. To see this, consider any set of vertices $\Vp$ with size $\np$, and we look into \emph{three} levels of splits on $\Vp$. Suppose the vertices sets are $\Vp\rightarrow (\Vp_{1}, \Vp_{2}, \Vp_{3}, \Vp_{4}, \Vp_{5}, \Vp_{6}, \Vp_{7}, \Vp_{8})$. We claim that there is 
\[\max_{i}\paren{\{\card{\Vp_{i}}\}_{i=1}^{8}}\leq (1-\frac{1}{10000\log^{2}{\np}})\cdot \np.\] 
We prove the above statement by using \Cref{lem:tree-split}. Conditioning on the high probability event of \Cref{lem:tree-split}, if the first split of $\Vp$ falls into the condition of $\Vp=\Vh$, then by \Cref{eff-split:size-lb} and \Cref{eff-split:size-ub} of \Cref{lem:tree-split}, the balanceness of sizes already follows. Otherwise, we can have the following cases
\begin{itemize}
\item If we enter the case of \cref{case:large-T-star}, note that we have $\card{T}\geq \frac{\np}{200}$. As such, in the first split $\Vp \rightarrow (\Vp\setminus T, T)$, we already have $\card{\Vp\setminus T}\leq \frac{199}{200}\cdot \np\leq (1-\frac{1}{10000\log^{2}{\np}})\cdot \np$. The size of $T$ might be large; however, it must have $\Vp=\Vh$ for the first split on $T$ (by \Cref{alg:partial-tree}). Therefore, in the next iteration, the maximum size is at most $(1-\frac{1}{10000\log^{2}{\np}})\cdot \np$ as well. 
\item If we enter the case of \cref{case:large-orphan}, note that the set $T$ of the current iteration is of size at most $\frac{\np}{100}\leq (1-\frac{1}{10000\log^{2}{\np}})\cdot \np$. Furthermore, in the next iteration, the $\Vorphan$ set accounts for at least $\frac{99}{100}$ fraction of vertices. Hence, we can apply \Cref{lem:large-orphan-good-sample,lem:large-orphan-bad-sample} and argue that the split will be in the case of \cref{case:large-T-star} with $T_2$ as the new set $T^*$. Now, we have $\Vp\setminus (T\cup T_2)$ with size at most $\frac{\np}{100}\leq (1-\frac{1}{10000\log^{2}{\np}})\cdot \np$. The new set $T_2 \cap (\Vp\setminus T)$ might be large, but since it forms a \emph{single} maximal subtree in $\cT(T_2 \cap (\Vp\setminus T))$, there is $\Vh=\Vp$, and in the \emph{third} iteration, the maximum size is going be to at most $(1-\frac{1}{10000\log^{2}{\np}})\cdot \np$, as desired.
\end{itemize}

Since the size reduces by a $(1-\frac{1}{10000\log^{2}{\np}})$ factor for every three level of splits, after $60000\log^3{\np}$ recursive calls, we have
\begin{align*}
\text{remaining size}& \leq \np\cdot (1-\frac{1}{10000\log^{2}{\np}})^{20000\log^{3}{\np}}\\
&\leq \np\cdot \exp\paren{-2\log{\np}}\leq O(1),
\end{align*}
to which point the remaining vertices will be collapsed to a super-vertex by our algorithm. Therefore, since $\np\leq n$, the longest sequence of dependent calls is at most $O(\log^3{n})$.

Finally, to complete the proof of efficiency, note that by \Cref{lem:tree-split}, the runtime for each call of \Cref{alg:split} is $O(\nh^2 \cdot \log{n})$. The tree has depth at most $\log^3{n}$; and at each level, the total number of runtime is at most $O(n^2\cdot \log{n})$ since we have $\sum \nh \leq n$. Similarly, each call of the merging algorithm will happen only after the split algorithm, which causes an overall $O(n^2)$ runtime overhead on any level. Therefore, the total runtime is bounded by $O(n^2\cdot \log^{4}{n})=O(n^2 \cdot \polylog{n})$, as desired.

\paragraph{Proof of correctness.} We inductively prove the correctness of \Cref{thm:weak-partial-tree}. On the level of the leaves in \Cref{alg:partial-tree}, by \Cref{lem:tree-split}, if the leave contains more than one vertex, it must be a composable set with out-degree at most $2$ in $\cTstar$, and exactly one of them connecting to a parent node, and the other connecting to the sibling of the orphaned vertex. As such, when we merge two components $X$, $Y$ in which at least one of them is a super-vertex, we can guarantee the out-degree is still at most $2$, and the $\texttt{LCA}$ of $X$ and $Y$ induces the same vertices as on $\cTstar(\lcatreeset{X\cup Y}{\cTstar})$, which implies the partial tree is weakly consistent with $\cTstar(X\cup Y)$. On the other hand, when merging two components who are both \emph{not} super-vertices, we can use \Cref{lem:tree-merge}, and the assumptions of weak consistency come from the guarantees on previous partitions. Therefore, the weak consistency inductively applies to every level of the merging process, which gives the desired correctness guarantee. 

\section*{Acknowledgements}
We thank anonymous ICML reviewers for the helpful comments and suggestions. 
Vladimir Braverman is supported in part by the Naval Research (ONR) grant N00014-23-1-2737 and NSF CNS-2333887 award. 
Samson Zhou is supported in part by NSF CCF-2335411. 
The work was conducted in part while Samson Zhou was visiting the Simons Institute for the Theory of Computing as part of the Sublinear Algorithms program.

\bibliographystyle{alpha}
\bibliography{reference}

\newpage
\appendix

\crefalias{section}{appendix}

\section{Additional Technical Preliminaries}
\label{app:info-theoretic-facts}

\subsection{Concentration inequalities}
\label{sub-app:concentration-inequ}
We now present the standard concentration inequalities used in our proofs. We start from the following standard variant of Chernoff-Hoeffding bound. 

\begin{proposition}[Chernoff-Hoeffding bound]\label{prop:chernoff}
	Let $X_1,\ldots,X_n$ be $n$ independent random variables with support in $[0,1]$. Define $X := \sum_{i=1}^{n} X_i$. Then, for every $\delta \in (0,1]$, there is 
	\begin{align*}
		\Pr\paren{\card{X - \expect{X}} > \delta\cdot \expect{X}} \leq 2 \cdot \exp\paren{-\frac{\delta^{2} \expect{X}}{3}}. 
	\end{align*}
 Furthermore, for every $\delta>0$, there is
     \begin{align*}
		\Pr\paren{\card{X - \expect{X}} > \delta\cdot \expect{X}} \leq 2 \cdot \exp\paren{-\frac{\delta^{2} \expect{X}}{2+\delta}}. 
	\end{align*}
\end{proposition}

\subsection{Standard results for trees}

\begin{fact}
\label{fct:num-splits-tree}
Any binary tree $\cT$ with $n$ leaves contains at most $n$ internal nodes.
\end{fact}
\begin{proof}
Consider the collection of the internal nodes that are the parents of the leaves. Since the tree is binary, there are \emph{at most} $n/2$ such nodes. Contract all these internal nodes with the leaves, and we can obtain a new tree $\cT'$ such that the number of leaves is $n/2$. As such, we can again count an upper bound for the number of the parents for the leaves in $\cT'$ as $n/4$. We can continue this until we only have the root, and the number of internal nodes is at most 
\begin{align*}
\sum_{i=1}^{n} \frac{n}{2^i} \leq \sum_{i=\infty}^{n} \frac{n}{2^i} = n,
\end{align*}
as desired.
\end{proof}

\subsection{The PRAM and the Massively Parallel Computation (MPC) Models}
\label{sec:pram-mpc-models}
We briefly introduce the parallel models we investigate for the Moseley-Wang objective. In particular, we investigated the classical PRAM model and the Massively Parallel Computation model in our work.

\paragraph{The PRAM model.} The Parallel Random Access Machine (PRAM) model is a widely used theoretical framework in parallel computing. It provides a simplified abstraction of a parallel computer system, where multiple processors work simultaneously to solve a computational problem. In the PRAM model, each processor has direct access to a common memory space (RAM), and communication between processors and the RAM is instantaneous (``parallel'' RAM). Processors can read from and write to any memory location in parallel, hence the term ``Random Access''.

In the PRAM model, there are usually two objectives for algorithm designers to optimize: the total \emph{work}, defined as the total number of elementary operations, and the \emph{depth}, defined as the length of the longest dependent call in the algorithm. In the theoretical abstract version of PRAM, we do \emph{not} care about the number of processors we use in the algorithm.

\paragraph{The MPC model.} The Massively Parallel Computation (MPC) model is a theoretical framework used to analyze algorithms designed for modern parallel computing architectures, e.g., the MapReduce framework. In this model, communications are conducted in \emph{synchronized rounds}, and a machine can communicate with any other machine in a round. Furthermore, each machine can do unlimited local computation between the rounds. Unlike traditional CONGEST models, the communication here is limited only by the \emph{memory size} -- a size $s$ machine cannot send or receive more than $s$ bits of information.

For graph problems, suppose each machine has size $s$, we typically use $\tilde{O}(n^2/s)$ machines, where $O(n^2)$ is the worst-case input size (alternatively, one can also target the more instance-optimal $\tilde{O}(m/s)$ machines). 

Our goal in the MPC model is to minimize two objectives: $i).$ the memory size $s$ of each machine; and $ii).$ the number of parallel rounds. In particular, if our algorithm works with $s=O(n^{\delta})$ memory for any $\delta\in (0,1)$, we call the algorithm \emph{fully scalable} in the MPC model. Typically, the best MPC algorithms would ask for fully scalable memory and $\polylog{n}$ rounds.

\paragraph{A reduction between the PRAM and the MPC algorithms.} The PRAM and the MPC models share a great deal of similarities. And indeed, the following reduction is known.

\begin{proposition}
\label{prop:PRAM-to-MPC}
Suppose there exists a PRAM algorithm that computes a function $f$ with $w(n)$ work and $d(n)$ depth, where $n$ is the input size of $f$. Then, there exists a fully scalable MPC algorithm that computes $f$ with $O(w(n))$ total memory and $O(d)$ rounds. The memory per machine can be made $O(n^\delta)$ for any $\delta\in (0,1)$, and the number of machines is $O(\frac{W(n)}{n^{\delta}}\cdot \polylog{n})$.
\end{proposition}

\section{Discussions about Additional Settings for Our Algorithms}
\label{app:additional-discussion}
We discuss our algorithms in additional settings, which include general success probability (other than $9/10$) and splitting oracle for \emph{approximately optimal} HC trees.

\subsection{General success probabilities}
\label{subsec:general-success-prob}
We use a success probability of ${9}/{10}$ in our algorithms for technical convenience. Here, we discuss algorithms with more general success probabilities. We remark that due to \emph{adversarial} incorrect answers, our algorithm cannot work with $\frac{1}{2}+\eps$ success probability for arbitrary $\eps$. In fact, it is unclear whether \emph{any} algorithm would work with $\frac{1}{2}+\eps$ success probability and \emph{adversarial} incorrect answers. Concretely, suppose that in the optimal HC tree, the first cut is balanced with size ($n/2$, $n/2$). This appears to be a quite easy example. Now, let us fix a vertex $u$ and determine whether a vertex $v$ is on the same side of $u$. If $v$ is on the same side of $u$, there are $n/2$ vertices $w\in V$ such that $w$ splits away from $(u, v)$; conversely, if $v$ is on the opposite side, there is no such $w$ vertex. However, due to \emph{adversarial} incorrect answers, we can report $(n/2-n/3 = n/6)$ such $w$ vertices in the former case, and $n/3$ such w vertices in the latter case. As such, in the above example, the correct probability for the oracle must be at least $3/4$ to get anything meaningful.
Finally, we remark our algorithm would work for any success probability $\frac{1}{2}+C$ for sufficiently large $C=\Omega(1)$: all the analysis will go through with changes in the constants. Furthermore, if we deal with \emph{random} incorrect answers instead, we will be able to work with $\frac{1}{2}+\eps$ success probability for any $\eps=\Omega(1/n)$.

Finally, we give a remark on the discrepancy between success probabilities between learning-augmented HC and other learning-augmented graph algorithms.
In some graph algorithm, e.g., in \cite{BravermanDSW24,Cohen-Addadd0LP24,Dong0V25}, the learning-augmented oracle could work with $1/2+C$ probability for any $C=\Omega(1)$.
The reason our algorithm should work with a sufficiently high success probability is due to the hierarchical structure. For instance, in the paper that studied learning-augmented max-cut (\cite{Dong0V25}), the errors in the algorithm are ``one shot''. However, in the HC problem, if the construction of the partial tree is wrong at any level, the error will propagate to all subsequent nodes, and it is not clear how to control the error if this happens. Therefore, a constant success probability sufficiently larger than $1/2$ is necessary.

\subsection{Splitting oracle with approximately optimal HC trees}
\label{subsec:split-oracle-approx-tree}
A natural extension of our algorithms is to explore HC algorithms with splitting oracles from an \emph{approximately} optimal HC tree. In other words, for a triplet of vertices $(u,v,w)$, the oracle $\cO$ answers the ``splitting away'' query based on an HC tree $\cT$ that achieves $\alpha$-approximation of the optimal tree $\cTstar$. We remark that our algorithms based on the strongly consistent partial HC trees (i.e., the algorithms of \Cref{thm:constant-poly-time-HC,thm:root-loglog-n4-time-HC,thm:hc-das-streaming}) could work with approximation HC trees. In particular, if the splitting oracle is constructed from an $\alpha$-approximation HC tree $\cT$, our algorithm will produce HC trees with an extra $O(\alpha)$ factor on the approximation guarantees.

On the other hand, however, it is not immediately clear whether our algorithms based on weakly consistent partial HC trees could work for oracles from approximate HC trees.
The main difficulty here is that to analyze the revenue decrement induced by the weakly consistent partial HC tree, we need to prove \Cref{lem:MW-opt-structure} that characterizes the revenue structure of the optimal tree. We proved the statement by showing that if the statement is not true, we can increase the revenue, which forms a contradiction with the optimal HC tree (see~\Cref{clm:weight-property-internal-nodes} for details). We cannot easily argue that the same structural statement with an approximately optimal HC tree. This could be an interesting problem to resolve for future work.

\section{Splitting Oracle and Learning Theory}
\label{sec:oracle-and-learning-theory}
In this section, we offer formal learning theory analysis for learning a splitting oracle for the learning-augmented hierarchical clustering problem. 
In particular, we utilize the PAC learning framework to show that a high-quality predictor can be efficiently learned, provided that the input instances are drawn from a specific distribution. 
We remark that similar results have been shown in other settings by~\cite{IzzoSZ21,ChenSVZ22,ErgunFSWZ22,GrigorescuLSSZ22}. 
Thus although the results of this section are by now standard techniques, they still provide an end-to-end framework for designing learning-augmented algorithms. 

First, we suppose that there exists an underlying distribution $\calD$, from which our input is drawn. 
In particular, $\calD$ generates independent instances for hierarchical clustering, corresponding to the setting where similar instances of hierarchical clustering are being solved. 
Note that this is exactly the setting where we would like to apply learning-augmented algorithms. 
If there is instead generalization failure or distribution-shift, then inherently machine learning models will perform poorly. 

Then our goal is to efficiently learn a predictor $f$ from a family $\calF$ of possible functions, where the input to each predictor $f$ is a weighted undirected graph $G=(V,E,w)$ and three specific nodes, and the output is a feature vector. 
We remark that each input instance $G$ can be encoded as a vector in $\mathbb{R}^{n^2+n}$, by first considering the weighted $n\times n$ adjacency matrix of the graph. 
We can then flatten the matrix into a vector of dimension $n^2$ and then append a $3$-sparse binary vector of length $n$, corresponding to the three vertices in the input. 
We also assume that the output of $f$ has at most $n$ dimension, indicating a binary vector for which vertex should be split from the other two vertices. 

We define a loss function $L:f\times G\to\mathbb{R}$, which intuitively defines how accurate a predictor $f$ performs on each input instance $G$. 
For example, $f$ can represent the splitting oracle on $G$ and $L$ can denote the number of inaccurate responses compared to the best hierarchical clustering on $G$. 

Now our goal is to learn the function $f \in \mathcal{F}$, which minimizes the following objective:
\begin{equation}
\label{eq:loss-f}
\EEx{G\sim\calD,(x,y,z)\in V^3}{L(f_G(x,y,z))}.
\end{equation}
Let $f^*$ be an optimal function in $\calF$,, so that $f^*=\argmin\EEx{G\sim\calD,(x,y,z)\in V^3}{L(f_G(x,y,z))}$ is a minimizer of the above objective. 
Assuming that for each graph instance $G$, triplet $(x,y,z)\in V$, and each $f\in\calF$, we can efficiently compute $f_G(x,y,z)$ as well as $L(f_G(x,y,z))$, say in polynomial time $T(n)$, then we have:
\begin{proposition}
\label{thm:learn:oracle}
    There exists an algorithm that uses $\poly\left(T(n),\frac{1}{\eps}\right)$ samples and returns a function $\hat{f}$, such that with probability at least $\frac{9}{10}$,
\[ \EEx{G\sim\calD,(x,y,z)\in V^3}{L(\hat{f}_G(x,y,z))}\le\min_f\EEx{G\sim\calD,(x,y,z)\in V^3}{L(f_G(x,y,z))} + \eps.\]
\end{proposition}
In particular, \Cref{thm:learn:oracle} is a PAC-style result that bounds the number of samples necessary to achieve a good probability of learning an approximately-optimal function $\hat{f}$. 
The algorithm corresponding to \Cref{thm:learn:oracle} is straightforward; it is simply the empirical minimizer after a sufficient number of samples are drawn. 
To prove correctness, we first require the following definition of pseudo-dimension for a function class, which is a generalization of VC dimension to real-valued functions. 
\begin{definition}[Pseudo-dimension, e.g., Definition $9$ in \cite{LucicFKF17}]
Let $\calX$ denote a ground set, and let $\calF$ be a collection of functions mapping elements from $\calX$ to the interval $[0,1]$. 
Consider a fixed set $S = \{x_1, \dots, x_n\} \subset \calX$, a set of real numbers $R = \{r_1, \dots, r_n\}$, where each $r_i \in [0,1]$, and a function $f \in \calF$. 
The subset $S_f = \{x_i \in S \mid f(x_i) \ge r_i\}$ is referred to as the induced subset of $S$ determined by the function $f$ and the real values $R$. 
We say that the set $S$ with associated values $R$ is shattered by $\calF$ if the number of distinct induced subsets is $|\{S_f \mid f \in \calF\}| = 2^n$. 
Then the \emph{pseudo-dimension} of $\calF$ is defined as the size of the largest subset of $\calX$ that can be shattered by $\calF$ (or it is infinite if no such maximum exists). 
\end{definition}
Using pseudo-dimension, we can now present an accuracy-sample complexity trade-off for empirical risk minimization with and the number of necessary samples. 
First, we define $\calH$ be the class of functions in $\calF$ composed with $L$, i.e., $\calH:= \{L \circ f : f \in \calF\}$. 
Moreover, by normalization, we can assume the range of $L$ is contained within $[0,1]$. 
Then we have the following generalization bounds:
\begin{theorem}
\cite{anthony_bartlett_1999}
\label{thm:uni-dim}
Let $\calD$ be a distribution over problem instances in $\calG$, and let $\calH$ be a class of functions $h:\calG\to[0,1]$ with pseudo-dimension $d_{\calG}$. 
Consider $t$ i.i.d.\ samples $G_1, G_2, \dots, G_t$ drawn from $\calD$. 
There exists a universal constant $c_0$ such that for any $\eps > 0$, if $t \ge c_0 \cdot \frac{d_{\calH}}{\eps^2}$, then for all $h\in\calH$, we have the following with probability at least $\frac{9}{10}$:
\[
\left| \frac{1}{t} \sum_{i=1}^{t} h(G_i) - \EEx{G \sim\calD} h(G) \right| \leq \eps.
\]
\end{theorem}
We have the following immediate corollary by applying the triangle inequality.
\begin{corollary}
\label{cor:uni-dim-cor}
Let $G_1,\ldots,G_t$ be a set of independent samples from $\calD$ and let $\hat{h}\in\calH$ be a function that minimizes $\frac{1}{t}\sum_{i=1}^t h(G_i)$. 
If the number of samples $t$ is chosen as in \Cref{thm:uni-dim}, then with probability at least $\frac{9}{10}$,
\[\EEx{G\sim\calD}{\hat{h}(G)}\le\min\EEx{G\sim\calD}{h^*(G)} + 2 \eps.\]
\end{corollary}
Thus, the main question is to analyze the pseudo-dimension of our function class $\calH$. 
To that end, we first relate the pseudo-dimension to the VC dimension of a related class of threshold functions. 
\begin{lemma}[Pseudo-dimension to VC dimension, Lemma $10$ in \cite{LucicFKF17}]
\label{lem:PD_to_VC}
For any $h\in\calH$, let $B_h$ denote the indicator function of the threshold function, i.e., $B_h(x,y)  =  \text{sgn}(h(x)-y)$. 
Then the pseudo-dimension of $\calH$ equals the VC-dimension of the subgraph class $B_{\calH}=\{B_h \mid h \in \calH\}$.
\end{lemma}
What remains is to bound the VC dimension of the function to compute in the class, which follows from the following standard result.

\begin{lemma}[Theorem $8.14$ in \cite{anthony_bartlett_1999}]
\label{lem:VC_bound}
Let $\tau: \mathbb{R}^a \times \mathbb{R}^b \to \{0,1\}$, defining the class
\[
\calT = \{ x \mapsto \tau(\theta, x) : \theta \in \mathbb{R}^a \}.
\]
Assume that any function $\tau$ can be computed by an algorithm that takes as input the pair $(\theta, x) \in \mathbb{R}^a \times \mathbb{R}^b$ and produces the value $\tau(\theta, x)$ after performing no more than $t$ of the following operations:
\begin{itemize}
    \item arithmetic operations $+, -, \times, /$ on real numbers,
    \item comparisons involving $>, \ge, <, \le, =,$ and outputting the result of such comparisons,
    \item outputting $0$ or $1$.
\end{itemize}
Then, the VC dimension of $\calT$ is bounded by $O(a^2 t^2 + t^2 a \log a)$.
\end{lemma}
We can now apply these results to prove \Cref{thm:learn:oracle} by instantiating Lemma \ref{lem:VC_bound} with the computational complexity of evaluating any function in the class $\calH$. 
\begin{proof}[Proof of \Cref{thm:learn:oracle}]
From Lemma \ref{lem:PD_to_VC}, we know that the pseudo-dimension of $\calH$ is equivalent to the VC dimension of the threshold functions defined by $\calH$. 
Next, from Lemma \ref{lem:VC_bound}, we observe that the VC dimension of the relevant class of threshold functions is polynomial in the computational complexity of evaluating a function from the class.
In other words, Lemma \ref{lem:VC_bound} implies that the VC dimension of $B_{\calH}$ (as defined in Lemma \ref{lem:PD_to_VC}) is polynomial in the number of arithmetic operations required to compute the threshold function corresponding to some $h\in\calH$.
According to our definition, this quantity is polynomial in $T(n)$. 
Thus, the pseudo-dimension of $\calH$ is also polynomial in $T(n)$, and the desired result follows.
\end{proof}
Note that we can initialize \Cref{thm:learn:oracle} with various oracles, in terms of the input and output predictions. 
Indeed, if each function in the family of oracles we are considering can be computed efficiently, then \Cref{thm:learn:oracle} guarantees that a polynomial number of samples is sufficient to learn a nearly optimal oracle.

\end{document}